\documentclass[twoside]{article}

\usepackage[accepted]{aistats2024}
\usepackage{amsmath,amssymb,trimclip,adjustbox}

\usepackage{algorithm}
\usepackage{algorithmic}

\usepackage{hyperref}
\usepackage{url}

\usepackage[english]{babel}
\usepackage{amsthm}

\usepackage{graphicx}
\usepackage{subcaption}

\theoremstyle{plain}

\theoremstyle{definition}
\newtheorem{definition}{Definition}[section]

\newtheorem{theorem}{Theorem}[section]

\newtheorem{lemma}[theorem]{Lemma}
\newtheorem*{remark}{Remark}

\newtheorem{exmp}{Example}[section]

\usepackage{textcomp}
\usepackage{xcolor}
\usepackage{array}

\newtheorem{assumption}{Assumption}

\newcommand{\E}{\mathbb{E}}
\newcommand{\Var}{\mathrm{Var}}
\newcommand{\Bias}{\mathrm{Bias}}
\newcommand{\AMISE}{\mathrm{AMISE}}
\newcommand{\Cov}{\mathrm{Cov}}
\newcommand{\Vol}{\mathrm{Vol}}
\newcommand{\bs}{\boldsymbol}

\DeclareMathOperator*{\argmax}{argmax} 
\DeclareMathOperator*{\argmin}{argmin}

\usepackage{mathtools}
\usepackage{xfrac}

\usepackage{caption}
\usepackage{comment}

\usepackage{amsfonts, 
            empheq}%

\usepackage[nottoc,notlof,notlot]{tocbibind} 

\usepackage[round]{natbib}

\begin{document}

%

%

\twocolumn[

\aistatstitle{Formal Verification of Unknown Stochastic Systems via Non-parametric Estimation}

\aistatsauthor{ Zhi Zhang \And Chenyu Ma \And Saleh Soudijani 
\And  Sadegh Soudjani  }
\aistatsaddress{ Newcastle University\\
United Kingdom \And Newcastle University\\
United Kingdom \And CISPA\\
Germany \And MPI-SWS\\
Germany} 
]

\begin{abstract}
A novel data-driven method for formal verification is proposed to study complex systems operating in safety-critical domains. The proposed approach is able to formally verify discrete-time stochastic dynamical systems against temporal logic specifications only using observation samples and without the knowledge of the model, and provide a probabilistic guarantee on the satisfaction of the specification. We first propose the theoretical results for using non-parametric estimation to estimate an asymptotic upper bound for the \emph{Lipschitz constant} of the stochastic system, which can determine a finite abstraction of the system. Our results prove that the asymptotic convergence rate of the estimation is $O(n^{-\frac{1}{3+d}})$, where $d$ is the dimension of the system and $n$ is the data scale.
We then construct interval Markov decision processes using two different data-driven methods, namely non-parametric estimation and empirical estimation of transition probabilities, to perform formal verification against a given temporal logic specification.
Multiple case studies are presented to validate the effectiveness of the proposed methods.
\end{abstract}

\section{INTRODUCTION}\label{sec1}

For safety-critical systems, formal verification plays an essential role in analysing the system and providing formal safety guarantees \citep{baier2008principles}, which promotes the development of autonomous systems such as self-driving cars, power grids, medical robotics, and unmanned aircraft. Formal verification requires the knowledge of a model of the system under study \citep{baier2008principles,clarke1994model}.
Complex systems interacting with unpredictable environments are challenging to model \citep{yeh2018autonomous,corso2021survey}. The availability of large amounts of data from such systems necessitates developing data-driven formal verification techniques with a weak dependence on the information of the system's model.

Formal verification focuses on checking whether a system satisfies a given specification described using temporal logic \citep{belta2017formal,doyen2018verification,tabuada2009verification}.
Formal verification of stochastic systems has been studied by abstracting the system with a continuous state space into a finite Markov decision process (MDP) \citep{baier2008principles,clarke1994model}.
For systems with an unknown model, constructing the MDP is not available, due to the difficulty in (i) guaranteeing the closeness between the specifications of the original unknown system and its finite abstraction, and in (ii) the construction of transition probabilities between states of the finite abstraction. 
The first difficulty can be addressed by introducing the \emph{Lipschitz constant} (LC) of the transition kernel of the unknown system, which is the underlying assumption in many abstraction-based verification and synthesis approaches \citep{SA13,Lavaei_Survey}.
The second difficulty can be addressed by developing data-driven verification methods that would allow verification to be carried out solely based on data, with no prior knowledge of the system's model. In this paper, we will use non-parametric estimation and an empirical approach to develop such data-driven verification methods.

\begin{figure}[hbt!]
\vspace{-1cm}
\centerline{\includegraphics[width=0.4\textwidth]{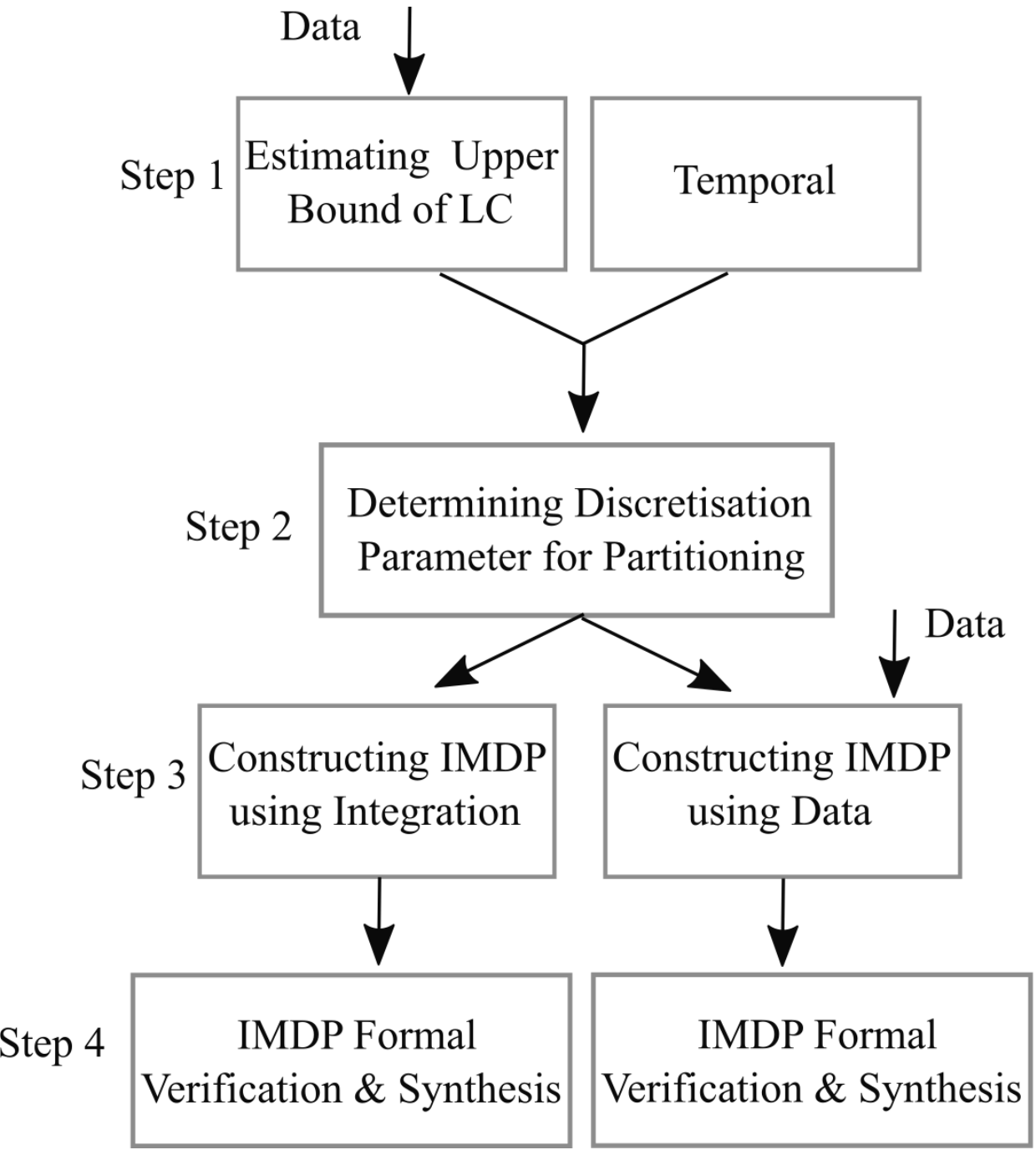}}
\vspace{-.7cm}
\caption{Workflow of the proposed data-driven formal verification framework.}
\label{verification_algorithm}
\vspace{-.8cm}
\end{figure}

Non-parametric methods, such as Gaussian process regression \citep{rasmussen2003gaussian}, non-parametric estimation (NPE) \citep{hardle2004nonparametric,scott2015multivariate}, and non-parametric least squares estimator \citep{ziemann2022single}, are widely used to study the dynamics of unknown systems based only on data.
For example, reachability \citep{jackson2021formal} and safety \citep{jagtap2020control,ahmadi2017safety} are studied for unknown dynamical systems by estimating a model with Gaussian process regression.
NPE has been used previously to approximate the invariant density of dynamical systems \citep{hang2018kernel}.
Among the aforementioned non-parametric methods, NPE has advantages in studying problems with no parametric form \citep{hardle2004nonparametric,scott2015multivariate}.
For instance, NPE is able to estimate the probability density function (PDF) of a random variable based on sampled data, ensuring it satisfies the statistical characteristics such as mean and variance and can be used for function regression.
For data-driven formal verification of unknown dynamical systems, the NPE can be employed to estimate the conditional stochastic kernel associated with the system dynamics. The LC of this estimated stochastic kernel provides a measure on the closeness between the system and its finite MDP abstraction. The transition probabilities between different states in the finite MDP can also be approximated using sampled data.
It should be noted that the LC estimation lacks well-established statistical guarantees, such as the bias of the LC, which affects the accuracy of the closeness mentioned in difficulty (i). This motivated us to give theoretical discussions on deriving asymptotic upper bound of the LC using NPE.
We present our results for the global LC of the stochastic kernel on a given domain. The same approach can be applied for computing LC locally on partitions of the state space, which then can be integrated with more efficient verification approaches based on local partitioning of the state space \citep{SA13}. 

The primary contribution of this paper is to present a data-driven formal verification approach for unknown stochastic systems based on the LC estimation as indicated in Fig.~\ref{verification_algorithm}.
We first adopt NPE to quantify an asymptotic upper bound for the LC estimation, which shows the LC estimating range has asymptotic convergence rate of $O(n^{-\frac{1}{3+d}})$ with $d\ge1$ being the state dimension of the system and $n$ being the data scale. This bound gives closeness guarantees between the system and its finite MDP abstraction. The process of determining the required partitioning for constructing the MDP abstraction is Step~2 in Fig.~\ref{verification_algorithm}, and is presented in Section~\ref{main_sec}. Then, two data-driven methods (an empirical approach and NPE), are used to construct an interval Markov decision processes (IMDPs) \citep{givan2000bounded} for the unknown stochastic system (Step~3 in Fig.~\ref{verification_algorithm}), which are presented in Section~\ref{datadriv_imdp}. Finally, we carry out formal verification on the constructed IMDP to verify the system against a given temporal logic specification, as presented in Section~\ref{case_sec}. 
Because of space limitations, the supplementary material in the appendix contains some preliminaries, proofs of statements, and numerical discussions.

\noindent\textbf{Related Work. }
A number of papers have studied the data-driven approaches for formal verification of dynamical systems. 
For systems with partial information of the model, safety and stability are studied using Bayesian framework \citep{schon2023verifying} and chance-constrained optimisation \citep{kenanian2019data}.
In addition, recently, formal verification for general unknown dynamical systems has drawn a lot of interest. For example, \cite{salamati2024data} verify the safety of unknown systems using barrier certificates which is construed by solving a scenario convex program based on sampled trajectories.
\cite{wicker2021certification} have studied reachability properties by employing Bayesian neural networks to make iterative predictions of the probability distributions of the system outputs and leveraging bound propagation techniques and backward recursion.
\cite{hashemi2023data} have proposed a data-driven approach for reachability analysis of stochastic systems with conformal inference.

A data-driven compositional reachability analysis framework is proposed by \cite{fan2017dryvr} for hybrid systems. Data-driven abstraction-based methods are studied by \cite{makdesi2021data,kazemi2024data} for control synthesis of non-probabilistic systems with formal guarantees.
Also, many data-driven formal verification approaches adopt Gaussian process regression to approximate the unknown system based on data and then take different formal verification techniques, such as IMDP abstraction \citep{lahijanian2015formal} and barrier certificates \citep{prajna2004safety}, to conduct formal verification \citep{jackson2021formal,jagtap2020control}.

Non-parametric estimation has been widely used in many fields including engineering, economics and biology \citep{hardle2004nonparametric,tsybakov2009introduction}. Parametric estimation assumes knowing a parameterised model of the system, and focuses on estimating the parameters. In contrast, non-parametric estimation does not assume any knowledge on the underlying model and estimates directly the model based on the observed data.
Non-parametric estimation has been used for deep learning and short-term forecasting \citep{huberman2021nonparametric}, estimation of graphical models \citep{zhu2017learning}, conditional information and divergence estimation \citep{poczos2012nonparametric},
and high-dimensional regression \citep{Izbicki2015NonparametricCD}.
Prediction of the LC of an unknown deterministic nonlinear function using the trajectory data is studied by \cite{chakrabarty2020safe}. Approximation of the invariant density of dynamical systems is studied by \cite{hang2018kernel}.

\section{PRELIMINARIES}\label{sec2}

\subsection{Problem Formulation}
We consider a discrete-time stochastic control system (DTSCS), which is a tuple $\Sigma_{ss}=(\mathcal S, U,w,f)$, 
where 
$\mathcal{S}\subset \mathbb R^n$ is the state space of the system,
$U$ is the input space of the system,
$w$ is a sequence of independent and identically-distributed (i.i.d.) random variables from a sample space $\Omega$ to the set $V_{w},$ i.e., $w:=\{ w(k):\Omega \to V_{w},k\in \mathbb N\},$
and $f:\mathcal{S}\times U\times V_{w}\to \mathcal{S}$ is a measurable function characterising the state evolution of $\Sigma_{ss}$ as
\begin{align}\label{dynamic_evolu}
    &x(k+1)=f(x(k),a(k),w(k)),\\
    &k\in \mathbb{N},~x(k)\in \mathcal{S},~ a(k)\in U \text{ and } w(k)\in V_{w}.\nonumber
\end{align}
Also, we define a set $\mathcal{U}_{\mathfrak a}$ that is the collection of sequences $\{a(k):\Omega\to U,k\in \mathbb N \}$. $a(k)$ is independent of $w(t)$ for any $k,t\in \mathbb N$ and $t\ge k$. 

In this paper, we assume that the system is unknown but data from sampled trajectories is available. For a compact presentation of the results, we focus on verifying the above system against temporal specifications through the data-driven construction of finite IMDPs when $\Sigma_{ss}$ is an autonomous system (i.e., the input space $U$ is a singleton). The presented approach is also applicable for control synthesis.

\subsection{Non-parametric Estimation of Density Functions}\label{sec:prelem}
Let $X=(X_1,\ldots,X_d)^T$ denote a $d$-dimensional random vector which has a continuous probability density function $f_{X}(\cdot)$.
For a given set of i.i.d. random samples $ \{\hat{X}_i = (\hat{X}_{i1},\ldots,\hat{X}_{id})^T, i=1,\ldots,n\}$, the general form of the multivariate kernel density estimator of $f_{X}(\cdot)$ is
\begin{equation}
\label{equ:MKDE}
\hat{f}_{X}(\bs{x})
=\frac{1}{n}
\sum_{i=1}^{n} 
K_{H}(\bs{x}-\hat{X}_i),\quad \forall \bs{x}\in \mathbb R^d,
\end{equation}
where
$K_{H}(\bs{u})=\frac{1}{|H|}K(H^{-1}\bs{u})$
and $K:\mathbb R^d\rightarrow\mathbb R_{\ge 0}$ is a \emph{multivariate kernel function}
 satisfying two moment conditions $\int K(\bs{u}) \,d\bs{u}=1$ and $\int \bs{u}K(\bs{u}) \,d\bs{u}=\bs{0}$.
 $H$ is a non-singular \emph{bandwidth matrix} and
 $|H|$ denotes the determinant of $H$.
 Examples of the univariate kernel function $K(\cdot)$ include uniform, triangle, quartic, and Gaussian kernel functions (see Appendix~\ref{kernel_bandwidth}). Multivariate kernel functions are typically chosen to be the product of univariate kernel functions \citep{hardle2004nonparametric}, i.e., the same kernel function with different bandwidths in each dimension:
 $$K(u_1,u_2,\ldots,u_d) = k(u_1)k(u_2)\ldots k(u_d)$$
for some univariate kernel function $k:\mathbb R\rightarrow\mathbb R_{\ge 0}$, and bandwidth matrix $H=diag(h_1,\ldots,h_d)$. A popular choice for the kernel function is the Gaussian kernel $k(u)=\frac{1}{\sqrt{2\pi}}\exp(-u^2/2) $, which leads to the following estimator for $f_X(\cdot)$
\begin{equation}
 \label{eq:gaussian}
    \hat{f}_{X}(\bs{x})
\!=\! \frac{(2\pi)^{-d/2}}{nh_1\cdots h_d} \cdot 
\sum_{i=1}^{n}
\prod_{j=1}^{d}
\exp\left[-\frac{1}{2}\left(\frac{x_j-\hat{X}_{ij}}{h_j}\right)^2\right], 
\end{equation}
with $x_j$ and $\hat X_{ij}$ being the $j^{\text{th}}$ elements of $\bs{x}$ and $\hat X_i$, respectively.
We will use this estimator in the rest of this paper to establish our theoretical results.

The accuracy of the estimation is widely assessed using the mean integrated squared error (MISE), which is used for selecting the kernel function and the bandwidth matrix.
The \emph{asymptotic} MISE, bias, and variance of the estimation is generally obtained by eliminating the higher-order terms.
The asymptotic bias and variance of $\hat{f}_X(\bs{x})$ in equation \eqref{equ:MKDE} are derived by \cite{hardle2004nonparametric} as
\begin{align*}
\Bias[\hat{f}_{X}(\bs{x})] & \approx \frac{1}{2}\mu_2(K) \,\textsf{tr}(H^T\mathcal{H}_f(\bs{x})H) \\
\Var[\hat{f}_{X}(\bs{x})]  &\approx \frac{1}{n\;|H|}\|K\|_2^2f_X(\bs{x}),
\end{align*}
where $\mu_2(K)$ is a constant defined with $\int \bs{u}\bs{u}^TK(\bs{u})\,d\bs{u}=\mu_2(K)\bs{I}_d$,
$\mathcal{H}_f(x)$ is the Hessian matrix of second partial derivatives of $f$, 
$\textsf{tr}(\cdot)$ denotes the trace of a matrix,
and $\|K\|_2$ is the $L_2$-norm of $K$.
Then, the asymptotic mean integrated squared error (AMISE) can be formulated as
\[
\AMISE(\hat{f}_{X} )\!=\!\frac{1}{4}\mu_2^2(K)\!\!\int\! \textsf{tr}(H^T\mathcal{H}_f(\bs{x})H)^2 d\bs{x} +\frac{\|K\|_2^2}{n|H|}.
\]
The AMISE is primarily determined by the choice of the bandwidth, with the kernel function having a minor effect only through specific characteristics such as the order of its first nonzero moments \citep{hardle2004nonparametric,scott2015multivariate}.
In addition, the optimal choice of the bandwidth improves the performance of kernel density estimators and keeps the balance between the bias and the variance.
Otherwise, the unsuitable bandwidths can result in a large variance and small bias (under-smoothing), or a small variance and large bias (over-smoothing). 
A tradeoff between over- and under-smoothing of the density estimator can be achieved by minimising the AMISE.
Related works on important properties of kernels and the choice of the bandwidth are presented in Appendix~\ref{kernel_bandwidth}.
{\color{black}In this paper, we mainly adopt Scott's formula shown in Appendix~\ref{kernel_bandwidth} to calculate the bandwidths.}


\paragraph{Estimating Conditional Density Functions.}
\citet{rosenblatt1969conditional} introduced the standard kernel estimator of a conditional density function (CoDF) by replacing the estimates of the joint and marginal densities in the definition of the CoDF. Using a kernel estimator in equation \eqref{equ:MKDE} that is the product of kernels for two random vectors $X$ and $Y$, the estimator of the joint density function of $(Y,X)$ and the marginal density function of $X$ are given by
\begin{align*}
\begin{split}
\hat{f}_{(Y,X)}(\bs{y},\bs{x})=&
\frac{1}{n}
\sum_{i=1}^n
K_{H_{\mathsf x}}(\bs{x}-\hat{X}_i)
\cdot K_{H_{\mathsf y}}(\bs{y}-\hat{Y}_i) \\
\hat{f}_{X}(\bs{x})=&
\frac{1}{n}
\sum_{j=1}^n
K_{H_{\mathsf x}}(\bs{x}-\hat{X}_j),
\end{split}
\end{align*}
where $K_{H_{\mathsf x}}$ and $K_{H_{\mathsf y}}$ are kernels with bandwidth matrices $H_{\mathsf x}$ and $H_{\mathsf y}$, respectively.
Thus, the kernel estimator of $f_{Y|X}$ is given by
\begin{align}
\hat{f}_{Y|X}(\bs{y},\bs{x})&=
\frac{\hat{f}_{(Y,X)}(\bs{y},\bs{x})}{\hat{f}_{X}(\bs{x})}\nonumber\\
&=
\frac{\sum_{i=1}^n
K_{H_{\mathsf x}}(\bs{x}-\hat{X}_i)
\cdot K_{H_{\mathsf y}}(\bs{y}-\hat{Y}_i)}
{\sum_{j=1}^n
K_{H_{\mathsf x}}(\bs{x}-\hat{X}_j)}.
\label{condi_densityestima}
\end{align}

The asymptotic bias and variance of the conditional density estimator~\eqref{condi_densityestima} have been obtained by \citet{hyndman1996estimating} for univariate $X$ and $Y$, and can be found in the supplementary material in equation \eqref{bias_from1996} and equation \eqref{var_from1996}. In the following sections, we use equation \eqref{condi_densityestima} to estimate the conditional density function of equation \eqref{dynamic_evolu}, which gives the density function of the next state as a random vector conditioned on the current state.

\subsection{Interval Markov Decision Processes}

Interval Markov Decision Process (IMDP) is a type of Markov decision process that has transition probabilities taking values inside given intervals \citep{givan2000bounded}.

\begin{definition}[\textbf{IMDP}]
An IMDP is a tuple $\Sigma=( Q, S_{\mathfrak a}, P_{lo}, P_{up}, AP, L ),$ where $Q$ is a finite set of states, $S_{\mathfrak a}$ is a finite set of actions and $S_{\mathfrak a}(q)$ is the set of actions at state $q\in Q$, $P_{lo}: Q\times S_{\mathfrak a}\times Q\to [0,1]$ is a function representing the lower bound of the transition probability from $q\in Q$ to $q^*\in Q$ under action $a\in S_{\mathfrak a}$, $P_{up}: Q\times S_{\mathfrak a}\times Q\to [0,1]$ is a function representing the upper bound of the transition probability from $q$ to $q^*$ under action $a\in S_{\mathfrak a}$, $AP$ is a finite set of atomic propositions, and $L:Q\to 2^{AP}$ is a labelling function assigning possibly several elements of $AP$ to each state $q$.
\end{definition}

For any $q,q^{*}\in Q$ and $a\in S_{\mathfrak a}(q)$, it is true that $P_{lo}(q,a,q^{*})\leq P_{up}(q,a,q^*)$ and $\sum_{q\in Q}P_{lo}(q,a,q^*)\leq1\leq \sum_{q\in Q}P_{up}(q,a,q^*)$.The set of probability distributions over $Q$ is denoted by $D(Q)$.
$\theta^{a}_{q}\in D(Q)$ represents a feasible distribution initiated from $q\in Q$ to all successor states in $Q$ under $a$, and satisfies $P_{lo}(q,a,q^*)\leq \theta^{a}_{q}(q^*)\leq P_{up}(q,a,q^*)$, where $q^*$ is the successor state. The set of all feasible distributions initiated from $q$ under $a$ is denoted by $\Theta^{a}_{q}$. A path $\nu=q_{0}\xrightarrow{a_0}q_{1}\xrightarrow{a_1}q_{2}\xrightarrow{a_2}\ldots,$ $a_{i}\in S_{\mathfrak a}(q_{i})$, represents a path of the IMDP and satisfies $P_{up}(q_{i}, a_{i}, q_{i+1})>0$ for all $i\in \mathbb{N}$. The last state of a finite path $\nu^{\mathsf{fin}}$ is denoted by $\textsf{last}(\nu^{\mathsf{fin}})$.
The sets of all finite and infinite paths are denoted by $\textsf{Paths}^{\mathsf{fin}}$ and $\textsf{Paths}$, respectively.
Let a function $\varpi: \textsf{Paths}^{\mathsf{fin}}\to S_{\mathfrak a}$ denote a strategy on the IMDP $\Sigma$, which maps a finite path $\nu^{\mathsf{fin}}$ of $\Sigma$ onto an action in $S_{\mathfrak a}$. The set of all such strategies is denoted by $\boldsymbol{\Pi}$. An MDP is an IMDP with all probability intervals being a singleton (i.e., with $P_{lo} = P_{up}$).

\begin{definition}[\textbf{Adversary}]
Consider an IMDP $\Sigma$. An adversary is a function $\kappa: \textsf{Paths}^{\mathsf{fin}}\times S_{\mathfrak a}\to D(Q)$, which maps the path-action pair $(\nu^{\mathsf{fin}},a)$  with $a\in S_{\mathfrak a}(\textsf{last}(\nu^{\mathsf{fin}}))$ to a feasible distribution $\theta^{a}_{q}\in\Theta^{a}_{\textsf{last}(\nu^{\mathsf{fin}})}$. The set of all adversaries is denoted by $\boldsymbol{K}$.
\end{definition}

\subsection{IMDP Verification}

Here, we give an introduction of IMDP verification and policy synthesis against a specification described in probabilistic computation tree logic (PCTL) \citep{baier2008principles}. The details of PCTL are provided in Appendix~\ref{PCTL}. 
For a PCTL path formula $\psi$ starting from an initial state $q\in Q$, the lower and upper bound of probabilities that the paths initialised at $q$ satisfy $\psi$ in $k$ steps can be defined as 
\begin{align}\label{lowprob_path} 
    P^{k}_{lo}(q)\!=\!\begin{cases}
       1,~\text{if}~q\in Q^{1},\\
    0,~\text{if}~q\in Q^{0}, \\
    0,~\text{if}~q\notin (Q^{0}\cup Q^{1})\wedge k=0,\\
    \min_{a}\min_{\theta^{a}_{q}} \sum_{q^{*}}\theta^{a}_{q}(q^{*})P^{k-1}_{lo}(q^{*}),~\text{otherwise},
\end{cases}
\end{align}
\begin{align}
\label{upprob_path}
     P^{k}_{up}(q)\!=\!\begin{cases}
       1,~\text{if}~q\in Q^{1},\\
    0,~\text{if}~q\in Q^{0},\\
    0,~\text{if}~q\notin (Q^{0}\cup Q^{1})\wedge k=0,\\
    \!\max_{a}\max_{\theta^{a}_{q}} \sum_{q^{*}}\theta^{a}_{q}(q^{*})P^{k-1}_{up}(q^{*}),~\text{otherwise},
\end{cases}
\end{align}
where $Q^{1}$ is the set of states that always satisfy the path formula $\psi$, $Q^{0}$ is the set of states that never satisfy $\psi$, and $\theta^{a}_{q}(q^*)\in [P_{lo}(q,a,q^*),P_{up}(q,a,q^*)]$, for any $q^{*}\in Q$.
The adversaries obtained from the procedures above determine a series of actions that lead to maximum and minimum probabilities satisfying path formula $\psi$ for each state.
The above recursive computation of the probability bounds can be performed in a finite number of steps for specifications with bounded until ($\mathcal{U}^{\leq k}$). The number of steps can be tuned with respect to any desired accuracy for specifications with unbounded until ($\mathcal{U}$).

\section{LC ESTIMATION METHOD AND CLOSENESS GUARANTEE}\label{main_sec}

In this section, we propose a novel algorithm using NPE to estimate the Lipschitz constant (LC) of the CoDF of a stochastic system $\Sigma_{ss}$, and give an asymptotic upper bound for the estimation of the LC. This upper bound is useful to find a partitioning strategy to guarantee the closeness between the satisfaction probabilities of the specifications on $\Sigma_{ss}$ and on its finite abstraction $\Sigma$.
Determining the partitioning strategy according to this upper bound is provided by \cite{SA13} as briefly discussed in Appendix~\ref{Prob_guarantee}.
In addition, the effectiveness of this estimation method is demonstrated by studying several cases in Appendix~\ref{supp_validation}.

\subsection{Estimation Method for LC}
We propose a method to estimate the LC of a CoDF $f_{Y|X}$ on a given domain $D_X\times D_Y$ using equation \eqref{condi_densityestima}. The estimation method is presented in Algorithm~\ref{algo_enviro}, which is based on taking samples of $X$ uniformly on $D_X$, then taking samples of $Y$ from $(Y|X)$ associated with samples of $X$, constructing $\hat f_{Y|X}$ according to equation \eqref{condi_densityestima},
taking partial derivative of $\hat f_{Y|X}$, and finally computing the maximum absolute values of derivatives on the domain $D_X\times D_Y$ and across dimensions.
The algorithm also iterates over these steps and compute the empirical mean of the results in Step~\ref{step_emp}. Note that the max operator in Step~\ref{step_max} corresponds to using infinity norm in the definition of the LC. Other norms could be used similarly.

\begin{algorithm}
\caption{Estimating the Lipschitz constant of $f_{Y|X}(\bs{y},
\bs{x})$}
\label{algo_enviro}
      
    \begin{algorithmic}[1]
    \REQUIRE  Domain $D_X\times D_Y$, sample generators of $(Y|X)$, number of iterations $m$
    \FOR{$\mu=1:m$}
    \STATE\label{step_H} Select bandwidths $H_{\mathsf x}$ and $H_{\mathsf y}$ and kernel $K(\cdot)$
    \STATE\label{step_sample} Select samples $ \{\hat{X}_i, i=1,\ldots,n\}$ uniformly from $D_X$
    \STATE For each $\hat X_i$, generate a sample $ \hat{Y}_i\in D_Y$ from $(Y|X_i)$, $i\in\{1,2,\ldots,n\}$
    \STATE Construct $\hat{f}_{Y|X}(\bs{y},
     \bs{x})$ using equation \eqref{condi_densityestima}, samples $(\hat Y_i,\hat X_i)$, and kernel $K(\cdot)$
    \STATE For each $j\in\{1,\ldots,d\}$, compute $\hat L_{\mu j}$ as
    \begin{equation}
    \label{eq:deriv_hat}
    \hat L_{\mu j} := \max_{(\bs{x},\bs{y})\in D_X\times D_Y}\left|\frac{\partial}{\partial x_{j}}\hat{f}_{Y|X}(\bs{y},\bs{x})\right|
    \end{equation}
    \ENDFOR
    \STATE\label{step_emp}Compute the empirical means $\hat L_j:=\frac{1}{m}\sum^{m}_{\mu=1}\hat L_{\mu j}$
    \STATE\label{step_max} Compute $\hat L=\max\{\hat L_{1},\hat L_{2},\ldots,\hat L_{d}\}$
    \ENSURE Estimated LC $\hat L$
    \end{algorithmic}
\end{algorithm}

In the rest of this section, we formulate bounds on the bias and variance of the estimator $\hat L_\mu$ for one- and multi-dimensional cases and discuss how to select the bandwidths $H_{\mathsf x}$ and $H_{\mathsf y}$ (cf. Step~\ref{step_H} of Algorithm~\ref{algo_enviro}) to tune the asymptotic bias of the estimation.
The total variance of the estimation can be reduced by increasing the iteration number $m$ of the algorithm.
Our theoretical results are established for the uniform distribution in Step~\ref{step_sample} of the algorithm. Similar results can be obtained for other distributions.

\subsection{Univariate Systems}
\label{theo_onedimen}

%
For the CoDF $f_{Y|X}(y,x)$ with one-dimensional $X$ and $Y$ and being continuously differentiable with respect to $x$, the LC on the domain $D_X\times D_Y\subset \mathbb R^2$ is 
\begin{equation}
\label{onelc_exp}
    L:=\max_{(x,y)\in D_X\times D_Y}|\frac{\mathrm{d}}{\mathrm{d}x}f_{Y|X}(y,x)|,
\end{equation}
with the LC estimator from one iteration of Algorithm~\ref{algo_enviro}:
\begin{equation}
\label{onelc_est}
\hat{L}:=\max_{(x,y)\in D_X\times D_Y}|\frac{\mathrm{d}}{\mathrm{d}x}\hat{f}_{Y|X}(y,x)|.
\end{equation}
Our main task is to show that the bias and variance of this $\hat L$ has nice asymptotic properties with respect to the data scale $n$.
We raise the following assumption.
\begin{assumption}
\label{ass_1}
\textbf{{\normalfont(a)}}
There exists a constant $C_{f}>0$ such that $| f_{Y|X}(y,x)|\leq C_{f}$ for all $(x,y)\in D_X\times D_Y$.
\textbf{{\normalfont(b)}}
There exist positive constants $C_{b1}$ and $C_{b2}$ such that $|\frac{\mathrm{d}^{3}}{\mathrm{d}x\mathrm{d}y^{2}}f_{Y|X}(y,x)|\leq C_{b1}$ and $|\frac{\mathrm{d}^{3}}{\mathrm{d}x^{3}}f_{Y|X}(y,x)|\leq C_{b2}$ for all $(x,y)\in D_X\times D_Y$.
\end{assumption}

\begin{remark}
Note that for the purpose of our discussions, it is sufficient to know any rough upper bounds $C_f,C_{b1},C_{b2}$.
In general, it is common to require some information on higher derivatives to make certain conclusions. For example, polynomial interpolation with degree $n$ requires a bound on the $(n+1)^{\text{st}}$ derivative to give the interpolation error. Otherwise, the convergence of the interpolation cannot be guaranteed \citep{burden2015numerical}. Another example is the unconstrained optimisation with the second order necessary condition that requires the Hessian matrix to be positive definite \citep{ruszczynski2011nonlinear}. In order to ensure the mean square error (MSE) of non-parametric density estimation of $f_X$ has certain asymptotic properties, the Hessian matrix $\mathcal{H}_f(x)$ of $f_X$ should satisfy the condition that $p^{T}\mathcal{H}_f(x)p$ is bounded for all $p\in \mathbb R^d$ \citep{hardle2004nonparametric}. Thus, it is natural and reasonable to require the rough upper bound of third derivatives as mentioned in Assumption~\ref{ass_1}(b).
\end{remark}

We first give upper bounds for the asymptotic bias and variance of $\frac{\mathrm{d}}{\mathrm{d}x}\hat{f}_{Y|X}(y,x)$ based on the approach of \citet{hyndman1996estimating} on the asymptotic bias and variance of univariate conditional density estimators.
We use the symbol $\lesssim$ to denote the asymptotic bound when $n\to+\infty$ by eliminating higher order terms.
In this subsection, 
we use $K:\mathbb{R} \to \mathbb{R}$ to be the \emph{Gaussian kernel function} in the estimator \eqref{condi_densityestima} with bandwidths $h_{\mathsf x}$ and $h_{\mathsf y}$.
We also denote $G_{ji}:=\int \nu^{i}K^{j}(\nu)d\nu$ for $i,j\in\{0,1,2,\dots,6\}$, $j\ge 1$.

\begin{lemma}
\label{le_vari}
Suppose that $f_{Y|X}(y,x)$ satisfies {\normalfont Assumption~\ref{ass_1}(a)}.
For any $(x,y)\in D_X\times D_Y$, $h_{\mathsf x},h_{\mathsf y}>0$, we have that for large $n$ if $nh_{\mathsf x}^{3}h_{\mathsf y}\to +\infty$ and $h_{\mathsf x},h_{\mathsf y}\to 0$ as $n\to +\infty$, then
\begin{align*}
    \Var\left[\frac{\mathrm{d}}{\mathrm{d}x}\hat{f}_{Y|X}(y,x)\right] \lesssim \frac{C_{1}}{n h_{\mathsf x}^{3}h_{\mathsf y}},
\end{align*}
where $\lesssim$ denotes an asymptotic bound for large $n$ and $C_{1}:=\Vol(D_X)G_{20}(K)C_{f}$ with $\Vol(\cdot)$ indicating the volume (Lebesgue measure) of a set.
\end{lemma}

\begin{lemma}\label{le_bias}
Suppose $f_{Y|X}(y,x)$ 
satisfies Assumption~\ref{ass_1}.
We have that if $h_{\mathsf x},h_{\mathsf y}\to 0$ as $n\to +\infty$, then
\begin{align*}
    &\left|\E\left[\frac{\mathrm{d}}{\mathrm{d}x}\hat{f}_{Y|X}(y,x)\right]-\frac{\mathrm{d}}{\mathrm{d}x}f_{Y|X}(y,x)\right|
    \lesssim \frac{h^{2}_{\mathsf x}}{2}A,
\end{align*}
with $A:=G_{12}(K)[\frac{h_{\mathsf y}^{2}}{h_{\mathsf x}^{2}}C_{b1}+C_{b2}].$
\end{lemma}

We get the following result by combining Lemmas~\ref{le_vari}--\ref{le_bias}.

\begin{theorem}\label{le_mse}
Suppose that $f_{Y|X}(y,x)$ satisfies Assumption~\ref{ass_1}. 
For any $(x,y)\in D_X\times D_Y$, $h_{\mathsf x},h_{\mathsf y}>0$, we have that for large $n$ if $nh_{\mathsf x}^{3}h_{\mathsf y}\to +\infty$ and $h_{\mathsf x},h_{\mathsf y}\to 0$ as $n\to +\infty$, then
\begin{equation*}
\label{eq:MSE}
    \E\left[\left(\frac{\mathrm{d}}{\mathrm{d}x}\hat{f}_{Y|X}(y,x)-\frac{\mathrm{d}}{\mathrm{d}x}f_{Y|X}(y,x)\right)^{2}\right]\lesssim \epsilon_{3},
\end{equation*}
where $\epsilon_{3}:=\frac{C_{1}}{nh_{\mathsf x}^{3}h_{\mathsf y}}+\frac{h_{\mathsf x}^{4}}{4}A^{2}$.
\end{theorem}

The above theorem has already presented an upper bound of the MSE between $\frac{\mathrm{d}}{\mathrm{d}x}\hat{f}_{Y|X}(y,x)$ and $\frac{\mathrm{d}}{\mathrm{d}x}f_{Y|X}(y,x)$, which can be used to give the following main result.

\begin{theorem}[\textbf{Bias of the Estimation}]
\label{LC_main}
Suppose $f_{Y|X}(y,x)$ satisfies {\normalfont Assumption~\ref{ass_1}}.
For any $(x,y)\in D_X\times D_Y$, $h_{\mathsf x},h_{\mathsf y}>0$, we can have for large $n$ if $nh_{\mathsf x}^{3}h_{\mathsf y}\to +\infty$ and $h_{\mathsf x},h_{\mathsf y}\to 0$ as $n\to +\infty$, then
\begin{equation*}
|\E[\hat{L}]-L|\lesssim \epsilon_{3}^{\frac{1}{2}}, 
\end{equation*}
with $\epsilon_{3}\!=\!\frac{C_{1}}{nh_{\mathsf x}^{3}h_{\mathsf y}}\!+\!\frac{h_{\mathsf x}^{4}}{4}A^{2}$ defined in Theorem~\ref{le_mse}.
\end{theorem}

\begin{remark}
The bandwidths should be selected appropriately such that $\epsilon_{3}\to 0$ for large $n$. The best convergence rate for $\epsilon_{3}$ is $O(n^{-1/2})$, which is obtained by setting $h_{\mathsf x} = h_{\mathsf y}$ in the order of $n^{-\frac{1}{8}}$. Hence, $\epsilon_{3}^{\frac{1}{2}}$ has the best convergence rate of $O(n^{-\frac{1}{4}})$.
\end{remark}

\begin{remark}
For a given precision $\epsilon^*$, we can use Theorem~\ref{LC_main} to select a large enough data scale $n$ and bandwidths $h_{\mathsf x} = h_{\mathsf y} = n^{-\frac{1}{8}}$ such that the bias in the output of the estimation algorithm is at most $\epsilon^*$.
\end{remark}

\subsection{Multi-variate Systems}\label{theo_twodimen}
The results obtained in Section~\ref{theo_onedimen} can be extended to high-dimensional cases $d\ge2$.
Consider the original and estimated LC across each dimension defined as
\begin{align*}
L_{i}& =\max_{(\bs{x},\bs{y})\in D_X\times D_Y}|\frac{\partial}{\partial x_{i}}f_{Y|X}(\bs{y},\bs{x})|,\\
\hat{L}_{i}& =\max_{(\bs{x},\bs{y})\in D_X\times D_Y}|\frac{\partial}{\partial x_{i}}\hat{f}_{Y|X}(\bs{y},\bs{x})|,
\end{align*}
for $i\in \{1,\ldots,d\}$ with $D_X$, $D_Y\subset \mathbb R^{d}$.
Thus, the original and estimated LC are $L:=\max_{i=1,\ldots,d}\{L_i\}$ and $\hat{L}:=\max_{i=1,\ldots,d}\{\hat L_i\}$.
\begin{assumption}\label{asstwod_1}
\textbf{{\normalfont(a)}}
There exists a constant $C_{f}\!>\!0$ such that $| f_{Y|X}(\bs{y},\bs{x})|\!\leq\! C_{f}$, for all $(\bs{x},\bs{y})\in D_X\!\times\! D_Y$.
\textbf{{\normalfont(b)}}
There exist constants $C_{ij}\!>\!0$, $i,j\in \{1,\ldots,d\}$, such that $|\frac{\partial^{3}}{\partial x_{i}\partial y_{j}^{2}}f_{Y|X}(\bs{y},\bs{x})|\!\leq\! C_{ij}$, for all $(\bs{x},\bs{y})\in D_X\!\times\! D_Y$.
\textbf{{\normalfont(c)}}
There exist constants $C_{xi}\!>\!0$, $i\in \{1,\ldots,d\}$, such that $| \frac{\partial^{3}}{\partial x_{i}\partial x^{2}_{\varsigma}}f_{Y|X}(\bs{y},\bs{x})|\!\leq\! C_{xi}$, for all $(\bs{x},\bs{y})\in D_X\!\times\! D_Y$ and $i\ne \varsigma$, $\varsigma\in\{1,\ldots,d \}$.
\end{assumption}

Following a similar method as before, we have
\begin{align*}
    \Var[& \frac{\partial}{\partial x_{i}} \hat{f}_{ Y|X }(\bs{y},\bs{x}) ]\\
    \approx&\! \frac{1}{nh_{\mathsf x i}^{2} \prod_{j=1}^d\! h_{\mathsf x j}h_{\mathsf y j}}G^{2d-1\!}_{2,0}(K)\Vol(D_{X})f_{Y|X}(\bs{y},\bs{x}),
    \end{align*}
\begin{align*}
    &\left|\E[\frac{\partial}{\partial x_{i}} \hat{f}_{ Y|X }(\bs{y},\bs{x})]-\frac{\partial}{\partial x_{i}} f_{ Y|X }(\bs{y},\bs{x})\right|\\
    &\quad\approx \frac{1}{2}h_{\mathsf x i}^{2} \bigg| \sum^{d}_{j=1}\frac{h^2_{\mathsf y j}}{h^2_{\mathsf x i}}\frac{\partial ^{3}}{\partial x_{i}\partial y^{2}_{j}} f_{Y|X}(\bs{y},\bs{x})\\
    &\quad\quad+\sum^{d}_{\varsigma\ne i}\frac{h_{\mathsf x \varsigma}^{2}}{h_{\mathsf x i}^{2}}\frac{\partial^{3}}{\partial x_{i}\partial x^{2}_{\varsigma}}f_{Y|X}(\bs{y},\bs{x})\bigg|,~i\in\{1,\ldots,d\},
\end{align*}
if $nh_{\mathsf x i}^{2} \prod_{j=1}^d\! h_{\mathsf x j}h_{\mathsf y j}\to +\infty$ and $h_{\mathsf x j},h_{\mathsf y j}\to 0$, $j\in\{1,\ldots,d\}$, as $n\to +\infty$,
where $\bs{x}:=(x_1,\ldots,x_{d})^{T}\in D_X$, $\bs{y}:=(y_1,\ldots,y_2)^{T}\in D_Y$, $h_{\mathsf x j}$ and $h_{\mathsf y j}$, $j\in\{1,\ldots,d\}$ are the bandwidths of d-dimensional random vector $\hat{X}$ and $\hat{Y}$.
Assuming that 
$
    A_{i}\!:=\!\sum^{d}_{j=1}\frac{h^2_{\mathsf y j}}{h^2_{\mathsf x i}}\frac{\partial ^{3}}{\partial x_{i}\partial y^{2}_{j}} f_{Y|X}(\bs{y},\bs{x})
    +\sum^{d}_{\varsigma\ne i}\frac{h_{\mathsf x \varsigma}^{2}}{h_{\mathsf x i}^{2}}\frac{\partial^{3}}{\partial x_{i}\partial x^{2}_{\varsigma}}f_{Y|X}(\bs{y},\bs{x})
$
and 
$
    \hat{C}:=\Vol(D_X)G_{20}^{2d-1}(K)C_{f}.
$
Then, we can guarantee that 
\begin{equation*}
    \label{main_twodimen}
     \left|\E[\hat{L}_{i}] - L_{i}\right|\lesssim \epsilon_{3i}^{\frac{1}{2}}, 
\end{equation*}
where $\epsilon_{3i}\!:=\!\frac{\hat{C}}{nh_{\mathsf x i}^{2} \prod_{j=1}^d\! h_{\mathsf x j}h_{\mathsf y j}}+\frac{h_{\mathsf x i}^{4}}{4}A^{2}_{i}, ~i\in \{1,\ldots,d\}.$

The bandwidths should be selected appropriately such that $\epsilon_{3i}\to 0$ for large $n$. The best convergence rate for $\epsilon_{3i}$ is $O(n^{-\frac{2}{3+d}})$, which is obtained by setting $h_{\mathsf x i} = h_{\mathsf x j}=h_{\mathsf y i} = h_{\mathsf y j}$, $i\ne j$, $i,j\in \{1,\ldots,d\}$, in the order of $n^{-\frac{1}{6+2d}}$.
Hence, $\epsilon^{\frac{1}{2}}_{3,i}$ has the best convergence rate of $O(n^{-\frac{1}{3+d}})$.


\subsection{Compositional Estimation for Structured Systems}\label{structured_system}
In this subsection, we discuss how the computation of the LC can be adapted to any structure of the system. Consider the dynamical system in the form of $\bs{x}(k+1) = g(\bs{x}(k),\bs{w}(k))$, $k\in\{0,1,2,\ldots\}$, written explicitly with its states $\bs{x}= [x_{1},\ldots,x_d]^T$, the vector field $g=[g_{1},\ldots,g_{d}]^T$, and stochastic disturbances $\bs{w}=[w_{1},\ldots,w_d]^T$, as follows: 
\begin{align*}
    &x_{1}(k+1)=g_{1}(x_{1}(k),x_{2}(k),\ldots,x_{d}(k),w_1(k)),\nonumber\\
    &x_{2}(k+1)=g_{2}(x_{1}(k),x_{2}(k),\ldots,x_{d}(k),w_2(k)),\nonumber\\
    &\vdots\\
    &x_{d}(k+1)=g_{d}(x_{1}(k),x_{2}(k),\ldots,x_{d}(k),w_d(k)).\nonumber
\end{align*}
If $w_{1},\ldots,w_{d}$ are independent, the CoDF of the system will take the following product form
\begin{equation}\label{pdf_structsyst}
    T(\bar{\bs{x}}|\bs{x})=T_{1}(\bar x_{1}|\bs{x})\times T_{2}(\bar x_{2}|\bs{x})\times \ldots \times T_{d}(\bar x_{d}|\bs{x}),
\end{equation}
where $\bar{\bs{x}} = [\bar x_{1},\ldots,\bar x_d]^T$,  the function $T_{i}:\mathbb R\times \mathbb R^{d}\to \mathbb R_{\ge 0}$ depends on $g_{i}$ and the distribution of $w_{i}$.
The estimation of the LC of $T(\bar{\bs{x}}|\bs{x})$ in equation \eqref{pdf_structsyst} can be reduced to the estimation of the LC of each $T_{i}$.
Moreover, any additional information on the dependency of the functions $g_i$ to $x_j$ can be reflected into the structure of $T_i$. This will result in estimating the LC of conditional densities that have smaller number of variables, thus improving the computational efficiency of the estimation. 
The asymptotic upper bound of each LC of $T_{i}$ can be derived utilising the theories presented in the section.
We demonstrate in Appendix~\ref{supp_validation} the estimation of the LC for a 7-dimensional autonomous vehicle \citep{althoff2019commonroad} using its structure. The model is in Appendix~\ref{auto_car} for reference.

\section{DATA-DRIVEN CONSTRUCTION OF IMDP}\label{datadriv_imdp}


\subsection{IMDP Based on an Empirical Approach}
\label{empirical_method}
Here, we use an empirical approach and Chebyshev's inequality to construct an IMDP as a finite abstraction of the system $\Sigma_{ss}=(\mathcal S, U,w,f)$ with formal closeness guarantees.
First consider an MDP $\hat \Sigma_{ss}=(Q, S_{\mathfrak a}, P, AP, L )$ with $Q$ representing a partition of the state space $\mathcal S$ with partition sets denoted by $q\in Q$. The input space $S_{\mathfrak a} = U$. For the transition probabilities $P^{a}_{ij}$, select representative points $\bar q\in q$ for each partition set. Define $P^{a}_{ij} :=\textsf{Prob}_w(f(\bar q_i,a,w)\in q_j)$. Next we build an IMDP $\bar \Sigma_{ss}=(Q, S_{\mathfrak a}, P_{lo},P_{up}, AP, L )$ from the empirical estimation of $P^{a}_{ij}$. 
By gathering $N_{ij}$ pair of sampled trajectories $(x(k) = \bar q_i, x(k+1) = q_j^m)$, $m\in\{1,2,\ldots,N_{ij}\}$ we can define the empirical value for $P^{a}_{ij}$ as $\bar{P}^{a}_{ij}=\frac{1}{N_{ij}}\sum^{N_{ij}}_{m=1}\boldsymbol{1}(q^m_j\in q_j|x(k) = \bar{q}_i)$, where $\boldsymbol{1}(\cdot)$ is an indicator function which is one if $q^m_j\in q_j$ and zero otherwise. 
By a-priori fixing a threshold $\bar{\epsilon}\in (0,1]$ and a confidence $\Bar{\beta}\in(0,1)$, according to Chebyshev's inequality \citep{saw1984chebyshev}, we have 
\begin{equation}
    \label{ineq_empircalappro}
    \mathbb{P}\{\Bar{P}^{a}_{ij}-\Bar{\epsilon}\leq P^{a}_{ij}\leq \Bar{P}^{a}_{ij}+\Bar{\epsilon} \}\ge 1-\Bar{\beta},
\end{equation}
where $N_{ij}\ge \frac{1}{4\Bar{\beta}\Bar{\epsilon}^2}.$ 
The inequality~\eqref{ineq_empircalappro} gives the probability interval $[\Bar{P}^{a}_{ij}-\Bar{\epsilon},\Bar{P}^{a}_{ij}+\Bar{\epsilon}]$ for $P^{a}_{ij}$, which holds with confidence at least $1-\Bar{\beta}$. 


The probability of verification against PCTL path formula using the empirical approach can be bounded to a certain range around its true probability by appropriately choosing $\Bar{\epsilon}$ satisfying the following lemma. The proof of the following lemma is shown in Appendix~\ref{proof_sec41}.

\begin{lemma}\label{truevalue_distan}
     Given $\epsilon_g\in (0,1)$, if $\Bar{\epsilon}=\frac{\epsilon_g}{2kn_Q}$, where $n_Q$ is the number of elements of $Q$ and $k$ is the time step, then $$|P^{k}_{up}(q)-\hat{P}^{k}_{up}(q)|\leq \epsilon_g,$$ where $P^{k}_{up}(q)$ is the solution of \eqref{upprob_path} obtained for the IMDP $\bar \Sigma_{ss}$
     and $\hat{P}^{k}_{up}(q)$ is the solutions of \eqref{upprob_path} computed for the MDP $\hat \Sigma_{ss}$.
\end{lemma}

Using this lemma, we establish the closeness between DTSCS $\Sigma_{ss}$ and its finite abstraction based on empirical data as follows.

\begin{theorem}\label{empirical_guarantee}
    Let $\Sigma_{ss}$ be the DTSCS and $\Bar{\Sigma}_{ss}$ be its finite abstraction based on empirical data with its transition probabilities are obtained from equation \eqref{ineq_empircalappro}. For any given PCTL specification $\psi$ through a certain strategy $\varpi\in \boldsymbol{\Pi}$ satisfying procedure \eqref{upprob_path}, if $\Bar{\epsilon}=\frac{\epsilon_g}{2kn_Q},~\epsilon_g\in (0,1)$, then we can have
    \begin{align*}
        |P(\Sigma_{ss}\vDash \psi)-P(\Bar{\Sigma}_{ss}\vDash \psi)|\leq \epsilon+\epsilon_g, \text{with $\epsilon=k\delta B_L \mathfrak{L}$,}
    \end{align*}
     where $P(\Sigma_{ss}\vDash \psi)$ is the probability that $\Sigma_{ss}$ satisfies the specification $\psi$ under the strategy $\varpi$, 
     $k$ is the number of steps, $n_Q$ is the number of elements of $Q$, $\delta$ is the state discretisation parameter, $B_L$ is the asymptotic upper bound of LC, and $\mathfrak{L}$ is the Lebesgue measure of the specification set.
\end{theorem}

\subsection{IMDP Based on Non-Parametric Estimation}

In this section, NPE is used to construct the IMDP from data. The upper and lower bounds of the transition probability from $q_{i}\in Q$ to $q_j\in Q$, $i,j\in\{1,\ldots,n_Q\}$, can be represented as  
\begin{align*}
    &P_{lo}(q_{i},q_{j})=\min_{x\in q_{i}} \int_{q_{j}} \hat{f}_{Y|X}(y,x) dy, \\
    &P_{up}(q_{i},q_{j})=\max_{x\in q_{i}} \int_{q_{j}} \hat{f}_{Y|X}(y,x) dy,
\end{align*}
where $\hat{f}_{Y|X}$ is the estimator of CoDF, shown in equation \eqref{condi_densityestima}.
Then, we can conduct formal verification against a PCTL path formula using NPE.
In case the system has control input $a$, $\hat{f}_{Y|X}$ must be computed for each value of the input, thus $P_{lo}$ and $P_{up}$ will also depend on $a$.
It should be pointed out that the reliability of this method can be assessed based on its statistical properties, such as variance, bias, and AMISE \citep{hardle2004nonparametric,scott2015multivariate}.

\section{CASE STUDIES}\label{case_sec}
\begin{exmp}[Verification]\label{verif_case1}
Consider an unknown linear stochastic system
$
X(k+1)=AX(k)+W(k)
$
with noise $W(\cdot)\sim \mathcal{N}(\mu,\Sigma)$. 
Its CoDF and relevant parameters are shown in Appendix~\ref{supp_verif_case1}. Labels $D$ and $O$ are the destination and avoiding regions, respectively.
The PCTL formula $\psi=\neg O \mathcal{U}^{\leq 3} D$ requires that the system does not visit $O$ until visiting D in $3$ steps. 
We assume that the upper bounds of the third derivatives of the system's CoDF is $0.2$. We select $\epsilon_g=0.2$ for the empirical approach and data scale $n=2000$ for the NPE. Based on Algorithm~\ref{algo_enviro} and Theorem~\ref{LC_main}, the asymptotic upper bound of the LC is $0.0722$. 
Thus, we can determine the state discretisation parameter $\delta = 0.1$ that ensures the distance between the satisfaction probabilities of the specification on the system and its finite abstraction is less than $0.1$.

\begin{figure}[hbt!]
\vspace{-1.3cm}
\centerline{\includegraphics[width=0.5\textwidth]{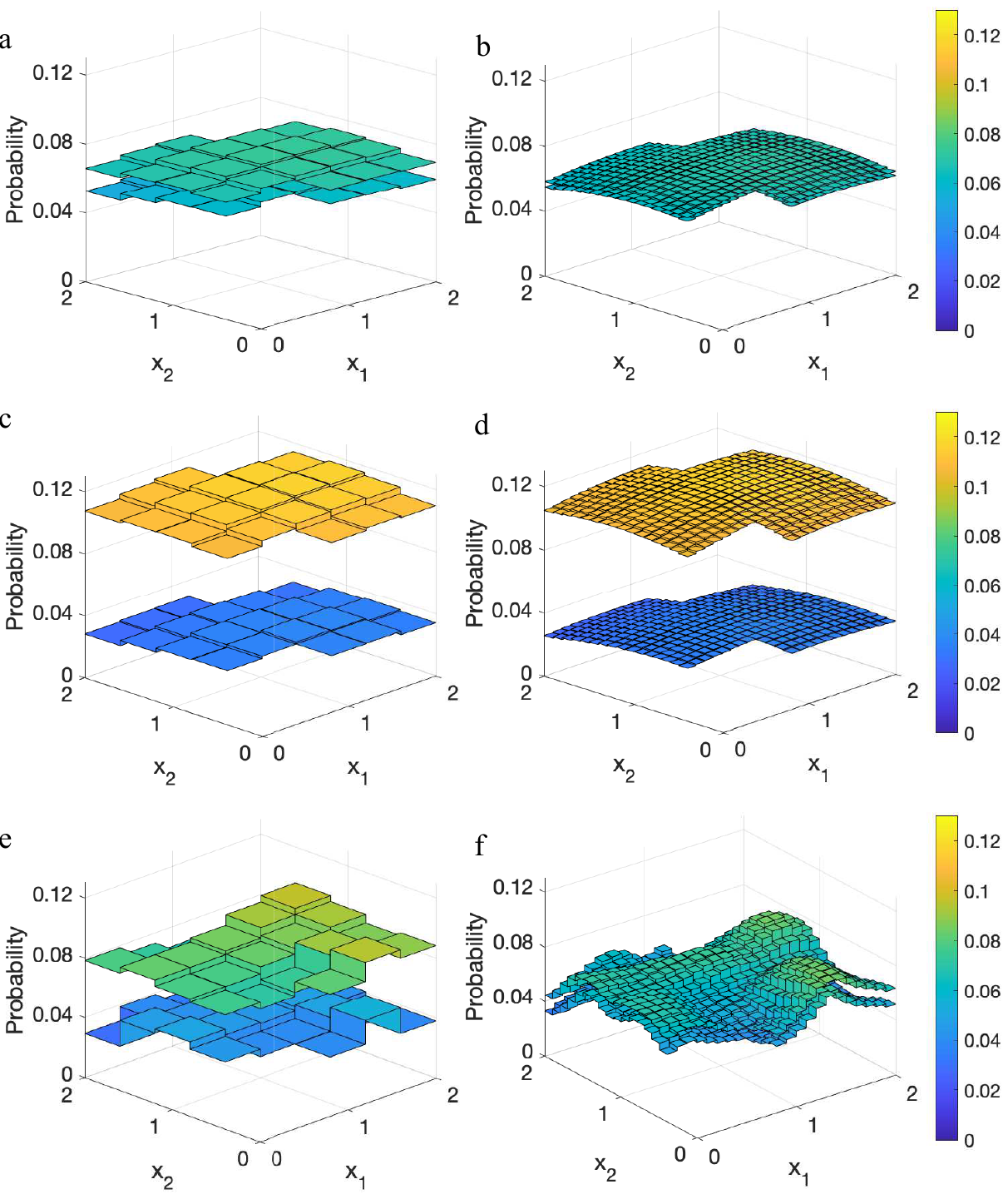}}
\vspace{-.7cm}
\caption{The upper and lower bounds on the probability of satisfying the specification by the linear system with $\delta=0.4$ in the left and $\delta=0.1$ in the right, and $O=[1.2,2]\times [1.6,2]$ and $D=[0,0.8]\times [0,0.4]$. The panels (a) and (b) show the results from a model-based approach. The panels (c) and (d) show the results of the data-driven approximation using the empirical approach. The panels (e) and (f) are for NPE.}
\label{no_strategy}
\vspace{-0.3cm}
\end{figure}

We apply the results under two different state discretisation parameters as shown in Fig.~\ref{no_strategy} with $\delta=0.4$ in the left and $\delta=0.1$ in the right. The top panels are the results of the model-based approach, the middle panels show the results of the empirical approach, and the bottom panels are for the NPE approach.
The distance between the upper and lower bounds of the satisfaction probabilities decreases with smaller $\delta$ for the model-based and the NPE approaches, but does not change significantly for the empirical approach.
The main reason is that we use the approach in Lemma~\ref{truevalue_distan} to allocate a value to $\bar{\epsilon}$. This phenomenon can be overcome using the approach in Lemma~\ref{truevalue_distan}, while relying on more data for obtaining $\Bar{P}_{ij}$ in equation~\eqref{ineq_empircalappro}.
%
Meanwhile, for NPE, the probabilities in Fig.~\ref{no_strategy}f are close to the results of the model-based approach in  Fig.~\ref{no_strategy}b and exhibit greater variation as a function of state than the results in Fig.~\ref{no_strategy}e.
\end{exmp}

\begin{exmp}[Synthesis]\label{verif_case2}
Consider an unknown switched system with two actions $S_{\mathfrak a}=\{a_1,a_2\}$, dynamics $X(k+1)=A_{i}X(k)+W(k)$ for each action $a_{i}$, $i=1,2$,
and noise $W\sim \mathcal{N}(\mu,\Sigma)$. Its CoDF and relevant parameters are shown in Appendix~\ref{supp_verif_case2}. 
In this example, we also consider the same specification $\psi$ as in the previous example. We assume that the upper bounds of the third derivatives of its CoDF are $0.2$. Select $\epsilon_g=0.2$ for the empirical approach and data scale $n=2000$ for the NPE. The asymptotic upper bound of the LC is $0.0918$. Thus, we can determine the state discretisation parameter $\delta = 0.1$ that ensures the distance between the satisfaction probabilities of the specification on the system and its finite abstraction is less than $0.1$.

\begin{figure}[hbt!]
\vspace{-1.3cm}
\centerline{\includegraphics[width=0.5\textwidth]{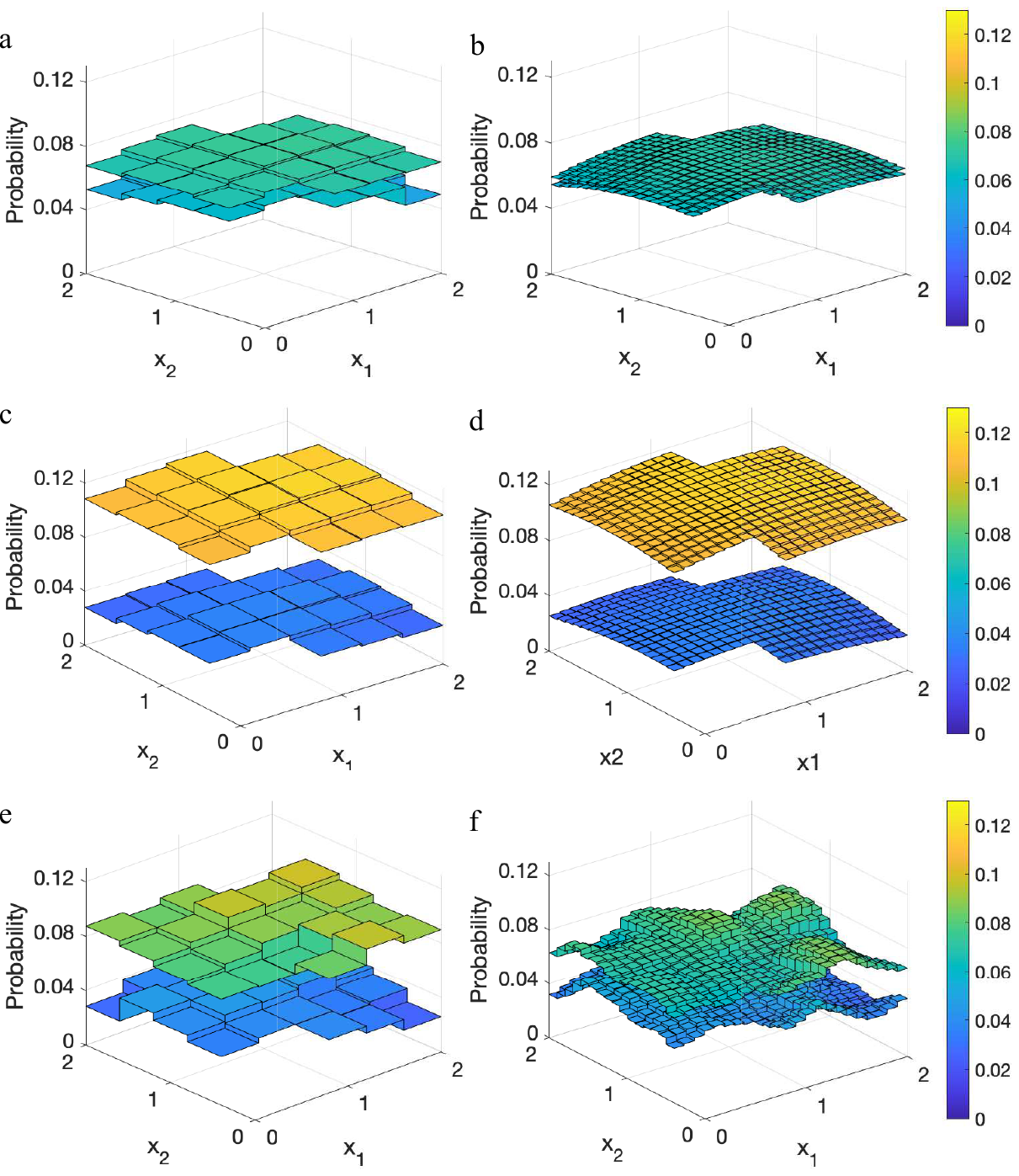}}
\vspace{-.7cm}
\caption{
The upper and lower bounds on the probability of satisfying the specification by the switched system with $\delta=0.4$ in the left and $\delta=0.1$ in the right, and $O=[1.2,2]\times [1.6,2]$ and $D=[0,0.8]\times [0,0.4]$. The panels (a) and (b) show the results from a model-based approach. The panels (c) and (d) show the results of the data-driven approximation using the empirical approach. The panels (e) and (f) are for NPE.
}
\label{strategy_probability}
\end{figure}

We apply the results under two different state discretisation parameters as shown in Fig.~\ref{strategy_probability} with $\delta=0.4$ in the left and $\delta=0.1$ in the right. The top panels are the results of the model-based approach, the middle panels show the results of the empirical approach, and the bottom panels are for the NPE approach.
The distance between the upper and lower bounds of the satisfaction probabilities decreases with smaller $\delta$ for the model-based and the NPE approaches, but does not change significantly for the empirical approach.
This can be overcome by relying on more data for obtaining $\Bar{P}_{ij}$ in equation~\eqref{ineq_empircalappro}.
In addition, the details of how to select an action at each state are presented in Figs.~\ref{allaction04}--\ref{allaction01}, as shown in Appendix~\ref{supp_verif_case2}.

\end{exmp}

\section{CONCLUSION}\label{concl_sec}
In this paper, we proposed a data-driven approach to perform verification of unknown stochastic systems with a guaranteed closeness relying on the asymptotic upper bound of the Lipschitz constant (LC) estimation. We provided theoretical results on quantifying the asymptotic bias of the LC estimation, and showed that the provided upper bound converges to the actual LC. This bound determines the partitioning size of the continuous space for building a finite abstraction and gives the distance between satisfaction probabilities of the specification on the original system and its abstraction.
Then, two data-driven methods were used to construct interval Markov decision processes containing the unknown finite abstraction. 
The effectiveness of the proposed data-driven framework was validated through formal verification and synthesis against PCTL specifications.
In the future, we plan to integrate our estimation method with discretisation-free approaches for  control synthesis to design control policies for complex unknown systems and satisfy high-level temporal requirements.

\subsubsection*{Acknowledgements}
The research was supported by the following grants: EPSRC EP/V043676/1, EIC 101070802, and ERC 101089047.

\bibliography{aistats2024}
\bibliographystyle{apalike}

\newpage



\clearpage
\appendix

\thispagestyle{empty}
\onecolumn
\aistatstitle{Formal Verification of Unknown Stochastic Systems via Non-parametric Estimation: Supplementary Material}



\section{SUPPLEMENTARY OF SECTION~\ref{sec2}}\label{prel_appen}
\subsection{Supplementary of Section~\ref{sec:prelem}}\label{kernel_bandwidth}
\paragraph{Choice of the Kernel Function.}
For two different kernel functions, in practice it is possible to get approximately the same degree of smoothness by multiplying one of the bandwidths with an adjustment factor \citep{hardle2004nonparametric}. This adjustment factor for two kernels $A$ and $B$ can be computed from $h_B=h_A\delta_0^B/\delta_0^A$ where $\delta_0$ is the canonical bandwidth. Table~\ref{tab:kernel} gives commonly used kernels and their canonical bandwidths. For the selection of the kernel, the conservative recommendation is the kernel which is smooth, clearly unimodal, and symmetric 
around the origin. Also, some factors (e.g., ease of computation and differentiability) should be considered, rather than concerning the loss of efficiency.
\begin{table*}[hbt!]
\centering
\caption{Commonly used kernels and their canonical bandwidths}
\begin{tabular}[t]{lcc}
\hline Kernel & Equation & Canonical Bandwidths \\
\hline
Uniform & $\frac{1}{2}I(|u|\leq 1)$ & 1.3510 \\
Triangle & $(1-|u|)I(|u|\leq 1)$ & 1.8890\\
Epanechnikov & $\frac{3}{4}(1-u^2)I(|u|\leq 1)$ & 1.7188 \\
Quartic (Biweight) & $\frac{15}{16}(1-u^2)^2I(|u|\leq 1)$ & 2.0362 \\
Triweight & $\frac{35}{32}(1-u^2)^3I(|u|\leq 1)$ & 2.3122 \\ 

Gaussian & $k(u)=\frac{1}{\sqrt{2\pi}}exp(-u^2/2)$& 0.7764 \\
\hline
\end{tabular}
\label{tab:kernel}
\end{table*}
\vspace{-0.5cm}

\paragraph{Related Works on the Selection of Bandwidth.}
There are several data-driven methods for choosing the optimal bandwidth. 
For example,
\cite{hardle2004nonparametric} proposed using Silverman’s rule-of-thumb bandwidth for unimodal distributions that are fairly symmetric and are not heavy-tailed, and using the cross-validation method which is fairly independent of the special structure of the parameter or function estimate.
\cite{tsybakov2009introduction} have employed the cross-validation method to choose the ideal value of the bandwidth and then constructed unbiased risk estimators using the Fourier analysis of density estimators.
\cite{sheather2004density} have provided a practical description of kernel density estimation methods and compared the performance of three methods for selecting the value of the bandwidth, including Rules of Thumb, Cross-Validation, and Plug-in Methods. 
\cite{wand1992error} have demonstrated numerical minimisation of the AMISE for general $H$, which can be used as a data-driven method for choosing the optimal bandwidth using a plug-in approach.
Several papers have recommended to construct a family of density estimates based on a number of values of the bandwidth \citep{marron2001presentation,scott2015multivariate}.
In general, the above methods for selection of the bandwidth is with respect to keeping the balance between the bias and the variance to avoid under-smoothing and over-smoothing scenarios.

\vspace{-0.1cm}
\paragraph{Scott's Formula for Selecting the Bandwidth.}
Scott's formula provides a method for selecting the bandwidth of the estimator for a normal distribution with covariance matrix $\Sigma\!=\!diag(\sigma_1^2,\ldots,\sigma_d^2)$, and ensures the optimal convergence rate $O(n^{-\frac{4}{4+d}})$ for the AMISE \citep{hardle2004nonparametric}.
The optimal bandwidth is 
$
H=n^{-1/(d+4)}\hat{\Sigma}^{1/2},
$
with
$\hat\Sigma    :=\frac{1}{n}\sum^{n}_{i=1}(\hat{X}_{i}-\bar{X})^{2}$ and $\bar{X}:=\frac{1}{n}\sum^{n}_{i=1}\hat{X}_{i}$.

\vspace{-0.1cm}

\paragraph{Cross-Validation.}
The cross-validation (CV) method finds the best bandwidth by minimising an unbiased or approximately unbiased estimator of MISE instead of minimising MISE. The CV selects the optimal bandwidth $H_{CV}$ by performing the minimisation 
\begin{align}
\label{cross-valid}
H_{CV}=&\argmin_{H>0}  CV(H), \text{ with }\\
CV(H) =&
\frac{1}{n^2 |H|}
\sum_{i=1}^{n}\sum_{j=1}^{n}
K\star K\{H^{-1}(\hat{X}_j-\hat{X}_i)\}-\frac{2}{n(n-1)} \sum_{i=1}^{n}\sum_{j=1\,j\ne i}^{n} K\{H^{-1}(\hat{X}_j-\hat{X}_i)\},\nonumber
\end{align}
where $K\!\star\! K(\bs{u})\!:=\!\int\! K(\bs{u}-\bs{v})K(\bs{v})d\bs{v}$.

The asymptotic bias and variance of the univariate conditional density estimator~\eqref{condi_densityestima} are as follows
\begin{align}\label{bias_from1996}
    \Bias[\hat{f}_{Y|X}(y,x)]=&\frac{h_{\mathsf x}^{2}G_{12}^{2}(K)}{2}\left\{ 2\frac{f'_{X}(x)}{f_{X}(x)}\frac{\mathrm{d}}{\mathrm{d}x}f_{Y|X}(y,x)+\frac{\mathrm{d}^{2}}{\mathrm{d}^{2}x}f_{Y|X}(y,x)
    +\frac{h_{\mathsf y}^{2}}{h_{\mathsf x}^{2}}\frac{\mathrm{d}^{2}}{\mathrm{d}y^{2}}f_{Y|X}(y,x)\right\}\nonumber\\
    &+O(h_{\mathsf x}^4)+O(h_{\mathsf y}^4)+O(h_{\mathsf x}^2h_{\mathsf y}^2)+O(\frac{1}{nh_{\mathsf x}}),
\end{align}
and
\begin{equation}\label{var_from1996}
    \Var[\hat{f}_{Y|X}(y,x)]
    =\frac{G_{20}(K)f_{Y|X}(y,x)}{nh_{\mathsf y}h_{\mathsf x}f_{X}(x)}[G_{20}(K)-h_{\mathsf y}f_{Y|X}(y,x)]+O(\frac{1}{n})+O(\frac{h_{\mathsf y}}{nh_{\mathsf x}})+O(\frac{h_{\mathsf x}}{nh_{\mathsf y}}),
\end{equation}
where $G_{12}(K)=\int u^{2}K(u)du$ and $G_{20}(K)=\int K^{2}(u)du$, if $h_{\mathsf x},~h_{\mathsf y}\to 0$ and $n\to +\infty$.

\subsection{Review of Probabilistic Computation Tree Logic (PCTL)}\label{PCTL}

PCTL \citep{baier2008principles} is a formal language for expressing the requirements on complex behaviours of stochastic systems.

\begin{definition}[\textbf{Syntax of PCTL}]
For a given set of atomic propositions $AP$, formulas in PCTL can be recursively defined as follows:
\begin{align*}
    &\text{State Formula } \phi:=\textsf{true}~|~\rho~|~\neg \phi~|~\phi \wedge \phi~|~ P_{\bowtie p}[\psi],\\
    &\text{Path Formula }~ \psi:= \mathcal{X} \phi~|~\phi~\mathcal{U}^{\leq k}~\phi~|~ \phi~\mathcal{U}~ \phi, 
\end{align*}
where $\rho \in AP$, $\neg$ is the negation operator, $\wedge$ is the conjunction operator, $P_{\bowtie p}$ is the probabilistic operator, $\bowtie\in \{\leq, <, \ge,>\} $ is a relation placeholder, and $p\in [0,1]$. $\mathcal{X}$ (next), $\mathcal{U}^{\leq k}$ (bounded until), and $\mathcal{U}$ (until) are temporal operators.
\end{definition}

\begin{definition}[\textbf{PCTL Semantics}]
For a labelling function $L:Q\rightarrow 2^{AP}$, the satisfaction relation $\vDash$ is defined inductively as follows. For any state $q\in Q$,
(1) $q\vDash\textsf{true}$ for all $q\in Q$;
(2) $q\vDash \rho \iff \rho \in L(q)$;
(3) $q\vDash (\phi_1 \wedge \phi_2) \iff (q\vDash \phi_1) \wedge (q\vDash \phi _2)$;
(4) $q\vDash \neg \phi \iff q\nvDash \phi$;
(5) $q\vDash P_{\bowtie p} [\psi] \iff \textsf{Prob}_{q}(\psi)\bowtie p,$ where $\textsf{Prob}_{q}(\psi)$ is the probability that infinite trajectories initialised at $q$ satisfy $\psi$.
Also, for any path $\upsilon\in \textsf{Paths}$, the satisfaction relation $\vDash$ is defined as: (1) $\upsilon \vDash \mathcal{X}\phi \iff \upsilon(1)\vDash \phi;$ (2) $\upsilon \vDash \phi_1 \mathcal{U}^{\leq k}\phi_2 \iff  \exists i\leq k~s.t.~\upsilon(i)\vDash \phi_2 \wedge \upsilon(j)\vDash \phi_1,~\forall~j\in [0,i);$ (3) $\upsilon\vDash \phi_1\mathcal{U}\phi_2 \iff \exists i\ge 0~s.t.~\upsilon(i)\vDash\phi_2 \wedge \upsilon(j)\vDash \phi_1,~\forall~j\in [0,i).$
\end{definition}

$\Diamond^{\leq k}$ (bounded eventually) and $\Diamond$ (eventually) are defined as $P_{\bowtie p}[\Diamond^{\leq k }\phi ]\equiv P_{\bowtie p}[\textsf{true}~ \mathcal{U}^{\leq k} \phi]$ and $P_{\bowtie p}[\Diamond \phi]\equiv P_{\bowtie p}[\textsf{true}~\mathcal{U} \phi]$ representing $\phi$ is satisfied within $k$ time steps and $\phi$ is satisfied at some point in the future, respectively.

\section{SUPPLEMENTARY OF SECTION~\ref{main_sec}}

\subsection{Probabilistic Closeness Guarantee Between DTSCS and its Finite Abstraction}\label{Prob_guarantee}

Here, we approximate a DTSCS $\Sigma_{ss}$ with a finite MDP $\hat \Sigma_{ss}=(Q, S_{\mathfrak a}, P, AP, L )$ with $Q$ representing a partition of the state space $\mathcal S$ with partition sets denoted by $q\in Q$. The input space $S_{\mathfrak a} = U$. For the transition probabilities $P^{a}_{ij}$, select representative points $\bar q\in q$ for each partition set. Define $P^{a}_{ij} :=\textsf{Prob}_w(f(\bar q_i,a,w)\in q_j)$.
%
%
Define the state discretisation parameter $\delta:=\sup\{ \| x-x'\|,x,x'\in q,\,\, q\in Q\}$. 
DTSCS $\Sigma_{ss}$ and its finite MDP abstraction $\hat{\Sigma}_{ss}$ under any strategy $\hat{\varpi}(\cdot)\in \mathcal{U}_{\mathfrak a}$ are denoted as $\Sigma^{\hat{\varpi}}_{ss}$ and $\hat{\Sigma}^{\hat{\varpi}}_{ss}$, respectively.
The following theorem \citep{SA13} provides the closeness guarantee between $\Sigma_{ss}$ and its finite abstraction $\hat{\Sigma}_{ss}$.

\begin{theorem}\label{SA13_bound}
      For a given PCTL specification $\psi$ over a finite horizon and any strategy $\hat{\varpi}(\cdot)\in \mathcal{U}_{\mathfrak a}$, the closeness between $\Sigma^{\hat{\varpi}}_{ss}$ and $\hat{\Sigma}^{\hat{\varpi}}_{ss}$ can be obtained as 
\begin{align*}
    |P(\Sigma^{\hat{\varpi}}_{ss}\vDash \psi)-P(\hat{\Sigma}^{\hat{\varpi}}_{ss}\vDash \psi) |\leq \epsilon, \text{ with $\epsilon:=T\delta L \mathfrak{L}$,}
\end{align*}
where $T$ is the finite time horizon, $\delta$ is the state discretisation parameter, $L$ is the Lipschitz constant of the stochastic kernel, and $\mathfrak{L}$ is the Lebesgue measure of the specification set. 

\end{theorem}

\begin{remark}
The upper bound $L$ of the LC impacts the algorithm as follows: One can initially fix the desired threshold $\epsilon$ in advance, and then select the partition parameter $\delta=\frac{\epsilon}{T L \mathfrak{L}}$ according to the values of $T$, $L$, $\mathfrak{L}$.
This partition provides a guarantee for the verification based on MDP abstraction, which ensures the absolute distance between the satisfaction probability of the original system and that of its finite MDP abstraction is smaller than $\epsilon$.
\end{remark}

\subsection{Preparation for Main Result}

In order to obtain the main results in Section~\ref{theo_onedimen}, we introduce the following lemmas.

\begin{lemma}\label{exp_main}
Suppose that one-dimensional random variable $X$ has density function $f_{X}(x)$, $K:\mathbb{R} \to \mathbb{R}$ is the Gaussian kernel function with bandwidth $h_{\mathsf x}$, $\theta(x)$ is at least twice continuously differentiable, $\hat{X}$ is the random sample of $X$, and $u_{x}:=\frac{\hat{X}-x}{h_{\mathsf x}}$. We have that if $h_{\mathsf x}\to 0$, then
\begin{align*}
    \E[K^{2}(u_{x})\theta(\hat{X})]=h_{\mathsf x} G_{20}(K)\theta(x)f_{X}(x)
    +\frac{1}{2}h_{\mathsf x}^{3} G_{22}(K)\frac{\mathrm{d}^{2}}{\mathrm{d}x^{2}}[\theta(x)f_{X}(x)]+O(h_{\mathsf x}^{5}),
\end{align*}
\begin{align*}
    \E[u^{2}_{x}K^{2}(u_{x})\theta(\hat{X})]
    =h_{\mathsf x} G_{22}(K)\theta(x)f_{X}(x)+\frac{1}{2}h_{\mathsf x}^{3}G_{24}(K)\frac{\mathrm{d}^{2}}{\mathrm{d}x^{2}}[\theta(x)f_{X}(x)]+O(h_{\mathsf x}^{5}),
\end{align*}
\begin{align*}
    \E[K(u_{x})\theta(\hat{X})]
    =h_{\mathsf x} \theta(x)f_{X}(x)+\frac{1}{2}h_{\mathsf x}^{3}G_{12}(K)\frac{\mathrm{d}^2}{\mathrm{d}x^2}[\theta(x)f_{X}(x)]+O(h_{\mathsf x}^{5}),
\end{align*}
\begin{align*}
    \E[u^{2}_{x}K^{4}(u_{x})\theta(\hat{X})]
    =h_{\mathsf x} \theta(x)f_{X}(x)G_{42}(K)+\frac{1}{2}h_{\mathsf x}^{3}G_{44}(K)\frac{\mathrm{d}^2}{\mathrm{d}x^2}[\theta(x)f_{X}(x)]
    +O(h_{\mathsf x}^{5}),
\end{align*}
\begin{align*}
    \E[u^{2}_{x}K^{3}(u_{x})\theta(\hat{X})]
    =&h_{\mathsf x} \theta(x)f_{X}(x)G_{32}(K)+\frac{1}{2}h_{\mathsf x}^{3}G_{34}(K)\frac{\mathrm{d}^2}{\mathrm{d}x^2}[\theta(x)f_{X}(x)]+O(h_{\mathsf x}^{5}),
\end{align*}
\begin{align*}
    \E[K^{4}(u_{x})\theta(\hat{X})]
    =h_{\mathsf x} \theta(x)f_{X}(x)G_{40}(K)+\frac{1}{2}h_{\mathsf x}^{3}G_{42}(K)\frac{\mathrm{d}^2}{\mathrm{d}x^2}[\theta(x)f_{X}(x)]
    +O(h_{\mathsf x}^{5}),
\end{align*}
\begin{align*}
    \E[u^{4}_{x}K^{6}(u_{x})\theta(\hat{X})]
    =h_{\mathsf x} \theta(x)f_{X}(x)G_{64}(K)+\frac{1}{2}h_{\mathsf x}^{3}G_{66}(K)\frac{\mathrm{d}^2}{\mathrm{d}x^2}[\theta(x)f_{X}(x)]+O(h_{\mathsf x}^5)
\end{align*}
\begin{align*}
    \E[u_{x}K(u_{x})\theta(\hat{X})]
    =&h_{\mathsf x}^{2}G_{12}(K)\frac{\mathrm{d}}{\mathrm{d}x}[\theta(x)f_{X}(x)]\nonumber\\
    &+\frac{1}{2}h_{\mathsf x}^{4}G_{14}(K)[\frac{\mathrm{d}}{\mathrm{d}x}\theta(x)\frac{\mathrm{d}^2}{\mathrm{d}x^2}f_{X}(x)+\frac{\mathrm{d}^{2}}{\mathrm{d}x^{2}}\theta(x)\frac{\mathrm{d}}{\mathrm{d}x}f_{X}(x)]+O(h_{\mathsf x}^{6}),
\end{align*}
\begin{align*}
    \E[u_{x}K^{2}(u_{x})\theta(\hat{X})]
    =&h_{\mathsf x}^{2}G_{22}(K)\frac{\mathrm{d}}{\mathrm{d}x}[\theta(x)f_{X}(x)]\nonumber\\
    &+\frac{1}{2}h_{\mathsf x}^{4}G_{24}(K)[\frac{\mathrm{d}}{\mathrm{d}x}\theta(x)\frac{\mathrm{d}^2}{\mathrm{d}x^2}f_{X}(x)+\frac{\mathrm{d}^{2}}{\mathrm{d}x^{2}}\theta(x)\frac{\mathrm{d}}{\mathrm{d}x}f_{X}(x)]+O(h_{\mathsf x}^{6}),
\end{align*}
\begin{align*}
    \E[u_{x}K^{3}(u_{x})\theta(\hat{X})]
    =&h_{\mathsf x}^{2}G_{32}(K)\frac{\mathrm{d}}{\mathrm{d}x}[\theta(x)f_{X}(x)]\nonumber\\
    &+\frac{1}{2}h_{\mathsf x}^{4}G_{34}(K)[\frac{\mathrm{d}}{\mathrm{d}x}\theta(x)\frac{\mathrm{d}^2}{\mathrm{d}x^2}f_{X}(x)+\frac{\mathrm{d}^{2}}{\mathrm{d}x^{2}}\theta(x)\frac{\mathrm{d}}{\mathrm{d}x}f_{X}(x)]
    +O(h_{\mathsf x}^{6}),
\end{align*}
and 
\begin{align*}
    \E[u_{x}K^{4}(u_{x})\theta(\hat{X})]
    =&h_{\mathsf x}^{2}G_{42}(K)\frac{\mathrm{d}}{\mathrm{d}x}[\theta(x)f_{X}(x)]\nonumber\\
    &+\frac{1}{2}h_{\mathsf x}^{4}G_{44}(K)[\frac{\mathrm{d}}{\mathrm{d}x}\theta(x)\frac{\mathrm{d}^2}{\mathrm{d}x^2}f_{X}(x)+\frac{\mathrm{d}^{2}}{\mathrm{d}x^{2}}\theta(x)\frac{\mathrm{d}}{\mathrm{d}x}f_{X}(x)]
    +O(h_{\mathsf x}^{6}),
\end{align*}
where 
$G_{ji}=\int \nu^{i}K^{j}(\nu)d\nu$, $j\in \{1,\dots,6\}$ and $i\in \{0,\dots,6\}$.
\end{lemma}
\begin{proof}
The first equation can be obtained as follows
\begin{align*}
    &\E[K^{2}(u_{x})\theta(\hat{X})]\\
    =&\int K^{2}(\frac{\mu-x}{h_{\mathsf x}})\theta(\mu)f_{X}(\mu)\mathrm{d} \mu\\
    =&\int h_{\mathsf x} K^{2}(z)[\theta(x)+zh_{\mathsf x} \frac{\mathrm{d}}{\mathrm{d}x}\theta(x)+\frac{1}{2}z^{2}h_{\mathsf x}^{2}\frac{\mathrm{d}^2}{\mathrm{d}x^2}\theta(x)+\frac{1}{6}z^{3}h_{\mathsf x}^{3}\frac{\mathrm{d}^3}{\mathrm{d}x^3}\theta(x)+O(h_{\mathsf x}^{4})][f_{X}(x)\\
    &+zh_{\mathsf x}\frac{\mathrm{d}}{\mathrm{d}x}f_{X}(x)
    +\frac{1}{2}z^{2}h_{\mathsf x}^{2}\frac{\mathrm{d}^2}{\mathrm{d}x^2}f_{X}(x)+\frac{1}{6}z^{3}h_{\mathsf x}^{3}\frac{\mathrm{d}^2}{\mathrm{d}x^3}f_{X}(x)+O(h_{\mathsf x}^{4})]\mathrm{d} z\\
    &\text{where $z=\frac{\mu-x}{h_{\mathsf x}}$ and using Taylor series expansion,}\\
    =&\int h_{\mathsf x} K^{2}(z)[\theta(x)f_{X}(x)+zh_{\mathsf x} \theta(x)\frac{\mathrm{d}}{\mathrm{d}x}f_{X}(x)+\frac{1}{2}z^{2}h_{\mathsf x}^{2}\theta(x)\frac{\mathrm{d}^2}{\mathrm{d}x^2}f_{X}(x)+\frac{1}{6}z^3 h_{\mathsf x}^3 \theta(x)\frac{\mathrm{d}^3}{\mathrm{d}x^3}f_{X}(x)\\
    &+zh_{\mathsf x} f_{X}(x)\frac{\mathrm{d}}{\mathrm{d}x}\theta(x)+z^{2}h_{\mathsf x}^{2}\frac{\mathrm{d}}{\mathrm{d}x}\theta(x)\frac{\mathrm{d}}{\mathrm{d}x}f_{X}(x)+\frac{1}{2}z^{3}h_{\mathsf x}^{3}\frac{\mathrm{d}}{\mathrm{d}x}\theta(x)\frac{\mathrm{d}^2}{\mathrm{d}x^2}f_{X}(x)\\
    &+\frac{1}{2}z^{2}h_{\mathsf x}^{2}f_{X}(x)\frac{\mathrm{d}^2}{\mathrm{d}x^2}\theta(x)+\frac{1}{2}z^{3}h_{\mathsf x}^{3}\frac{\mathrm{d}^2}{\mathrm{d}x^2}\theta(x)\frac{\mathrm{d}}{\mathrm{d}x}f_{X}(x)+\frac{1}{6}z^3 h_{\mathsf x}^3f_{X}(x)\frac{\mathrm{d}^3}{\mathrm{d}x^3}\theta(x)+O(h_{\mathsf x}^{4})]\mathrm{d} z\\
    =&h_{\mathsf x} G_{20}(K) \theta(x)f_{X}(x)+\frac{1}{2}h_{\mathsf x}^{3} G_{22}(K)\frac{\mathrm{d}^{2}}{\mathrm{d}x^{2}}[\theta(x)f_{X}(x)]+O(h_{\mathsf x}^{5}).
\end{align*}
Using similar derivations, the rest of the results can be obtained.
\end{proof}

The following lemma can be derived based on Lemma~\ref{exp_main}.

\begin{lemma}\label{element_expect}
Suppose that $X$ and $Y$ are one-dimensional random variables, there is a CoDF $f_{Y|X}(y,x)$ for all $(x,y)\in D_X\times D_Y$, $X$ is from the uniform distribution with density function $f_{X}(x)$, and $K:\mathbb{R} \to \mathbb{R}$ is the Gaussian kernel function with bandwidths $h_{\mathsf x}$ and $h_{\mathsf y}$. In addition, samples $\{\hat{X}_{i},~i=1,\ldots,n \}$ are selected uniformly from $D_{X}$ with density function $f_{X}(x)$, and for each $\hat{X}_{i}$, sample $\hat{Y}_{i}$ is generated from $(Y|\hat{X}_{i})$ with the CoDF $f_{Y|X}(y,x)$, $i\in \{1,\ldots,n \}$. Denoting $u_{x_{i}}:=\frac{\hat{X}_{i}-x}{h_{\mathsf x}}$ and $u_{y_{i}}:=\frac{\hat{Y}_{i}-y}{h_{\mathsf y}}$, $i\in \{ 1,\ldots,n\}$.
For $\hat{X}_{i}$ and $\hat{Y}_{i}$, $i\in \{1,\ldots,n\}$, we have that if $h_{\mathsf x},~h_{\mathsf y}\to 0$ as $n\to+\infty$, then
\begin{align*}
    \E[u_{x_{i}}K(u_{x_{i}})K(u_{y_{i}})]\nonumber
    =&h_{\mathsf x}^{2}h_{\mathsf y} G_{12}(K)f_{X}(x) \frac{\mathrm{d}}{\mathrm{d}x}f_{Y|X}(y,x)\\
    &+\frac{1}{2}h_{\mathsf x}^{2}h_{\mathsf y}^{3}G^{2}_{12}(K)f_{X}(x)\frac{\mathrm{d}^{3}}{\mathrm{d}x\mathrm{d}y^{2}}f_{Y|X}(y,x)
    +O(h_{\mathsf x}^{2}h_{\mathsf xy}^{5})+O(h_{\mathsf x}^{6}),
\end{align*}
\begin{align*}
    \E[K(u_{x_{i}})K(u_{y_{i}})]
    =&h_{\mathsf x} h_{\mathsf y} f_{X}(x)f_{Y|X}(y,x)+\frac{1}{2}h_{\mathsf x} h_{\mathsf y}^{3}G_{12}(K)f_{X}(x)[ \frac{\mathrm{d}^{2}}{\mathrm{d}y^{2}}f_{Y|X}(y,x)\\
    &+\frac{h_{\mathsf x}^{2}}{h_{\mathsf y}^{2}}\frac{\mathrm{d}^{2}}{\mathrm{d}x^{2}}f_{Y|X}(y,x) ]
    +O(h_{\mathsf x}^{5})+O(h_{\mathsf x}^{3}h_{\mathsf y}^{3})
    +O(h_{\mathsf x} h_{\mathsf y}^{5}),
\end{align*}
\begin{align*}
    \E[K^{2}(u_{x_{i}})K^{2}(u_{y_{i}})]
    =&h_{\mathsf x} h_{\mathsf y}G^{2}_{20}(K)f_{X}(x)f_{Y|X}(y,x)\\
    &+\frac{1}{2}h_{\mathsf x} h_{\mathsf y}^{3}G_{20}(K)G_{22}(K)f_{X}(x)[ \frac{\mathrm{d}^{2}}{\mathrm{d}y^{2}}f_{Y|X}(y,x)+\frac{h_{\mathsf x}^{2}}{h_{\mathsf y}^{2}}\frac{\mathrm{d}^{2}}{\mathrm{d}x^{2}}f_{Y|X}(y,x) ]\\
    &+O(h_{\mathsf x} h_{\mathsf y}^{5})+O(h_{\mathsf x}^{5}),
\end{align*}
\begin{align*}
    \E[u^{2}_{x_{i}}K^{2}(u_{x_{i}})K^{2}(u_{y_{i}})]
    =&h_{\mathsf x} h_{\mathsf y}G_{20}(K)G_{22}(K)f_{X}(x)f_{Y|X}(y,x)\\
    &
    +\frac{1}{2}h_{\mathsf x} h_{\mathsf y}^{3}G^2_{22}(K)f_{X}(x) \frac{\mathrm{d}^{2}}{\mathrm{d}y^{2}}f_{Y|X}(y,x)\\
    &+\frac{1}{2}h^3_{\mathsf x} h_{\mathsf y}G_{20}(K)G_{24}(K)\frac{\mathrm{d}^{2}}{\mathrm{d}x^{2}}f_{Y|X}(y,x)
    +O(h_{\mathsf x}^{5})+O(h_{\mathsf x} h_{\mathsf y}^{5})\\
    &+O(h_{\mathsf x}^3 h_{\mathsf y}^{3}),
\end{align*}
\begin{align*}
    \E[u^{2}_{x_{i}}K^{4}(u_{x_{i}})K^{2}(u_{y_{i}})]
    =&h_{\mathsf x} h_{\mathsf y} G_{20}(K)G_{42}(K)f_{X}(x)f_{Y|X}(y,x)\\
    &+\frac{1}{2}h_{\mathsf x} h_{\mathsf y}^{3}G_{22}(K)G_{42}(K)f_{X}(x)\frac{\mathrm{d}^{2}}{\mathrm{d}y^{2}}f_{Y|X}(y,x)\nonumber\\
    &+\frac{1}{2}h_{\mathsf x}^{3} h_{\mathsf y} G_{20}(K)G_{44}(K)f_{X}(x)\frac{\mathrm{d}^{2}}{\mathrm{d}x^{2}}f_{Y|X}(y,x)
    +O(h_{\mathsf x}^{5})+O(h_{\mathsf x}^{3}h_{\mathsf y}^{3})\\
    &+O(h_{\mathsf x}h_{\mathsf y}^{5})\nonumber,
\end{align*}
\begin{align*}
    \E[u_{x_{i}}K^{2}(u_{x_{i}})K(u_{y_i})]
    =&h_{\mathsf x}^{2} h_{\mathsf y} G_{22}(K)f_{X}(x)\frac{\mathrm{d}}{\mathrm{d}x}f_{Y|X}(y,x)\\
    &+\frac{1}{2}h_{\mathsf x}^{2} h_{\mathsf y}^{3}G_{12}(K)G_{22}(K)f_{X}(x)\frac{\mathrm{d}^{3}}{\mathrm{d}x\mathrm{d}y^{2}}f_{Y|X}(y,x)
    +O(h_{\mathsf x}^{6})+O(h_{\mathsf x}^{2}h_{\mathsf y}^{5}),
\end{align*}
\begin{align*}
    \E[u^{2}_{x_{i}}K^{4}(u_{x_{i}})K^{4}(u_{y_{i}})]
    =&h_{\mathsf x} h_{\mathsf y} G_{40}(K)G_{42}(K)f_{X}(x)f_{Y|X}(y,x)
    +\frac{1}{2}h_{\mathsf x} h_{\mathsf y}^{3}G^{2}_{42}(K)f_{X}(x)\frac{\mathrm{d}^{3}}{\mathrm{d}x\mathrm{d}y^{2}}f_{Y|X}(y,x)\\
    &+\frac{1}{2}h_{\mathsf x}^{3}h_{\mathsf y} G_{40}(K)G_{44}(K)f_{X}(x)\frac{\mathrm{d}^{3}}{\mathrm{d}x\mathrm{d}x^{2}}f_{Y|X}(y,x)
    +O(h_{\mathsf x}^{5})+O(h_{\mathsf x} h_{\mathsf y}^{5})+O(h_{\mathsf x}^{3}h_{\mathsf y}^{3})\nonumber,
\end{align*}
\begin{align*}
    \E[u_{x_{i}}K^{2}(u_{x_{i}})K^{2}(u_{y_{i}})]
    =&h_{\mathsf x}^{2} h_{\mathsf y} G_{20}(K)G_{22}(K)f_{X}(x)\frac{\mathrm{d}}{\mathrm{d}x}f_{Y|X}(y,x)
    +\frac{1}{2}h_{\mathsf x}^{2}h_{\mathsf y}^{3}G^{2}_{22}(K)f_{X}(x)\frac{\mathrm{d}^{3}}{\mathrm{d}x\mathrm{d}y^{2}}f_{Y|X}(y,x)\\
    &+O(h_{\mathsf x}^{2}h_{\mathsf y}^{5})
    +O(h_{\mathsf x}^{6})
\end{align*}
\begin{align*}
    \E[u^{4}_{x_{i}}K^{6}(u_{x_{i}})K^{4}(u_{y_{i}})]
    =&h_{\mathsf x} h_{\mathsf y} G_{40}(K)G_{64}(K)f_{X}(x)f_{Y|X}(y,x)
    +\frac{1}{2}h_{\mathsf x} h_{\mathsf y}^{3}G_{42}(K)G_{64}(K)f_{X}(x)\frac{\mathrm{d}^{2}}{\mathrm{d}y^{2}}f_{Y|X}(y,x)\\
    &+\frac{1}{2}h_{\mathsf x}^{3}h_{\mathsf y} G_{40}(K)G_{66}(K)f_{X}(x)\frac{\mathrm{d}^{2}}{\mathrm{d}y^{2}}f_{Y|X}(y,x)+O(h_{\mathsf x}^{5})
    +O(h_{\mathsf x}^{3}h_{\mathsf y}^{3})+O(h_{\mathsf x}h_{\mathsf y}^{5}),
\end{align*}
\begin{align*}
    \E[u^{2}_{x_{i}}K^{3}(u_{x_{i}})K^{2}(u_{y_{i}})]
    =&h_{\mathsf x} h_{\mathsf y} G_{20}(K)G_{32}(K)f_{X}(x)f_{Y|X}(y,x)
    +\frac{1}{2}h_{\mathsf x} h_{\mathsf y}^{3}G_{22}(K)G_{32}(K)f_{X}(x)\frac{\mathrm{d}^{2}}{\mathrm{d}y^{2}}f_{Y|X}(y,x)\\
    &+\frac{1}{2}h_{\mathsf x}^{3}h_{\mathsf y} G_{20}(K)G_{34}(K)f_{X}(x)\frac{\mathrm{d}^{2}}{\mathrm{d}x^{2}}f_{Y|X}(y,x)
    +O(h_{\mathsf x}^{5})+O(h_{\mathsf x}h_{\mathsf y}^{5})
    +O(h_{\mathsf x}^{3}h_{\mathsf y}^{3}),
\end{align*}
\begin{align*}
    \E[u_{x_{i}}K^{3}(u_{x_{i}})K(u_{y_{i}})]
    =&h_{\mathsf x}^{2} h_{\mathsf y} G_{32}(K)f_{X}(x)\frac{\mathrm{d}}{\mathrm{d}x}f_{Y|X}(y,x)
    +\frac{1}{2}h_{\mathsf x}^{2} h_{\mathsf y}^{3}G_{12}(K)G_{32}(K)f_{X}(x)\frac{\mathrm{d}^{3}}{\mathrm{d}x\mathrm{d}y^{2}}f_{Y|X}(y,x)\\
    &+O(h_{\mathsf x}^{2}h_{\mathsf y}^{5})+O(h_{\mathsf x}^{6}),
\end{align*}
and
\begin{align*}
    \E[u_{x_{i}}K^{4}(u_{x_{i}})K(u_{y_{i}})]
    =&h_{\mathsf x}^{2} h_{\mathsf y} G_{42}(K)f_{X}(x)\frac{\mathrm{d}}{\mathrm{d}x}f_{Y|X}(y,x)
    +\frac{1}{2}h_{\mathsf x}^{2} h_{\mathsf y}^{3}G_{12}(K)G_{42}(K)f_{X}(x)\frac{\mathrm{d}^{3}}{\mathrm{d}x\mathrm{d}y^{2}}f_{Y|X}(y,x)\\
    &+O(h_{\mathsf x}^{2}h_{\mathsf y}^{5})+O(h_{\mathsf x}^{6}).
\end{align*}
\end{lemma}

We also use Lemma 2 from \citet{hyndman1996estimating}, which is restated as Lemma~\ref{baspro_cdf} and will be applied to the proofs in Section~\ref{main_sec}.

\begin{lemma}\label{baspro_cdf}
Supposing that $\{X_{1},\dots,X_{n} \}$ is an i.i.d. sequence of random variables and $q_{1}(X_{i})$ and $q_{2}(X_{i})$ are two random variables with means $u_{1}$ and $u_{2}$ and variances $v_{1}$ and $v_{2}$ respectively, and with covariance $v_{12}$. Defining $\hat{\Sigma}_{1}=\frac{1}{n}\Sigma^{n}_{i=1}q_{1}(X_{i})$ , $\hat{\Sigma}_{2}=\frac{1}{n}\Sigma^{n}_{i=1}q_{2}(X_{i})$ and $\hat{R}:=\hat{R}(\hat{\Sigma}_{1}, \hat{\Sigma}_{2})=\frac{\hat{\Sigma}_{1}}{\hat{\Sigma}_{2}}$. Then the second-order approximation of $\E [\hat{R}]$ is 
\begin{equation}\label{eq:mean_appro}
    \E [\hat{R}]\approx \frac{\mu_{1}}{\mu_{2}}+\frac{1}{n}(\frac{\mu_{1}v_{2}}{\mu^{3}_{2}}-\frac{v_{12}}{\mu^{2}_{2}}),
\end{equation}
and the first-order approximation of $\Var[\hat{R}]$ is 
\begin{equation}\label{eq:var_appro}
    \Var [\hat{R}]\approx \frac{1}{n\mu^{2}_{2}}( v_{1}+\frac{\mu^{2}_{1}v_{2}}{\mu^{2}_{2}}-2\frac{\mu_{1}v_{12}}{\mu_{2}} ).
\end{equation}
\end{lemma}

\smallskip
\begin{remark}
Let $X$ and $Y$ be two random variables with means $u_{x}$ and $u_{y}$ and variances $v_{x}$ and $v_{y}$ respectively, and with covariance $v_{xy}$. Let $\hat{R}:=\frac{X}{Y}$.
The second-order approximation of $\E [\hat{R}]$ is
\begin{align*}
&\E [\hat{R}(X,Y)]\\
\approx&\E\Bigg[\hat{R}(\hat{\Sigma})+\frac{\mathrm{d}}{\mathrm{d}x}\hat{R}(\hat{\Sigma})(X-u_x)+\frac{\mathrm{d}}{\mathrm{d}y}\hat{R}(\hat{\Sigma})(Y-u_y)\\
&+\frac{1}{2}\left[ \frac{\mathrm{d}^2}{\mathrm{d}x^2}\hat{R}(\hat{\Sigma})(X-u_x)^2
+2\frac{\mathrm{d}^2}{\mathrm{d}x\mathrm{d}y}\hat{R}(\hat{\Sigma})(X-u_x)(Y-u_y)+\frac{\mathrm{d}^2}{\mathrm{d}y^2}\hat{R}(\hat{\Sigma})(Y-u_y)^2 \right] \Bigg]\\
=&\frac{u_x}{u_y}+\frac{1}{2}\left[ \frac{\mathrm{d}^2}{\mathrm{d}x^2}\hat{R}(\hat{\Sigma})v_{x}+2\frac{\mathrm{d}^2}{\mathrm{d}x\mathrm{d}y}\hat{R}(\hat{\Sigma})v_{xy}+\frac{\mathrm{d}^2}{\mathrm{d}y^2}\hat{R}(\hat{\Sigma})v_{y} \right],
\end{align*}
and the first-order approximation of $\Var[\hat{R}]$ is
\begin{align*}
\Var[\hat{R}(X,Y)]\approx& \E\left[ \left( \hat{R}(\hat{\Sigma})+\frac{\mathrm{d}}{\mathrm{d}x}\hat{R}(\hat{\Sigma})(X-u_x)+\frac{\mathrm{d}}{\mathrm{d}y}\hat{R}(\hat{\Sigma})(Y-u_y)-\hat{R}(\hat{\Sigma})\right)^2 \right]\\
=&(\frac{\mathrm{d}}{\mathrm{d}x}\hat{R}(\hat{\Sigma}))^2v_{x}+2\frac{\mathrm{d}}{\mathrm{d}x}\hat{R}(\hat{\Sigma})\frac{\mathrm{d}}{\mathrm{d}y} \hat{R}(\hat{\Sigma})v_{xy} +(\frac{\mathrm{d}}{\mathrm{d}y}\hat{R}(\hat{\Sigma}))^2 v_{y}\\
=&\frac{1}{u^2_{y}}\left[v^2_{x}-2\frac{u_{x}v_{xy}}{u_{y}}+\frac{u^2_{x}v_{y}}{u^2_y} \right],
\end{align*}
where 
$\hat{\Sigma}=(u_{x},u_{y})$.
\end{remark}

\subsection{Proofs of the Main Theorems in Section~\ref{main_sec} }

\subsubsection{The Proof of Lemma~\ref{le_vari}.}
\begin{proof}
Based on the univariate form of the kernel estimator \eqref{condi_densityestima}, we have 
\begin{align}\label{firderiv_1dimen}
    \frac{\mathrm{d}}{\mathrm{d}x}\hat{f}_{Y|X}(y,x)=\left(\sum^{n}_{j=1}K(u_{x_{j}})\right)^{-1}\sum^{n}_{j=1} \frac{1}{h_{\mathsf x}h_{\mathsf y}}\left[ u_{x_{j}}K(u_{x_{j}})K(u_{y_{j}})
    -\frac{ \sum_{i=1}^{n}u_{x_{i}}K(u_{x_{i}})K(u_{x_{j}})K(u_{y_{j}}) }{\sum_{i=1}^{n}K(u_{x_{i}})}  \right],
\end{align}
where $u_{x_{i}}=\frac{\hat{X}_{i}-x}{h_{\mathsf x}}$ and $u_{y_{i}}=\frac{\hat{Y}_{i}-y}{h_{\mathsf y}}$.

In order to obtain the variance of $ \frac{\mathrm{d}}{\mathrm{d}x}\hat{f}_{Y|X}(y,x)$, according to Lemma~\ref{baspro_cdf}, we need to calculate the variances and means of $K(u_{x_{i}})$ and
\begin{align*}
    \lambda_{main}:=&\frac{1}{h_{\mathsf x} h_{\mathsf y}}\left[ u_{x_{j}}K(u_{x_{j}})K(u_{y_{j}})
    -\frac{ \sum_{i=1}^{n}u_{x_{i}}K(u_{x_{i}})K(u_{x_{j}})K(u_{y_{j}}) }{\sum_{i=1}^{n}K(u_{x_{i}})} \right],
\end{align*}
and their covariance.

For $\lambda_{main}$, we obtain that  
\begin{align}\label{var_variance}
    &\Var\left[\frac{1}{h_{\mathsf x} h_{\mathsf y}} \left( u_{x_{j}}K(u_{x_{j}})K(u_{y_{j}})
    -\frac{ \sum_{i=1}^{n}u_{x_{i}}K(u_{x_{i}})K(u_{x_{j}})K(u_{y_{j}}) }{\sum_{i=1}^{n}K(u_{x_{i}})}  \right) \right]\nonumber\\
    =&\frac{1}{h_{\mathsf x}^{2} h_{\mathsf y}^{2}}\Var\left[\frac{ \sum_{i=1}^{n}u_{x_{i}}K(u_{x_{i}})K(u_{x_{j}})K(u_{y_{j}}) }{\sum_{i=1}^{n}K(u_{x_{i}})}\right]
    +\frac{1}{h_{\mathsf x}^{2} h_{\mathsf y}^{2}}\Var[u_{x_{j}}K(u_{x_{j}})K(u_{y_{j}})]\nonumber\\
    &-\frac{2}{h_{\mathsf x}^{2} h_{\mathsf x}^{2}} \E\left[\frac{ \sum_{i=1}^{n}u_{x_{i}}K(u_{x_{i}})u_{x_{j}}K^{2}(u_{x_{j}})K^{2}(u_{y_{j}}) }{\sum_{i=1}^{n}K(u_{x_{i}})}\right]\nonumber\\
    &+ \frac{2}{h_{\mathsf x}^{2} h_{\mathsf x}^{2}}\E\left[\frac{ \sum_{i=1}^{n}u_{x_{i}}K(u_{x_{i}})K(u_{x_{j}})K(u_{y_{j}}) }{\sum_{i=1}^{n}K(u_{x_{i}})}\right]\E[u_{x_{j}}K(u_{x_{j}})K(u_{y_{j}})] ,
\end{align}
\begin{align}\label{mean_variance}
    &\E\left[\frac{1}{h_{\mathsf x} h_{\mathsf y}}\left( u_{x_{j}}K(u_{x_{j}})K(u_{y_{j}})
    -\frac{ \sum_{i=1}^{n}u_{x_{i}}K(u_{x_{i}})K(u_{x_{j}})K(u_{y_{j}}) }{\sum_{i=1}^{n}K(u_{x_{i}})}  \right)\right]\nonumber\\
    =&\frac{1}{h_{\mathsf x} h_{\mathsf y}}\E[u_{x_{j}}K(u_{x_{j}})K(u_{y_{j}})]
    -\frac{1}{h_{\mathsf x} h_{\mathsf y}}E\left[\frac{ \sum_{i=1}^{n}u_{x_{i}}K(u_{x_{i}})K(u_{x_{j}})K(u_{y_{j}}) }{\sum_{i=1}^{n}K(u_{x_{i}})} \right],\end{align}
and the covariance 
\begin{align}\label{covar_variance}
    &\Cov\left[\frac{1}{h_{\mathsf x} h_{\mathsf y}}\left( u_{x_{j}}K(u_{x_{j}})K(u_{y_{j}})
    -\frac{ \sum_{i=1}^{n}u_{x_{i}}K(u_{x_{i}})K(u_{x_{j}})K(u_{y_{j}}) }{\sum_{i=1}^{n}K(u_{x_{i}})}  \right),K(u_{x_{j}})\right]\nonumber \\
    =&\frac{1}{h_{\mathsf x} h_{\mathsf y}} \E[ u_{x_{j}}K^{2}(u_{x_{j}})K(u_{y_{j}})]
    -\frac{1}{h_{\mathsf x} h_{\mathsf y}}\E\left[\frac{ \sum_{i=1}^{n}u_{x_{i}}K(u_{x_{i}})K^{2}(u_{x_{j}})K(u_{y_{j}}) }{\sum_{i=1}^{n}K(u_{x_{i}})} \right]\nonumber\\
    &+\frac{1}{h_{\mathsf x} h_{\mathsf y}}\E\left[\frac{ \sum_{i=1}^{n}u_{x_{i}}K(u_{x_{i}})K(u_{x_{j}})K(u_{y_{j}}) }{\sum_{i=1}^{n}K(u_{x_{i}})} \right]\E[ K(u_{x_{j}}) ]\nonumber\\
    &-\frac{1}{h_{\mathsf x}h_{\mathsf y}}\E[ u_{x_{j}}K(u_{x_{j}})K(u_{y_{j}})] \E[ K(u_{x_{j}}) ].
\end{align}

In order to calculate equation \eqref{var_variance}, \eqref{mean_variance} and \eqref{covar_variance}, we have to calculate the mean and variance of
\begin{equation*}
    \lambda:=\frac{ \sum_{i=1}^{n}u_{x_{i}}K(u_{x_{i}})K(u_{x_{j}})K(u_{y_{j}}) }{\sum_{i=1}^{n}K(u_{x_{i}})},
\end{equation*}
and the means of 
\begin{equation*}
     \lambda^*:=\frac{ \sum_{i=1}^{n}u_{x_{i}}K(u_{x_{i}})K^{2}(u_{x_{j}})K(u_{y_{j}}) }{\sum_{i=1}^{n}K(u_{x_{i}})},
\end{equation*}
and 
\begin{equation*}
    \hat{\lambda}:=\frac{ \sum_{i=1}^{n}u_{x_{i}}K(u_{x_{i}})u_{x_{j}}K^{2}(u_{x_{j}})K^{2}(u_{y_{j}}) }{\sum_{i=1}^{n}K(u_{x_{i}})}.
\end{equation*}

Thus, at this step,
we need to calculate the variance and mean of $ \lambda$.
Similarly to the previous steps, we also first need to calculate the variances and means of $u_{x_{i}}K(u_{x_{i}})K(u_{x_{j}})K(u_{y_{j}})$ and $K(u_{x_{i}})$, and their covariance.

For $u_{x_{i}}K(u_{x_{i}})K(u_{x_{j}})K(u_{y_{j}})$, we have
\begin{align*}
    \Var[u_{x_{i}}K(u_{x_{i}})K(u_{x_{j}})K(u_{y_{j}})]
    =&\E[(u_{x_{i}}K(u_{x_{i}})K(u_{x_{j}})K(u_{y_{j}})
    -\E[u_{x_{i}}K(u_{x_{i}})K(u_{x_{j}})K(u_{y_{j}})])^{2}]\\
    =&P(i\not = j)\E[(u_{x_{i}}K(u_{x_{i}})K(u_{x_{j}})K(u_{y_{j}})
    -\E[u_{x_{i}}K(u_{x_{i}})K(u_{x_{j}})K(u_{y_{j}})])^{2}]\\
    &+P(i= j)\E[(u_{x_{i}}K(u_{x_{i}})K(u_{x_{j}})K(u_{y_{j}})
    -\E[u_{x_{i}}K(u_{x_{i}})K(u_{x_{j}})K(u_{y_{j}})])^{2}]\\
    =&\frac{n-1}{n}\left( \E[u^{2}_{x_{i}}K^{2}(u_{x_{i}})] \E[K^{2}(u_{x_{j}})K^{2}(u_{y_{j}})]
    -\E^{2}[u_{x_{i}}K(u_{x_{i}})] \E^{2}[K(u_{x_{j}})K(u_{y_{j}})] \right)
    \\
    &+\frac{1}{n}\left( \E[u^{2}_{x_{i}}K^{4}(u_{x_{i}})K^{2}(u_{y_{i}})]-\E^{2}[u_{x_{i}}K^{2}(u_{x_{i}})K(u_{y_i})] \right),
\end{align*}
\begin{align*}
    \E[u_{x_{i}}K(u_{x_{i}})K(u_{x_{j}})K(u_{y_{j}})]
    =&\frac{n-1}{n}\E[u_{x_{i}}K(u_{x_{i}})]
    \E[K(u_{x_{j}})K(u_{y_{j}})]
    +\frac{1}{n}\E[u_{x_{i}}K^{2}(u_{x_{i}})K(u_{y_{i}})],
\end{align*}
and the covariance of $u_{x_{i}}K(u_{x_{i}})K(u_{x_{j}})K(u_{y_{j}})$ and $K(u_{x_{i}})$,
\begin{align*}
    \Cov[u_{x_{i}}K(u_{x_{i}})&K(u_{x_{j}})K(u_{y_{j}}), K(u_{x_{i}})]\\
    =&\E[u_{x_{i}}K^{2}(u_{x_{i}})K(u_{x_{j}})K(u_{y_{j}})]
    -\E[u_{x_{i}}K(u_{x_{i}})K(u_{x_{j}})K(u_{y_{j}})]\E[K(u_{x_{i}})]\\
    =&\frac{n-1}{n}\E[u_{x_{i}}K^{2}(u_{x_{i}})]\E[K(u_{x_{j}})K(u_{y_{j}})]
    +\frac{1}{n}\E[u_{x_{i}}K^{3}(u_{x_{i}})K(u_{y_{i}})]\\
    &-\frac{n-1}{n}\E[u_{x_{i}}K(u_{x_{i}})]\E[K(u_{x_{j}})K(u_{y_{j}})]\E[K(u_{x_{i}})]
    -\frac{1}{n}\E[u_{x_{i}}K^{2}(u_{x_{i}})K(u_{y_{i}})]\E[K(u_{x_{i}})].
\end{align*}

Based on Lemma~\ref{exp_main} and $f_{X}(x)$ is the density function of the uniform distribution, we can have $\E[u_{x_{i}}K(u_{x_{i}})]=\E[u_{x_{i}}K^{2}(u_{x_{i}})]=0$, $\E[K(u_{x_{i}})]= h_{\mathsf x} f_{X}(x)$, $\Var[K(u_{x_{i}})]= h_{\mathsf x} f_{X}(x)(G_{20}(K)-1)$ and
$
\E[u^{2}_{x_{i}}K^{2}(u_{x_{i}})]=h_{\mathsf x} G_{22}(K)f_{X}(x).
$
In addition, based on Lemma~\ref{element_expect}, we can obtain that 
\begin{align}\label{eq:lamup_var}
\Var[u_{x_{i}}K(u_{x_{i}})K(u_{x_{j}})K(u_{y_{j}})]
  =&\frac{n-1}{n}[ h_{\mathsf x}^{2} h_{\mathsf y} G^{2}_{2,0}(K)G_{22}(K)f^{2}_{X}(x)f_{Y|X}(y,x) ]+O(h_{\mathsf x}^2 h_{\mathsf y}^3)+O(h_{\mathsf x}^4h_{\mathsf y})\nonumber\\
  &+O(h_{\mathsf x}^6)+O(\frac{h_{\mathsf x}h_{\mathsf y}}{n})
  +O(\frac{h_{\mathsf x}^5}{n}),
\end{align}
\begin{align}\label{eq:lamup_mean}
    \E[u_{x_{i}}K(u_{x_{i}})&K(u_{x_{j}})K(u_{y_{j}})]\nonumber\\
    =& \frac{1}{n} \!\left[ h_{\mathsf x}^{2} h_{\mathsf y} G_{22}(K)f_{X}(x)\frac{\mathrm{d}}{\mathrm{d}x}f_{Y|X}(y,x)
    \!+\!\frac{1}{2}h_{\mathsf x}^{2} h_{\mathsf y}^{3}G_{12}(K)G_{22}(K)f_{X}(x)\frac{\mathrm{d}^{3}}{\mathrm{d}x\mathrm{d}y^{2}}f_{Y|X}(y,x)\!+\!O(h_{\mathsf x}^{5})\! \right]\nonumber\\
    =& \frac{1}{n}  h_{\mathsf x}^{2} h_{\mathsf y} G_{22}(K)f_{X}(x)\frac{\mathrm{d}}{\mathrm{d}x}f_{Y|X}(y,x)+O(\frac{h_{\mathsf x}^2 h_{\mathsf y}^3}{n})+O(\frac{h_{\mathsf x}^5}{n}),
\end{align}
and
\begin{align}\label{eq:lamup_cov}
    \Cov[u_{x_{i}}K(u_{x_{i}})&K(u_{x_{j}})K(u_{y_{j}}), K(u_{x_{i}})]\nonumber\\
    =& \frac{1}{n} \bigg[ h_{\mathsf x}^{2} h_{\mathsf y} G_{32}(K)f_{X}(x)\frac{\mathrm{d}}{\mathrm{d}x}f_{Y|X}(y,x)
    +\frac{1}{2}h_{\mathsf x}^{2} h_{\mathsf y}^{3}G_{12}(K)G_{32}(K)f_{X}(x)\frac{\mathrm{d}^{3}}{\mathrm{d}x\mathrm{d}y^{2}}f_{Y|X}(y,x)\nonumber\\
    &+O(h_{\mathsf x}^{2} h_{\mathsf y}^{5})+O(h_{\mathsf x}^{6}) \bigg]
    -\frac{1}{n} h_{\mathsf x} f_{X}(x) \bigg[ h_{\mathsf x}^{2} h_{\mathsf y} G_{22}(K)f_{X}(x)\frac{\mathrm{d}}{\mathrm{d}x}f_{Y|X}(y,x)\nonumber\\
    &+\frac{1}{2}h_{\mathsf x}^{2} h_{\mathsf y}^{3}G_{12}(K)G_{22}(K)f_{X}(x)\frac{\mathrm{d}^{3}}{\mathrm{d}x\mathrm{d}y^{2}}f_{Y|X}(y,x)
    +O(h_{\mathsf x}^{6})+O(h_{\mathsf x}^{2} h_{\mathsf y}^{5}) \bigg] \nonumber\\
    =&\frac{h_{\mathsf x}^{2} h_{\mathsf y}}{n} f_{X}(x)\left[ G_{32}(K)-h_{\mathsf x} G_{22}(K)f_{X}(x) \right] \frac{\mathrm{d}}{\mathrm{d}x}f_{Y|X}(y,x)+O(\frac{h_{\mathsf x}^2 h_{\mathsf y}^3}{n})+O(\frac{h_{\mathsf x}^6}{n}),
\end{align}
if $h_{\mathsf x},h_{\mathsf y}\to 0$ and $n\to +\infty$.
Then, by substituting equation \eqref{eq:lamup_var}, \eqref{eq:lamup_mean} and \eqref{eq:lamup_cov} into equation \eqref{eq:mean_appro} and \eqref{eq:var_appro}, we can obtain the variance and mean of $\lambda$ as follows
\begin{align*}
    &\E[\lambda]    = \frac{ 1 }{n}h_{\mathsf x}h_{\mathsf y}G_{22}(K)\frac{\mathrm{d}}{\mathrm{d}x}f_{Y|X}(y,x)+O(\frac{h_{\mathsf x} h_{\mathsf y}^3}{n})+O(\frac{h_{\mathsf x}^4}{n})+O(\frac{h_{\mathsf y}}{n^2})+O(\frac{h_{\mathsf x}^3}{n^2}),
\end{align*}
and
\begin{align*}
    \Var[\lambda]
    = \frac{n-1}{n^{2}f_{X}(x)}h_{\mathsf y} G^{2}_{2,0}(K)G_{22}(K)f_{Y|X}(y,x)+O(\frac{h_{\mathsf x}h_{\mathsf y}^3}{n})+O(\frac{h_{\mathsf x}^3 h_{\mathsf y}}{n})+O(\frac{h_{\mathsf x}^5}{n})+O(\frac{h_{\mathsf y}}{h_{\mathsf x}n^2}).
\end{align*}

The next step is to
calculate the mean of $\lambda^*$.
To calculate its mean, we need to calculate the variance and mean of $K(u_{x_{i}})$ and $u_{x_{i}}K(u_{x_{i}})K^{2}(u_{x_{j}})K(u_{y_{j}})$, and their covariance.
Since by using Lemma~\ref{element_expect} we have 
\begin{align*}
    \E[u_{x_{i}}K(u_{x_{i}})K^{2}(u_{x_{j}})K(u_{y_{j}})]
    =&\frac{n-1}{n}\E[u_{x_{i}}K(u_{x_{i}})]\E[K^{2}(u_{x_{j}})K(u_{y_{j}})]
    +\frac{1}{n}\E[u_{x_{i}}K^{3}(u_{x_{i}})K(u_{y_{i}})]\\
    =&\frac{1}{n}h_{\mathsf x}^{2} h_{\mathsf y} G_{32}(K)f_{X}(x)\frac{\mathrm{d}}{\mathrm{d}x}f_{Y|X}(y,x)+O(\frac{h_{\mathsf x}^2 h_{\mathsf y}^3}{n})+O(\frac{h_{\mathsf x}^6}{n}),
\end{align*}
and 
\begin{align*}
    \Cov[u_{x_{i}}K(u_{x_{i}})&K^{2}(u_{x_{j}})K(u_{y_{j}}),K(u_{x_{i}})]\\
    =&\E[u_{x_{i}}K^{2}(u_{x_{i}})K^{2}(u_{x_{j}})K(u_{y_{j}})]
    -\E[u_{x_{i}}K(u_{x_{i}})K^{2}(u_{x_{j}})K(u_{y_{j}})]\E[K(u_{x_{i}})]\\
    =&\frac{n-1}{n}\E[u_{x_{i}}K^{2}(u_{x_{i}})]\E[K^{2}(u_{x_{j}})K(u_{y_{j}})]
    -\frac{n-1}{n}\E[u_{x_{i}}K(u_{x_{i}})]\E[K^{2}(u_{x_{j}})K(u_{y_{j}})]\E[K(u_{x_{i}})]\\
    &+\frac{1}{n}\E[u_{x_{i}}K^{4}(u_{x_{i}})K(u_{y_{i}})]
    -\frac{1}{n}\E[u_{x_{i}}K^{3}(u_{x_{i}})K(u_{y_{i}})]\E[K(u_{x_{i}})]\\
    =&\frac{h_{\mathsf x}^{2} h_{\mathsf y}}{n}f_{X}(x)[G_{42}(K)-h_{\mathsf x}G_{32}(K)f_{X}(x)]\frac{\mathrm{d}}{\mathrm{d}x}f_{Y|X}(y,x)+O(\frac{h_{\mathsf x}^2 h_{\mathsf y}^3}{n})+O(\frac{h_{\mathsf x}^6}{n}),
\end{align*}
the mean of $\lambda^{*}$ can be obtained as follows
\begin{align*}
    \E[\lambda^{*}]
    = \frac{ 1 }{n}h_{\mathsf x} h_{\mathsf y} G_{32}(K)\frac{\mathrm{d}}{\mathrm{d}x}f_{Y|X}(y,x)+O(\frac{h_{\mathsf x} h_{\mathsf y}^3}{n})+O(\frac{h_{\mathsf x}^5}{n})+O(\frac{h_{\mathsf y}}{n^2})+O(\frac{h_{\mathsf x}^4}{n^2}),
\end{align*}
if $h_{\mathsf x},h_{\mathsf y}\to 0$ and $n\to +\infty$.

Next, we need to calculate the mean of $\hat{\lambda}$.
Similarly, the variance and mean of $u_{x_{i}}K(u_{x_{i}})u_{x_{j}}K^{2}(u_{x_{j}})K^{2}(u_{y_{j}})$ and $K(u_{x_{i}})$, and their covariance, all need to be calculated.
By using Lemma~\ref{element_expect}, we have
\begin{align*}
    \Var[u_{x_{i}}K(u_{x_{i}})&u_{x_{j}}K^{2}(u_{x_{j}})K^{2}(u_{y_{j}})]\\
    =&\frac{n-1}{n} ( \E[u^{2}_{x_{i}}K^{2}(u_{x_{i}})]\E[u^{2}_{x_{j}}K^{4}(u_{x_{j}})K^{4}(u_{y_{j}})]
    -\E^{2}[u_{x_{i}}K(u_{x_{i}})]\E^{2}[u_{x_{j}}K^{2}(u_{x_{j}})K^{2}(u_{y_{j}})] )\\
    &+\frac{1}{n} ( \E[u^{4}_{x_{i}}K^{6}(u_{x_{i}})K^{4}(u_{y_{i}})]-\E^{2}[u^{2}_{x_{i}}K^{3}(u_{x_{i}})K^{2}(u_{y_{i}})] )\\
    =&\frac{n-1}{n}h_{\mathsf x}^{2} h_{\mathsf y} G_{22}(K)G_{40}(K)G_{42}(K)f^{2}_{X}(x)f_{Y|X}(y,x)+O(h_{\mathsf x}^2 h_{\mathsf y}^3)+O(h_{\mathsf x}^4 h_{\mathsf y})\\
    &+O(h_{\mathsf x}^6)+O(\frac{h_{\mathsf x} h_{\mathsf y}}{n})+O(\frac{h_{\mathsf x}^5}{n}),
\end{align*}
\begin{align*}
    \E[u_{x_{i}}K(u_{x_{i}})u_{x_{j}}K^{2}(u_{x_{j}})K^{2}(u_{y_{j}})]
    =&\frac{n-1}{n}\E[u_{x_{i}}K(u_{x_{i}})]\E[u_{x_{j}}K^{2}(u_{x_{j}})K^{2}(u_{y_{j}})]
    +\frac{1}{n}\E[u^{2}_{x_{i}}K^{3}(u_{x_{i}})K^{2}(u_{y_{i}})]\\
    =& \frac{h_{\mathsf x} h_{\mathsf y}}{n} G_{20}(K)G_{32}(K) f_{X}(x)f_{Y|X}(y,x)+O(\frac{h_{\mathsf x} h_{\mathsf y}^3}{n})+O(\frac{h_{\mathsf x}^5}{n})+O(\frac{h_{\mathsf x}^3 h_{\mathsf y}}{n}),
\end{align*}
and
\begin{align*}
\Cov[u_{x_{i}}K(u_{x_{i}})&u_{x_{j}}K^{2}(u_{x_{j}})K^{2}(u_{y_{j}}), K(u_{x_{i}})]\\
    =&\frac{n-1}{n}\E[u_{x_{i}}K^{2}(u_{x_{i}})]\E[u_{x_{j}}K^{2}(u_{x_{j}})K^{2}(u_{y_{j}})]
    +\frac{1}{n}\E[u^{2}_{x_{i}}K^{4}(u_{x_{i}})K^{2}(u_{y_{i}})]\\
    &-\frac{n-1}{n}\E[u_{x_{i}}K(u_{x_{i}})]\E[u_{x_{j}}K^{2}(u_{x_{j}})K^{2}(u_{y_{j}})]\E[K(u_{x_{i}})]\\
    &-\frac{1}{n}\E[u^{2}_{x_{i}}K^{3}(u_{x_{i}})K^{2}(u_{y_{i}})]\E[K(u_{x_{i}})]\\
    =&\frac{1}{n}h_{\mathsf x} h_{\mathsf y} G_{20}(K)f_{X}(x)f_{Y|X}(y,x)[ G_{42}(K)-h_{\mathsf x}G_{32}(K)f_{X}(x) ]\\
    &+O(\frac{h_{\mathsf x} h_{\mathsf y}^3}{n})+O(\frac{h_{\mathsf x}^3 h_{\mathsf y}}{n})+O(\frac{h_{\mathsf x}^5}{n}).
\end{align*}
Thus, the mean of $\hat{\lambda}$ can be obtained according to Lemma~\ref{baspro_cdf}, as follows
\begin{align*}
    \E[\hat{\lambda}]
    =& \frac{1}{n}h_{\mathsf x} G_{20}(K)G_{32}(K)f_{Y|X}(y,x)
    +O(\frac{h_{\mathsf y}^3}{n})+O(\frac{h_{\mathsf x}^2 h_{\mathsf y}}{n})+O(\frac{h_{\mathsf x}^4}{n})+O(\frac{h_{\mathsf y}}{h_{\mathsf x}n^2})+O(\frac{h_{\mathsf x}^3}{n^2}).
\end{align*}
if $h_{\mathsf x},h_{\mathsf y}\to 0$ and $n\to +\infty$.

Through combining the above results, we can obtain that 
\begin{align*}
    \Var[\lambda_{main} ]
    =&\frac{1}{h_{\mathsf x}^{2} h_{\mathsf y}^{2}}\Var[u_{x_{j}}K(u_{x_{j}})K(u_{y_{j}})]+O(\frac{1}{h_{\mathsf x}^2 h_{\mathsf y}n})+O(\frac{h_{\mathsf x}^3}{h_{\mathsf y}^2n})+O(\frac{1}{h_{\mathsf x}^3h_{\mathsf y}n^2})+O(\frac{h_{\mathsf x}^2}{h_{\mathsf y}^2n})\\
    =& \frac{1}{h_{\mathsf x}h_{\mathsf y}} G_{20}(K)G_{22}(K)f_{X}(x)f_{Y|X}(y,x)
    +O(\frac{h_{\mathsf y}}{h_{\mathsf x}})+O(\frac{h_{\mathsf x}}{h_{\mathsf y}})+O(\frac{1}{h_{\mathsf x}^2 h_{\mathsf y}n})\\
    &+O(\frac{h_{\mathsf x}^3}{h_{\mathsf y}^2 n})+O(\frac{1}{h_{\mathsf x}^3 h_{\mathsf y} n^2})+O(\frac{h_{\mathsf x}^2}{h_{\mathsf y}^2 n^2}) ,
\end{align*}
\begin{align*}
    \E[\lambda_{main}]
    =&h_{\mathsf x} f_{X}(x)\frac{\mathrm{d}}{\mathrm{d}x}f_{Y|X}(y,x)+O(h_{\mathsf x} h_{\mathsf y}^2)+O(\frac{h_{\mathsf x}^5}{h_{\mathsf y}})+O(\frac{1}{n})+O(\frac{h_{\mathsf x}^3}{h_{\mathsf y}n})+O(\frac{1}{h_{\mathsf x}n^2})+O(\frac{h_{\mathsf x}^2}{h_{\mathsf y}n^2}),
\end{align*}
and
\begin{align*}
    \Cov[\lambda_{main}]
    =& \frac{1}{h_{\mathsf x} h_{\mathsf y}} \E[ u_{x_{j}}K^{2}(u_{x_{j}})K(u_{y_{j}})]
    -\frac{1}{h_{\mathsf x} h_{\mathsf y}}\E[ u_{x_{j}}K(u_{x_{j}})K(u_{y_{j}})] \E[ K(u_{x_{j}}) ]\\
    &+O(\frac{1}{n})+O(\frac{h_{\mathsf x}^4}{h_{\mathsf y}n})+O(\frac{h_{\mathsf y}^2}{n^2})+O(\frac{h_{\mathsf x}^3}{h_{\mathsf y}n^2})\\
    =&  h_{\mathsf x} G_{22}(K)f_{X}(x)\frac{\mathrm{d}}{\mathrm{d}x}f_{Y|X}(y,x)
    -h_{\mathsf x}^{2}G_{12}(K)f^{2}_{X}(x) \frac{\mathrm{d}}{\mathrm{d}x}f_{Y|X}(y,x)+O(h_{\mathsf x} h_{\mathsf y}^2)+O(\frac{h_{\mathsf x}^5}{h_{\mathsf y}})\\
    &+O(\frac{1}{n})+O(\frac{h_{\mathsf x}^4}{h_{\mathsf y}n})+O(\frac{h_{\mathsf y}^2}{n^2})+O(\frac{h_{\mathsf x}^3}{h_{\mathsf y}n^2}),
\end{align*}
if $h_{\mathsf x},h_{\mathsf y}\to 0$ and $n\to +\infty$.

Finally, based on Lemma~\ref{baspro_cdf}, we can have 
\begin{align*}
    \Var\left[\frac{\mathrm{d}}{\mathrm{d}x}\hat{f}_{Y|X}(y,x)\right]
    =&\frac{1}{nh_{\mathsf x}^{3}h_{\mathsf y}f_{X}(x)} G_{20}(K)G_{22}(K)f_{Y|X}(y,x)
    +O(\frac{h_{\mathsf y}}{nh_{\mathsf x}^3})+O(\frac{1}{nh_{\mathsf x}h_{\mathsf y}})+O(\frac{h_{\mathsf x}^7}{nh_{\mathsf y}^2})+O(\frac{1}{n^2h_{\mathsf x}^4 h_{\mathsf y}^2})\\
    &+O(\frac{1}{n^3h_{\mathsf x}^5 h_{\mathsf y}^2})+O(\frac{1}{n^4h_{\mathsf x}^4 h_{\mathsf y}^2})+O(\frac{1}{n^5h_{\mathsf x}^5 h_{\mathsf y}^2})\\
    =&\frac{1}{nh_{\mathsf x}^{3} h_{\mathsf y}f_{X}(x)} G_{20}(K)G_{22}(K)f_{Y|X}(y,x)+O(\frac{h_{\mathsf y}}{nh_{\mathsf x}^{3}})+O(\frac{1}{nh_{\mathsf x}h_{\mathsf y}})+O(\frac{h_{\mathsf x}^7}{nh_{\mathsf y}^{2}})+O(\frac{1}{n^{2}h_{\mathsf x}^5h_{\mathsf y}^2})
\end{align*}
Also, we have that for large $n$ if $nh_{\mathsf x}^{3}h_{\mathsf y}\to +\infty$ and $h_{\mathsf x},h_{\mathsf y}\to 0$ as $n\to +\infty $, then 
$$
    \Var\left[\frac{\mathrm{d}}{\mathrm{d}x}\hat{f}_{Y|X}(y,x)\right]\lesssim \frac{1}{nh_{\mathsf x}^{3}h_{\mathsf y}}C_{1},
$$
where \[C_{1}=\Vol(D_X)(x)G_{20}(K)G_{22}(K)\max_{(x,y)\in D_X\times D_Y}f_{Y|X}(y,x),\]
and $\Vol(\cdot)$ indicates the volume (Lebesgue measure) of a set.
\end{proof}

\subsubsection{The Proof of Lemma~\ref{le_bias}}
\begin{proof}
Based on equation~\eqref{bias_from1996}, if $h_{\mathsf x}\to 0$, $h_{\mathsf y}\to 0$ and $n\to +\infty$,
we can derived 
\begin{align}\label{eq_bias}
    \E\left[\frac{\mathrm{d}}{\mathrm{d}x}\hat{f}_{Y|X}(y,x)\right]-\frac{\mathrm{d}}{\mathrm{d}x}f_{Y|X}(y,x)
    =&\frac{h_{\mathsf x}^{2}G_{12}(K)}{2}\left[\frac{h_{\mathsf y}^{2}}{h_{\mathsf x}^{2}}\frac{\mathrm{d}^{3}}{\mathrm{d}x \mathrm{d}y^{2}}f_{Y|X}(y,x)+\frac{\mathrm{d}^{3}}{\mathrm{d}x^{3}}f_{Y|X}(y,x)\right]+O(h_{\mathsf x}^{2}h_{\mathsf y}^{2})\nonumber\\
    &+O(h_{\mathsf y}^{4})+O(h_{\mathsf x}^{4})+O(\frac{1}{nh_{\mathsf x}})
   \lesssim \frac{h_{\mathsf x}^{2}}{2}A,
\end{align}
where 
\[A:=G_{12}(K)\left[\frac{h_{\mathsf y}^{2}}{h_{\mathsf x}^{2}}\max_{(x,y)\in D_X\times D_Y}\left|\frac{\mathrm{d}^{3}}{\mathrm{d}x \mathrm{d}y^{2}}f_{Y|X}(y,x)\right|+\max_{(x,y)\in D_X\times D_Y}\left|\frac{\mathrm{d}^{3}}{\mathrm{d}x^{3}}f_{Y|X}(y,x)\right|\right].\]
\end{proof}

\subsubsection{The Proof of Theorem~\ref{le_mse}}
\begin{proof}
Based on Lemmas~\ref{le_vari}--\ref{le_bias} and the following property
\begin{align*}
  &\E\left[\left(\frac{\mathrm{d}}{\mathrm{d}x}\hat{f}_{Y|X}(y,x)-\frac{\mathrm{d}}{\mathrm{d}x}f_{Y|X}(y,x)\right)^{2}\right]\\
  =&\E\left[\left(\frac{\mathrm{d}}{\mathrm{d}x}\hat{f}_{Y|X}(y,x)-\E[\frac{\mathrm{d}}{\mathrm{d}x}\hat{f}_{Y|X}(y,x)]\right)^{2}\right]
  +\left(\E\left[\frac{\mathrm{d}}{\mathrm{d}x}\hat{f}_{Y|X}(y,x)\right]-\frac{\mathrm{d}}{\mathrm{d}x}f_{Y|X}(y,x)\right)^{2},
\end{align*}
we have that for large $n$ if $nh_{\mathsf x}^{3}h_{\mathsf y}\to +\infty$ and $h_{\mathsf x},~h_{\mathsf y}\to 0$ as $n\to +\infty$
\begin{align*}
    \E\left[\left(\frac{\mathrm{d}}{\mathrm{d}x}\hat{f}_{Y|X}(y,x)-\frac{\mathrm{d}}{\mathrm{d}x}f_{Y|X}(y,x)\right)^{2}\right] \lesssim \epsilon_3, \text{ where $\epsilon_3:=\frac{C_{1}}{nh_{\mathsf x}^{3}h_{\mathsf y}}+\frac{h_{\mathsf x}^{4}}{4}A^{2}.$}
\end{align*}
\end{proof}

\subsubsection{The Proof of Theorem~\ref{LC_main}}
\begin{proof}
In order to give the bound, we need to introduce a sequence $\{(x_{j},y_{j})| \| (x_{j+1},y_{j+1})-(x_{j},y_{j})\|=\frac{D_B}{M},j=1,\ldots,M, M=n, \text{ and } D_B:=\max_{(x,y),(x',y')\in D_{X}\times D_{Y} }\|(x,y)-(x',y') \|\}$.
Then, the following inequality can be obtained
\begin{align}\label{bias_finaline}
    &\E\left[\max_{(x,y)\in D_{X}\times D_{Y}}\left|\frac{\mathrm{d}}{\mathrm{d}x}\hat{f}_{Y|X}(y,x)\right |\right]-\max_{(x,y)\in D_{X}\times D_{Y}} \left|\frac{\mathrm{d}}{\mathrm{d}x}f_{Y|X}(y,x)\right|  \nonumber\\
    \leq & \E\left[\max_{1\leq j\leq M} \left|\frac{\mathrm{d}}{\mathrm{d}x}\hat{f}_{Y|X}(y_{j},x_{j})\right|\right]-\max_{1\leq j\leq M} \left|\frac{\mathrm{d}}{\mathrm{d}x}f_{Y|X}(y,x)\right| \nonumber\\
    &+\E\left[\max_{\substack{(x,y),(x',y')\\ \in D_{X}\times D_{Y}:\\ \| (x,y)-(x',y') \|\leq \frac{D_B}{M} }}\left| \frac{\mathrm{d}}{\mathrm{d}x}\hat{f}_{Y|X}(y,x)-\frac{\mathrm{d}}{\mathrm{d}x}\hat{f}_{Y|X}(y',x') \right|\right]
\end{align}

To get the upper bound of the second line of equation \eqref{bias_finaline}, we need to consider 
\begin{align*}
     &\exp{ \left( t\E\left[\max_{1\leq j\leq M} \left|\frac{\mathrm{d}}{\mathrm{d}x}\hat{f}_{Y|X}(y_{j},x_{j})\right|\right]-t\max_{1\leq j\leq M} \left|\frac{\mathrm{d}}{\mathrm{d}x}f_{Y|X}(y_{j},x_{j})\right| \right) }\\
     =&\exp{ \E\left[ t\max_{1\leq j\leq M} \left|\frac{\mathrm{d}}{\mathrm{d}x}\hat{f}_{Y|X}(y_{j},x_{j})\right|-t\max_{1\leq j\leq M} \left|\frac{\mathrm{d}}{\mathrm{d}x}f_{Y|X}(y_{j},x_{j})\right|\right] }, \text{ where $t>0$.}
\end{align*}

Then, we can divide equality into two different situations: $ D_{\hat{f}}\ge 0$ and $  D_{\hat{f}}< 0$, where $D_{\hat{f}}:=\frac{\mathrm{d}}{\mathrm{d}x}\hat{f}_{Y|X}(y_{j^*},x_{j^*})$ and  $(x_{j^*},y_{j^*})=\argmax_{1\leq j\leq M}\left| \frac{\mathrm{d}}{\mathrm{d}x}\hat{f}_{Y|X}(y_{j},x_{j}) \right|$.
Thus, we have
\begin{align}\label{expon_bias}
    &\exp{ \E\left[ \max_{1\leq j\leq M} t\left|\frac{\mathrm{d}}{\mathrm{d}x}\hat{f}_{Y|X}(y_{j},x_{j})\right|-\max_{1\leq j\leq M} t\left|\frac{\mathrm{d}}{\mathrm{d}x}f_{Y|X}(y_{j},x_{j})\right|\right] }\nonumber\\
    =& \exp{ \E\left[  t\left|\frac{\mathrm{d}}{\mathrm{d}x}\hat{f}_{Y|X}(y_{j^*},x_{j^*})\right|\right] }   \exp{ \E\left[   -\max_{1\leq j\leq M} t\left|\frac{\mathrm{d}}{\mathrm{d}x}f_{Y|X}(y_{j},x_{j})\right| \right]  }    \nonumber\\
    =&\exp{ \bigg[ \E\left[  t\left|\frac{\mathrm{d}}{\mathrm{d}x}\hat{f}_{Y|X}(y_{j^*},x_{j^*})\right| \mid D_{\hat{f}}\ge 0\right] P(D_{\hat{f}}\ge 0) \bigg]}\nonumber\\
    &\exp{\bigg[ \E\left[  t\left|\frac{\mathrm{d}}{\mathrm{d}x}\hat{f}_{Y|X}(y_{j^*},x_{j^*})\right| \mid D_{\hat{f}}< 0\right]P(D_{\hat{f}}<0)\bigg] }  \exp{ \E\left[   -\max_{1\leq j\leq M} t\left|\frac{\mathrm{d}}{\mathrm{d}x}f_{Y|X}(y_{j},x_{j})\right| \right]  }\nonumber\\
    =&\exp{ \left(   -\max_{1\leq j\leq M} t\left|\frac{\mathrm{d}}{\mathrm{d}x}f_{Y|X}(y_{j},x_{j})\right| \right)  } \bigg[\exp{ \E\left[  t\left|\frac{\mathrm{d}}{\mathrm{d}x}\hat{f}_{Y|X}(y_{j^*},x_{j^*})\right| \mid D_{\hat{f}}\ge 0\right]  } \bigg]^{P(D_{\hat{f}}\ge 0)}\nonumber\\
    &\bigg[\exp{ \E\left[  t\left|\frac{\mathrm{d}}{\mathrm{d}x}\hat{f}_{Y|X}(y_{j^*},x_{j^*})\right| \mid D_{\hat{f}}< 0\right]  } \bigg] ^{P(D_{\hat{f}}<0)}.
\end{align}

In addition, we have 
\begin{align}\label{exp_part}
    &\exp{ \E\left[ \max_{1\leq j\leq M} t\left|\frac{\mathrm{d}}{\mathrm{d}x}\hat{f}_{Y|X}(y_{j},x_{j})\right|\right] }\nonumber\\
    \leq& \bigg[ \E\left[ \exp\left(  t\left|\frac{\mathrm{d}}{\mathrm{d}x}\hat{f}_{Y|X}(y_{j^*},x_{j^*})\right| \right) \mid D_{\hat{f}}\ge 0 \right]   \bigg]^{P(D_{\hat{f}}\ge 0)}
    \bigg[ \E\left[ \exp\left(  t\left|\frac{\mathrm{d}}{\mathrm{d}x}\hat{f}_{Y|X}(y_{j^*},x_{j^*})\right|\right) \mid D_{\hat{f}}< 0 \right]  \bigg] ^{P(D_{\hat{f}}<0)}\nonumber\\
    =&\bigg[ \E\left[ \exp\left(  t\frac{\mathrm{d}}{\mathrm{d}x}\hat{f}_{Y|X}(y_{j^*},x_{j^*}) \right) \mid D_{\hat{f}}\ge 0 \right]   \bigg]^{P(D_{\hat{f}}\ge 0)}
    \bigg[ \E\left[ \exp\left(  -t\frac{\mathrm{d}}{\mathrm{d}x}\hat{f}_{Y|X}(y_{j^*},x_{j^*})\right) \mid D_{\hat{f}}< 0 \right]  \bigg] ^{P(D_{\hat{f}}<0)}\nonumber\\
    \leq&\sum_{j=1}^{M} \bigg[ \E\left[ \exp\left(  t\frac{\mathrm{d}}{\mathrm{d}x}\hat{f}_{Y|X}(y_{j},x_{j}) \right) \mid D_{\hat{f}}\ge 0 \right]   \bigg]^{P(D_{\hat{f}}\ge 0)}
    \cdot \bigg[ \E\left[ \exp\left(  -t\frac{\mathrm{d}}{\mathrm{d}x}\hat{f}_{Y|X}(y_{j},x_{j})\right) \mid D_{\hat{f}}< 0 \right]  \bigg] ^{P(D_{\hat{f}}<0)}\nonumber\\
    \leq&M \max_{1\leq j\leq M} \bigg[ \E\left[ \exp\left(  t\frac{\mathrm{d}}{\mathrm{d}x}\hat{f}_{Y|X}(y_{j},x_{j}) \right) \mid D_{\hat{f}}\ge 0 \right]   \bigg]^{P(D_{\hat{f}}\ge 0)}
    \cdot \bigg[ \E\left[ \exp\left(  -t\frac{\mathrm{d}}{\mathrm{d}x}\hat{f}_{Y|X}(y_{j},x_{j})\right) \mid D_{\hat{f}}< 0 \right]  \bigg] ^{P(D_{\hat{f}}<0)}.
\end{align}

At the next step, we need to approximate $\E\left[ \exp{ (t \frac{\mathrm{d}}{\mathrm{d}x}\hat{f}_{Y|X}(y_{j},x_{j}))} \right]$ using the second-order Taylor expansion. Now assuming that $z:=\frac{\mathrm{d}}{\mathrm{d}x}\hat{f}_{Y|X}(y_{j},x_{j})$
which is a random variable, and 
\begin{align}\label{part_meanvar}
    &u_{lc1}:=\E[z\mid D_{\hat{f}}\ge 0],~v_{lc1}:=\Var [z\mid D_{\hat{f}}\ge 0]=\E[(z-u_{lc1})^2\mid D_{\hat{f}}\ge 0],\nonumber\\
    &u_{lc2}:=\E[z\mid D_{\hat{f}}< 0],~v_{lc2}:=\Var [z\mid D_{\hat{f}}<0]=\E[(z-u_{lc2})^2\mid D_{\hat{f}}< 0].
\end{align}
Then, the second-order Taylor expansion around the means and variances, as shown in equation \eqref{part_meanvar}, can be obtained, as following,
\begin{align}\label{condi_expecta1}
    &\E\left[ \exp{ (t \frac{\mathrm{d}}{\mathrm{d}x}\hat{f}_{Y|X}(y_{j},x_{j}))} \mid D_{\hat{f}}\ge 0 \right]=\E\left[ \exp{ (t z)} \mid D_{\hat{f}}\ge 0 \right]\nonumber\\
    \approx&\E\left[ \exp(t u_{lc1})+\frac{\mathrm{d}}{\mathrm{d}z}\exp(t u_{lc1})(z-u_{lc1})  +\frac{1}{2}\frac{\mathrm{d}^2}{\mathrm{d}z^2}\exp(t u_{lc1}) (z-u_{lc1})^2 \mid D_{\hat{f}}\ge 0\right]\nonumber\\
    =&\exp(t u_{lc1})+\frac{1}{2}\frac{\mathrm{d}^2}{\mathrm{d}z^2}\exp(t u_{lc1}) v_{lc1}=\exp(t u_{lc1})+\frac{1}{2}t^2\exp(t u_{lc1}) v_{lc1},
\end{align}
and
\begin{align}\label{condi_expecta2}
    &\E\left[ \exp{ (-t \frac{\mathrm{d}}{\mathrm{d}x}\hat{f}_{Y|X}(y_{j},x_{j}))} \mid D_{\hat{f}}< 0 \right]=\E\left[ \exp{ (-t z)} \mid D_{\hat{f}}< 0 \right]\nonumber\\
    \approx&\E\left[ \exp(-t u_{lc2})+\frac{\mathrm{d}}{\mathrm{d}z}\exp(-t u_{lc2})(z-u_{lc2})  +\frac{1}{2}\frac{\mathrm{d}^2}{\mathrm{d}z^2}\exp(-t u_{lc2}) (z-u_{lc2})^2 \mid D_{\hat{f}}> 0\right]\nonumber\\
    =&\exp(-t u_{lc2})+\frac{1}{2}\frac{\mathrm{d}^2}{\mathrm{d}z^2}\exp(-t u_{lc2}) v_{lc2}=\exp(-t u_{lc2})+\frac{1}{2}t^2\exp(-t u_{lc2}) v_{lc2}.
\end{align}
Thus, through plugging the conditional expectation \eqref{condi_expecta1} and \eqref{condi_expecta2} into equation \eqref{exp_part} ,we can have 
\begin{align}\label{impor_step}
     &\exp{ \E\left[ \max_{1\leq j\leq M} t\left|\frac{\mathrm{d}}{\mathrm{d}x}\hat{f}_{Y|X}(y_{j},x_{j})\right|\right] }\nonumber\\
     \lesssim& M\! \max_{1\leq j\leq M}\!  \left[ \exp(t u_{lc1})\!+\!\frac{1}{2}t^2\exp(t u_{lc1}) v_{lc1} \right]^{P(D_{\hat{f}}\ge 0)} \! \left[ \exp(-t u_{lc2})\!+\!\frac{1}{2}t^2\exp(-t u_{lc2}) v_{lc2} \right]^{P(D_{\hat{f}}< 0)}\nonumber\\
     \leq&M\max_{1\leq j\leq M} \bigg[ \exp\left( t |u_{lc1}|P(D_{\hat{f}}\ge 0)+t |u_{lc2}|P(D_{\hat{f}}< 0) \right) \nonumber\\
     &\cdot\left[ 1+\frac{1}{2}t^2 v_{lc1} \right]^{P(D_{\hat{f}}\ge 0)}
     \left[ 1+\frac{1}{2}t^2 v_{lc2} \right]^{P(D_{\hat{f}}< 0)}  \bigg]\nonumber\\
     =&M\max_{1\leq j\leq M} \bigg[ \exp\left( t|\E[z\mid D_f\ge 0]| P(D_{\hat{f}}\ge 0)+t |\E[z\mid D_f< 0]|P(D_{\hat{f}}< 0) \right) \nonumber\\
     &\cdot\left[ 1+\frac{1}{2}t^2 v_{lc1} \right]^{P(D_{\hat{f}}\ge 0)}
     \left[ 1+\frac{1}{2}t^2 v_{lc2} \right]^{P(D_{\hat{f}}< 0)}  \bigg]\nonumber\\
     \leq& {\color{black}M\max_{1\leq j\leq M} \bigg[ \exp\left( t\E[|z|\mid D_{\hat{f}}\ge 0] P(D_{\hat{f}}\ge 0)+t \E[|z|\mid D_{\hat{f}}< 0]P(D_{\hat{f}}< 0) \right) } \nonumber\\
     &\cdot\left[ 1+\frac{1}{2}t^2 v_{lc1} \right]^{P(D_{\hat{f}}\ge 0)}
     \left[ 1+\frac{1}{2}t^2 v_{lc2} \right]^{P(D_{\hat{f}}< 0)}  \bigg]\nonumber\\
     =&M\max_{1\leq j\leq M} \bigg[ \exp\left( t\E[|z|] \right) \left[ 1+\frac{1}{2}t^2 v_{lc1} \right]^{P(D_{\hat{f}}\ge 0)}
     \left[ 1+\frac{1}{2}t^2 v_{lc2} \right]^{P(D_{\hat{f}}< 0)}  \bigg].
\end{align}

Now, by substituting equation \eqref{impor_step} into equation \eqref{expon_bias}, we have 
\begin{align*}
    &\exp{ \E\left[ \max_{1\leq j\leq M} t\left|\frac{\mathrm{d}}{\mathrm{d}x}\hat{f}_{Y|X}(y_{j},x_{j})\right|-\max_{1\leq j\leq M} t\left|\frac{\mathrm{d}}{\mathrm{d}x}f_{Y|X}(y_{j},x_{j})\right|\right] }\nonumber\\
    \lesssim& M\max_{1\leq j\leq M} \bigg[ \exp\left( t\E[|z|] \right) \left[ 1+\frac{1}{2}t^2 v_{lc1} \right]^{P(D_{\hat{f}}\ge 0)}
     \left[ 1+\frac{1}{2}t^2 v_{lc2} \right]^{P(D_{\hat{f}}< 0)}  \bigg]\\
     &\cdot \exp{ \left(   -\E\left[\max_{1\leq j\leq M} t\left|\frac{\mathrm{d}}{\mathrm{d}x}f_{Y|X}(y_{j},x_{j})\right| \right] \right)   }\\
     =&M\max_{1\leq j\leq M} \bigg[ \exp\left( t\E\left[\left|z-\frac{\mathrm{d}}{\mathrm{d}x}f_{Y|X}(y_{j},x_{j})+\frac{\mathrm{d}}{\mathrm{d}x}f_{Y|X}(y_{j},x_{j})\right|\right] \right) \\
     &\cdot\left[ 1+\frac{1}{2}t^2 v_{lc1} \right]^{P(D_{\hat{f}}\ge 0)}
     \left[ 1+\frac{1}{2}t^2 v_{lc2} \right]^{P(D_{\hat{f}}< 0)}  \bigg]\exp{ \left(   -\max_{1\leq j\leq M} t\left|\frac{\mathrm{d}}{\mathrm{d}x}f_{Y|X}(y_{j},x_{j})\right| \right)  }\\
     \leq&M\max_{1\leq j\leq M} \bigg[ \exp\left( t\E\left[\left|z-\frac{\mathrm{d}}{\mathrm{d}x}f_{Y|X}(y_{j},x_{j})\right|\right]+t\E\left[\left|\frac{\mathrm{d}}{\mathrm{d}x}f_{Y|X}(y_{j},x_{j})\right|\right] \right)\\
     &\cdot\left[ 1+\frac{1}{2}t^2 v_{lc1} \right]^{P(D_{\hat{f}}\ge 0)}
     \left[ 1+\frac{1}{2}t^2 v_{lc2} \right]^{P(D_{\hat{f}}< 0)}  \bigg]\exp{ \left(   -\max_{1\leq j\leq M} t\left|\frac{\mathrm{d}}{\mathrm{d}x}f_{Y|X}(y_{j},x_{j})\right| \right)  }\\
     \leq&M\max_{1\leq j\leq M} \bigg[ \exp\left( t\E\left[\left|z-\frac{\mathrm{d}}{\mathrm{d}x}f_{Y|X}(y_{j},x_{j})\right|\right] \right)\left[ 1+\frac{1}{2}t^2 v_{lc1} \right]^{P(D_{\hat{f}}\ge 0)}
     \left[ 1+\frac{1}{2}t^2 v_{lc2} \right]^{P(D_{\hat{f}}< 0)}  \bigg]\\
     &\cdot \exp\left( t\E\left[\max_{1\leq j\leq M}\left|\frac{\mathrm{d}}{\mathrm{d}x}f_{Y|X}(y_{j},x_{j})\right|\right] \right)\exp{ \left(   -\E\left[\max_{1\leq j\leq M} t\left|\frac{\mathrm{d}}{\mathrm{d}x}f_{Y|X}(y_{j},x_{j})\right| \right] \right)  }\\
     =&M\max_{1\leq j\leq M} \bigg[ \exp\left( t\E\left[\left|z-\frac{\mathrm{d}}{\mathrm{d}x}f_{Y|X}(y_{j},x_{j})\right|\right] \right)\left[ 1+\frac{1}{2}t^2 v_{lc1} \right]^{P(D_{\hat{f}}\ge 0)}
     \left[ 1+\frac{1}{2}t^2 v_{lc2} \right]^{P(D_{\hat{f}}< 0)}  \bigg]\\
     \leq& M \exp\left( \max_{1\leq j\leq M}  t\left(\E\left[(z-\frac{\mathrm{d}}{\mathrm{d}x}f_{Y|X}(y_{j},x_{j}))^2\right]\right)^{1/2} \right) \\
     &\cdot \max\left\{  \max_{1\leq j\leq M}\left|1+\frac{1}{2}t^2v_{lc1}\right|, \max_{1\leq j\leq M}\left|1+\frac{1}{2}t^2v_{lc2}\right|  \right\}.
\end{align*}
Let $C_2=\max_{1\leq j\leq M}|v_{lc1}|$ and $C_3=\max_{1\leq j\leq M}|v_{lc2} |$.
By taking logarithm and then dividing $t$ to both sides of the above inequality, we have 
\begin{align*}
    &\E\left[ \max_{1\leq j\leq M} \left|\frac{\mathrm{d}}{\mathrm{d}x}\hat{f}_{Y|X}(y_{j},x_{j})\right|-\max_{1\leq j\leq M} \left|\frac{\mathrm{d}}{\mathrm{d}x}f_{Y|X}(y_{j},x_{j})\right|\right] \\
    \lesssim& \max_{1\leq j\leq M}  \left(\E\left[(z-\frac{\mathrm{d}}{\mathrm{d}x}f_{Y|X}(y_{j},x_{j}))^2\right]\right)^{1/2}+\frac{\log(M)}{t}+\frac{\log(1+C_4 t^2)}{t}\\
    \approx&\max_{1\leq j\leq M}  \left(\E\left[(z-\frac{\mathrm{d}}{\mathrm{d}x}f_{Y|X}(y_{j},x_{j}))^2\right]\right)^{1/2},
\end{align*}
where $t$ is arbitrary positive value and $C_4=\max\{C_2,C_3 \}$.

Since 
\begin{align*}
    \E\left[\max_{\substack{(x,y),(x',y')\\ \in D_{X}\times D_{Y}:\\ \| (x,y)-(x',y') \|\leq \frac{D_B}{M} }}\left| \frac{\mathrm{d}}{\mathrm{d}x}\hat{f}_{Y|X}(y,x)-\frac{\mathrm{d}}{\mathrm{d}x}\hat{f}_{Y|X}(y',x') \right|\right]\leq \hat{L}_{df}\frac{D_{B}}{M},
\end{align*}
where $\hat{L}_{df}$ is the LC of $\frac{\mathrm{d}}{\mathrm{d}x}\hat{f}_{Y|X}(y,x)$ among $D_{X}\times D_{Y}$, according to Theorem~\ref{le_mse}, we have
\begin{align*}
     &\E\left[\max_{(x,y)\in D_{X}\times D_{Y}}\left|\frac{\mathrm{d}}{\mathrm{d}x}\hat{f}_{Y|X}(y,x)\right|\right]-\max_{(x,y)\in D_{X}\times D_{Y}} \left|\frac{\mathrm{d}}{\mathrm{d}x}f_{Y|X}(y,x)\right|  \nonumber\\
     \lesssim& \max_{1\leq j\leq M}  \left(\E\left[(\frac{\mathrm{d}}{\mathrm{d}x}\hat{f}_{Y|X}(y_{j},x_{j})-\frac{\mathrm{d}}{\mathrm{d}x}f_{Y|X}(y_{j},x_{j}))^2\right]\right)^{1/2}
     \lesssim \epsilon_3^{1/2}.
\end{align*}

\end{proof}

\subsection{Validation Study of the LC Estimation}\label{supp_validation} 
In this section, we apply our estimation algorithm to two case studies. We show that our approach gives better results in comparison with the standard selection of bandwidths using Scott's formula. The CV method for selecting the bandwidths did not provide any solution for the optimisation in equation \eqref{cross-valid} within 6 hours. 

\begin{exmp}\label{examp1}
Consider a univariate $(Y|X)$ such that $Y=aX+w$, where $w$ is a Gaussian noise with mean $\mu$ and variance $\sigma^2$. The CoDF $f_{Y|X}(y,x)$ is 
\begin{equation*}
    f_{Y|X}(y,x)=\frac{1}{\sigma\sqrt{2\pi}}
    \exp{\left[-\frac{1}{2}\frac{ ( y-ax )^{2} }{ \sigma^2}\right]}.
\end{equation*}
We fix the parameters $a=0.5$, $\mu=0$, $\sigma=1$, and the domain $D_X=[-1,1]$ and $D_Y=[-4.38,4.24]$. We assume that Assumption~\ref{ass_1} holds with $C_{f}=1$ and $C_{b1}=C_{b2} = 0.5$. We run Algorithm~\ref{algo_enviro} with $m = 20$.

\begin{figure}[hbt!]
\centering
\includegraphics[width=\textwidth]{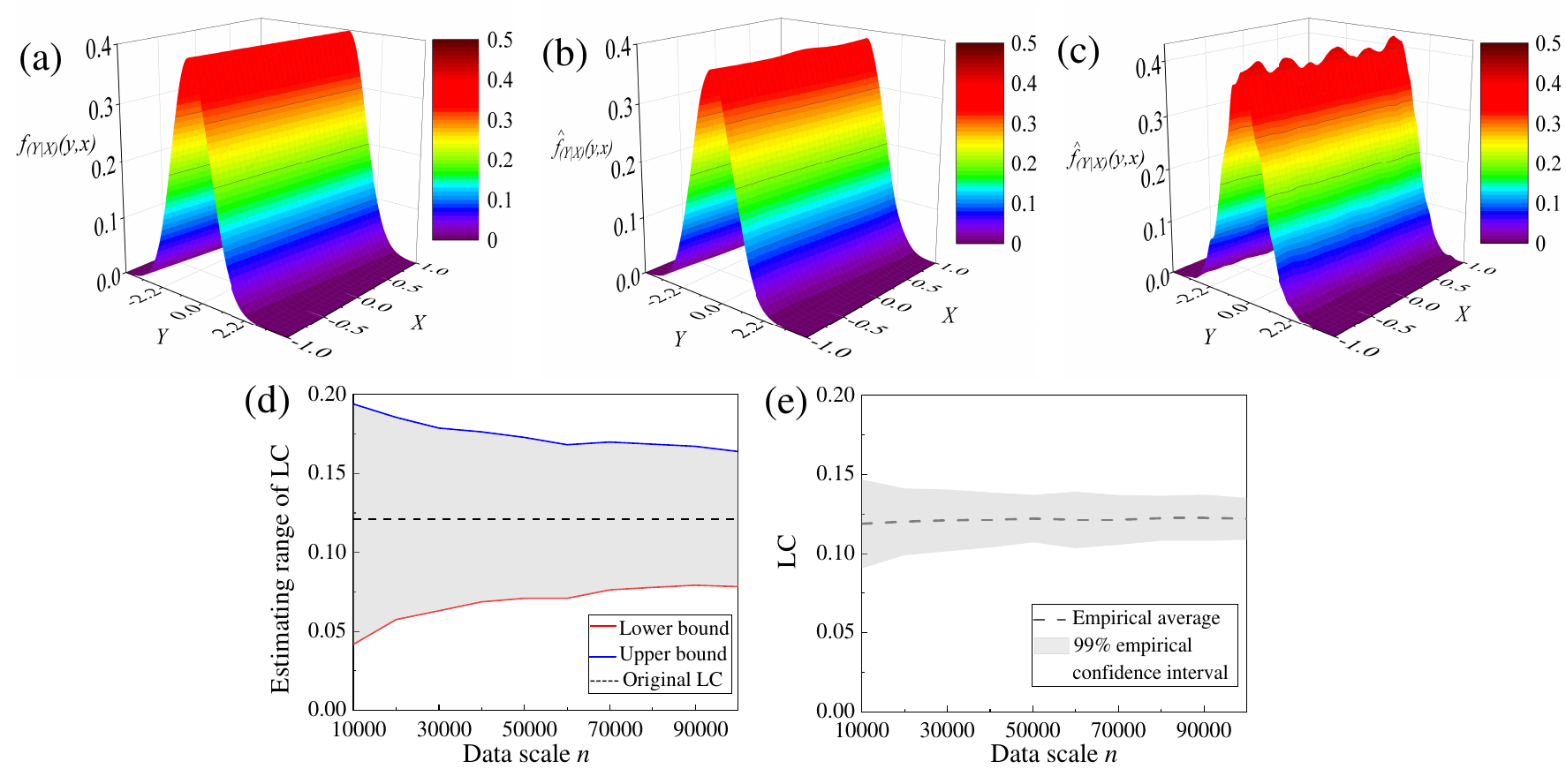}
\caption{(a) The original CoDF. (b) Estimated CoDF with $h_{\mathsf x}=h_{\mathsf y}=n^{-\frac{1}{8}}$ and data scale $n=6\times 10^{4}$.
(c) Estimated CoDF with bandwidth from the Scott's formula and data scale $n=6\times 10^{4}$.
(d) Asymptotic bound on the original LC provided by Theorem~\ref{LC_main} as a function of data scale $n$. The dashed line is the original LC of the CoDF, $L=0.1210$.
(e) The estimated LC averaged over 150 computations with the grey area indicating the 99\% empirical confidence interval (3 times the empirical standard deviation from the mean).}
\label{case1}
\end{figure}

Fig.~\ref{case1} shows the original CoDF together with the estimated CoDF using the bandwidth from Theorem~\ref{LC_main} and from the Scott's formula. The estimated CoDF based on the Scott's formula in Fig.~\ref{case1}(c) has many values larger than $0.42$ which are much higher than the largest value of the original CoDF, at $f_{Y|X}(y,x)=0.3989$ as shown in Fig.~\ref{case1}(a).
This estimation presents bad smoothness in comparison with the estimation using our theoretical bandwidths $h_{\mathsf x}=h_{\mathsf y}=n^{-\frac{1}{8}}$, which has a smooth peak around $0.381$ as shown in Fig.~\ref{case1}(b).
Thus, selecting of bandwidths according to the discussion in Section~\ref{theo_onedimen} gives a better smoothness than Scott's formula.

Fig.~\ref{case1}(d) gives the
asymptotic bounds on the original LC provided by Theorem~\ref{LC_main} as a function of data scale $n\in [10^4, 10^5]$.
For example, $L\in [0.042,0.194]$ using $n=10^4$,
$L\in [0.071,0.168]$ using
$n=6\times 10^4$,
and
$L\in [0.078,0.164]$
using $n=10^5$,
which confirms asymptotic convergence of the bound as a function of $n$.
The empirical mean and the 99\% empirical confidence interval are shown in Fig.~\ref{case1}(e) using 150 runs of the algorithm for each $n$.
All the values are below the analytical upper bound shown in Fig.~\ref{case1}(d).

\end{exmp}

\begin{exmp}\label{examp2}
\begin{figure}[hbt!]
\centering
\includegraphics[width=\textwidth]{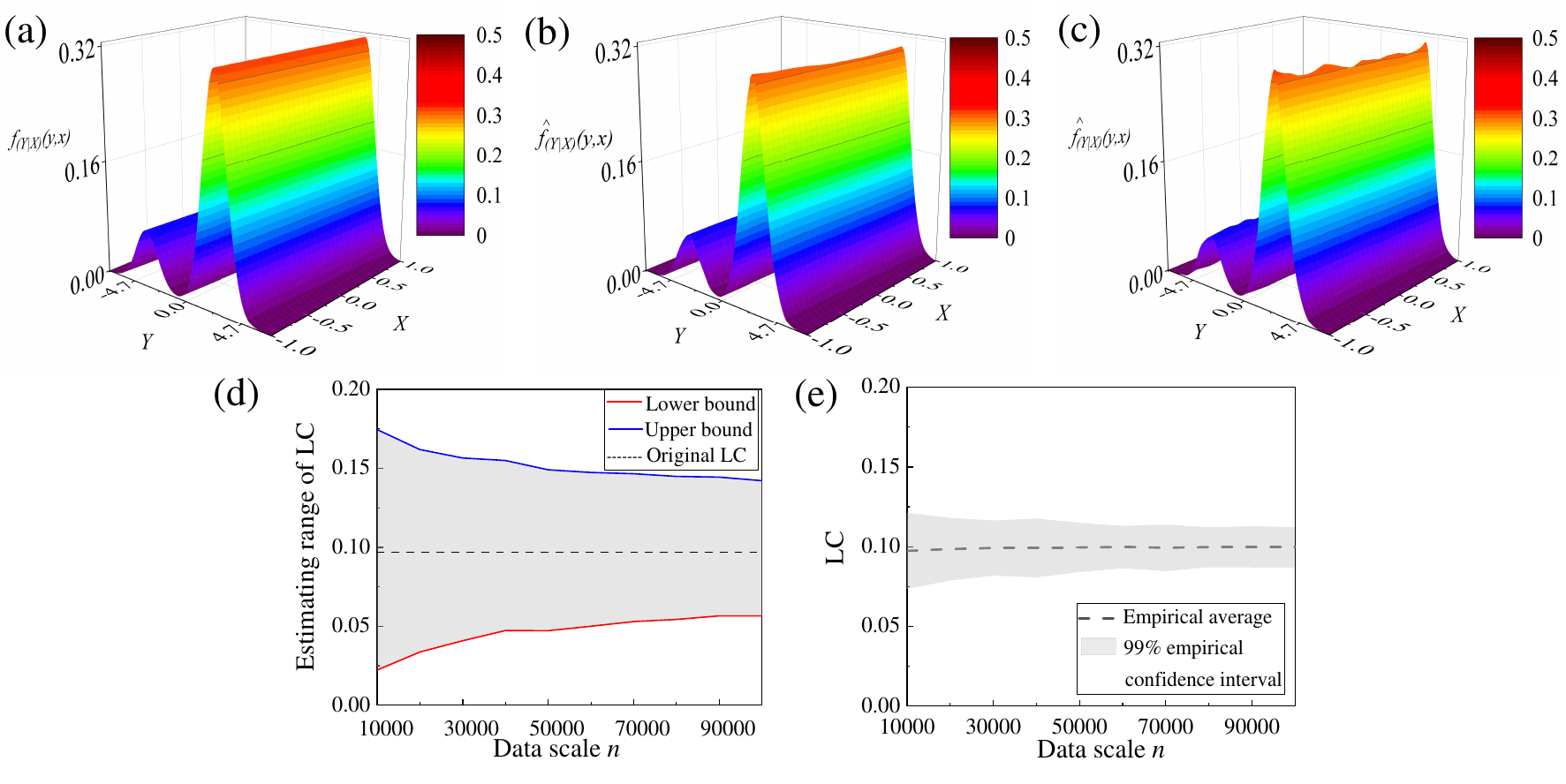}
\caption{(a) The original CoDF. (b) Estimated CoDF with $h_{\mathsf x}=h_{\mathsf y}=n^{-\frac{1}{8}}$ and data scale $n=6\times 10^{4}$.
(c) Estimated CoDF with bandwidth from the Scott's formula and data scale $n=6\times 10^{4}$.
(d) Asymptotic bound on the original LC provided by Theorem~\ref{LC_main} as a function of data scale $n$. The dashed line is the original LC of the CoDF, $L=0.0968$.
(e) The estimated LC averaged over 150 computations with the grey area indicating the 99\% empirical confidence interval (3 times the empirical standard deviation from the mean).}
\label{case2}
\end{figure}
Consider a univariate $(Y|X)$ such that $Y=aX+\delta w_{1}+(1-\delta)w_{2}$, where
$w_1$ and $w_{2}$ have Gaussian distributions with means $\mu_1$ and $\mu_2$ and variances $\sigma_1^{2}$ and $\sigma_2^{2}$. The variable $\delta\in\{0,1\}$ has Bernoulli distribution with success probability $P(\sigma=1)=p$. 
The CoDF $f_{Y|X}(y,x)$ is
\begin{align*}\label{examp2_density}
    f_{Y|X}(y,x)=& \frac{p}{\sigma_1\sqrt{2\pi}}\!\exp{( -\frac{1}{2\sigma_1^2}(y-ax-\mu_1)^{2} ) }
    \!+\!\frac{1-p}{\sigma_2\sqrt{2\pi}}\exp{ ( -\frac{1}{2\sigma_2^2}(y-ax-\mu_2)^{2} ) }.
\end{align*}
We fix the parameters $a=0.5$, $\mu_1=3$, $\mu_2=-3$, $\sigma_1=\sigma_2=1$, $p=0.8$, and the domain $D_X=[-1,1]$ and  $D_Y=[-7.177,6.965]$.
We assume that Assumption~\ref{ass_1} holds with $C_{f}=1$ and $C_{b1}=C_{b2} = 0.5$. We run Algorithm~\ref{algo_enviro} with $m = 20$.

Fig.~\ref{case2} shows the original CoDF together with the estimated CoDF using the bandwidth from Theorem~\ref{LC_main} and from the Scott's formula.
The estimated CoDF based on the bandwidths $h_{\mathsf x}=h_{\mathsf y}=n^{-\frac{1}{8}}$ as shown in Fig.~\ref{case2}(b), is more similar to the original CoDF in Fig.~\ref{case2}(a) and shows better smoothness in comparison with the estimation based on Scott's formula in Fig.~\ref{case2}(c).


Fig.~\ref{case2}(d) gives the
asymptotic bounds on the original LC $L=0.0968$ provided by Theorem~\ref{LC_main} as a function of data scale $n\in [10^4, 10^5]$.
For example, $L\in [0.022,0.175]$ using $n=10^4$,
$L\in [0.050,0.147]$ using
$n=6\times 10^4$,
and
$L\in [0.057,0.142]$
using $n=10^5$,
which confirms asymptotic convergence of the bound as a function of $n$.
The empirical mean and the 99\% empirical confidence interval are shown in Fig.~\ref{case2}(e) using 150 runs of the algorithm for each $n$.
All the values are below the analytical upper bound shown in Fig.~\ref{case2}(d).

\end{exmp}

\begin{exmp}
\label{example:3}
Consider a bivariate $(Y|X)$ with $Y=AX+W$ with $W$ having Gaussian distribution with mean $\mu$ and covariance matrix $\Sigma$. The CoDF is \begin{equation}\label{two_dimenscdf}
   f_{Y|X}(\bs{y},\bs{x})=\frac{1}{ 2\pi\sqrt{|\Sigma|} }\exp{ ( -\frac{1}{2}( \bs{y}-A\bs{x} )^{T} \Sigma^{-1} ( \bs{y}-A\bs{x} )  ) }.
\end{equation}
We consider two cases, one with small LC and another with large LC.
\begin{figure}[hbt!]
\centering
\includegraphics[width=\textwidth]{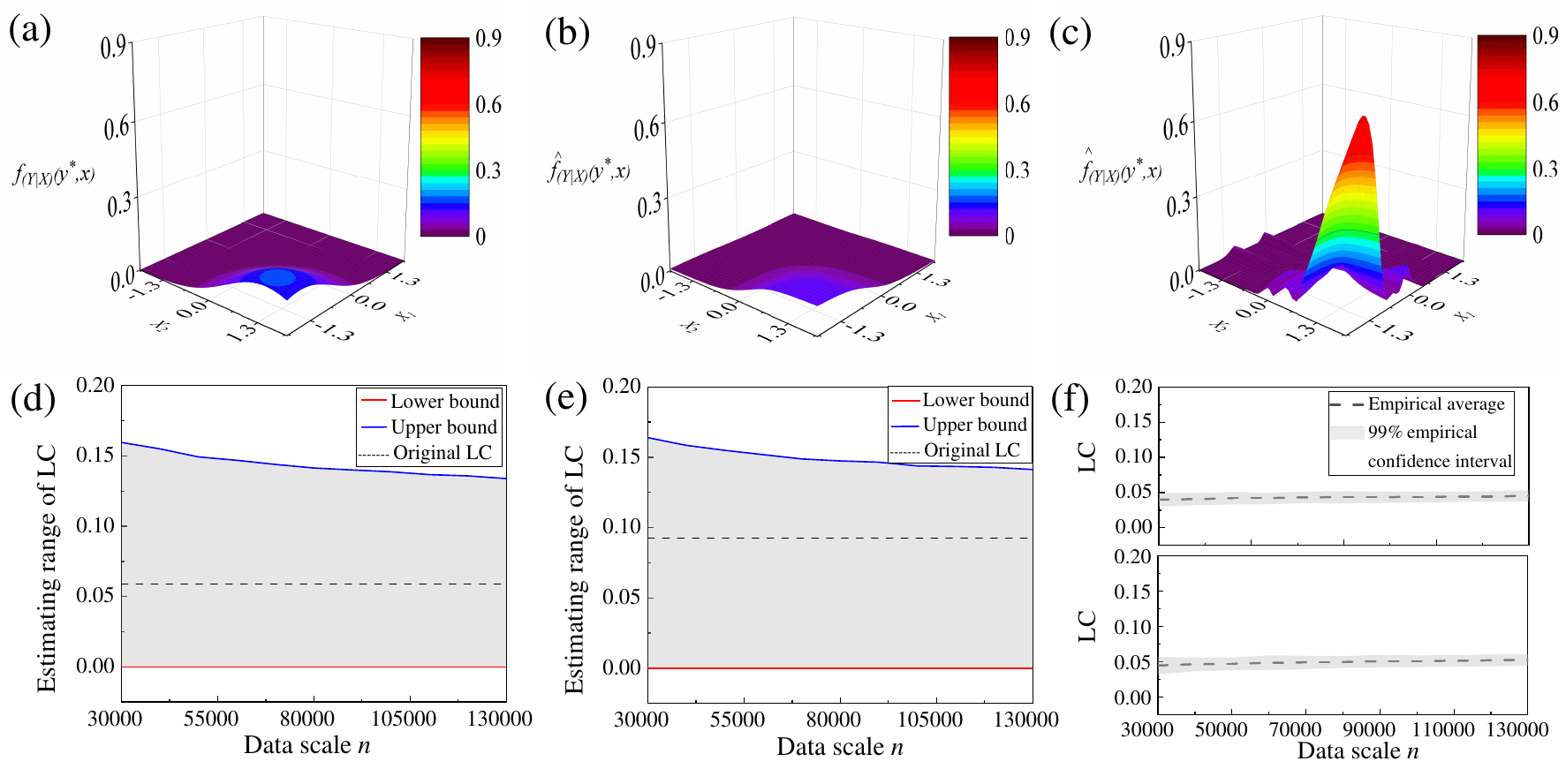}
\caption{(a) The original CoDF at point $\bs{y}^{*}=[1.4615;-1.6154]$. (b) Estimated CoDF with $h_{\mathsf x i}=h_{\mathsf y j}=n^{-\frac{1}{10}}$, $i,j=1,2$, with data scale $n=1.3\times 10^{5}$ at point $\bs{y}^{*}$.
(c) Estimated CoDF with bandwidth from the Scott's formula with data scale $n=1.3\times 10^{5}$ at point $\bs{y}^{*}$.
(d) Asymptotic bound on the original LC provided in Section~\ref{theo_twodimen} as a function of data scale $n$ for $ D_X = D_Y = [-0.2,0.2]^2$. The dashed line is the original LC of the CoDF, which is $L=0.0588$.
(e) Asymptotic bound on the original LC provided in Section~\ref{theo_twodimen} as a function of data scale $n$ for $ D_X = [0.2,0.6]^2$ and $D_Y = [-0.2,0.2]^2$. The dashed line is the original LC of the CoDF, which is $L=0.0925$.
(f) The estimated LCs averaged over 150 computations with the grey areas indicating the 99\% empirical confidence intervals (3 times the empirical standard deviation from the mean).
}\label{case3}
\end{figure}

\textbf{Case 1:} 
We fix the parameters
$
A=\Sigma=\begin{bmatrix}
1 & 0  \\
0 & 1 
\end{bmatrix}$, $\mu=\begin{bmatrix}
0   \\
0  
\end{bmatrix}$.
We assume that Assumption~\ref{asstwod_1} holds with
$C_{f}=0.5$,
$C_{ij}=0.5$, and $C_{xi}=0.5$, $i,j=1,2$.
We run Algorithm~\ref{algo_enviro} with $m = 20$.
Fig.~\ref{case3} shows the original CoDF together with the estimated CoDF using the bandwidth from Section~\ref{theo_twodimen} and from the Scott's formula. For better visualisation, we provide the estimation at $\bs{y}^{*}=[1.4615;-1.6154]$ under different bandwidths.
The original CoDF in Fig.~\ref{case3}(a) has a similar shape with its estimation based on bandwidth $h_{\mathsf x i}=h_{\mathsf y j}=n^{-\frac{1}{10}}$, $i,j=1,2$, as shown in Fig.~\ref{case3}(b). The largest values of $f_{Y|X}(\bs{y}^{*},\bs{x})$ in Fig.~\ref{case3}(a) and $\hat{f}_{Y|X}(\bs{y}^{*},\bs{x})$ in Fig.~\ref{case3}(b) are respectively $0.159$ and $0.111$. By contrast, the difference between the original CoDF in Fig.~\ref{case3}(a) and its estimation in Fig.~\ref{case3}(c) is large. The largest value of $\hat{f}_{Y|X}(\bs{y}^{*},\bs{x})$ in Fig.~\ref{case3}(c) is around $0.7$, which is much larger than the original CoDF. 
In addition, the CoDF in Fig.~\ref{case3}(a) and estimation in Fig.~\ref{case3}(b) are smoother than the estimation in
Fig.~\ref{case3}(c). 

Fig.~\ref{case3}(d) gives the
asymptotic bounds on the original LC provided by the inequality~\eqref{main_twodimen} as a function of data scale $n$ for $ D_X = D_Y = [-0.2,0.2]^2$.
For example, $L\in [0,0.160]$ using $n=3\times 10^{4}$,
$L\in [0,0.142]$ using $n=8\times 10^4$,
and
$L\in [0,0.134]$
using $n=1.3\times 10^{5}$,
which confirms asymptotic convergence of the bound as a function of $n$.
The original LC in this domain is $0.059$ and belongs to the above estimated ranges.
A similar experiment is reported in Fig.~\ref{case3}(e) for the domain $ D_X = [0.2,0.6]^2$ and $D_Y = [-0.2,0.2]^2$.
%
The empirical means and the 99\% empirical confidence intervals are shown in Fig.~\ref{case3}(f) using 150 runs of the algorithm for each $n$.
All the values are below the analytical upper bound shown in Fig.~\ref{case3}(d) and (e).

\textbf{Case 2:} 
We fix the parameters $
A=\begin{bmatrix}
1 & 0  \\
0 & 1 
\end{bmatrix}$, $\mu=\begin{bmatrix}
0   \\
0  
\end{bmatrix}$, and $\Sigma=\begin{bmatrix}
0.2 & 0  \\
0 & 0.2 
\end{bmatrix}$.
We assume that Assumption~\ref{asstwod_1} holds with
$C_{f}=1$,
$C_{ij}=10$, and $C_{xi}=10$, $i,j=1,2$.
This choice of $\Sigma$ makes the LC larger, which in turn requires a larger data scale $n$ for the estimation. The original LC on the domain $ D_X = [0,0.2]^2$ and $D_Y = [-0.2,-0.1]^2$ is $L=1.04$. The estimated LC and the asymptotic intervals computed using our approach are reported in Table~\ref{tab} for $n=5\times 10^6$, $n=5\times 10^7$ and $n=10^8$.

\begin{table}[hbt!]
\begin{center}
\begin{tabular}{|c|c|c|c|}
\hline
{Data scale $n$}  & {$5\times 10^{6}$}& {$5\times 10^{7}$} & {$1\times 10^{8}$} 
\\ \hline 
Range of LC &[0,1.428] & [0.143,1.238] & [0.247,1.2]\\
$\hat L$    &$0.563$ & $0.691$ &  $0.724$\\
\hline
\end{tabular}
\end{center}
\caption{The estimated LC and the asymptotic ranges computed using our approach for Example~\ref{example:3} (case 2) for different values of $n$.
The range contains the true LC $L=1.04$.}
\label{tab}
\end{table}
\end{exmp}

\begin{exmp}\label{example:4}
We consider the 7-dimensional model of a BMW 320i car \citep{althoff2019commonroad} reported also in Appendix~\ref{auto_car} to validate the effectiveness of our proposed method. We assume the space is $[x_{1}(k);x_{2}(k)]\in [0.8,1.2]^2$, $x_{3}(k)\in[0,0.3]$, $[x_{4}(k);x_{5}(k)]\in[0,0.1]^2$,
$x_{6}(k)\in [0.5,1]$, and $x_{7}(k)\in [0,0.2]$. The dynamics are affected by the standard normal variable.
%
We take the bound $A_{i} = 20$ with bandwidths $h_{\mathsf x i}=h_{\mathsf y j}=n^{-\frac{1}{14}}$, $i,j\in\{1,\ldots,7\}$, to generate the CoDF $\hat{T}_{i}(x_{i}(k+1)|\bs{x}(k))$, $i\in\{1,\ldots,7\}$ and estimate the LC $L_{i}$.
Table~\ref{tab2} shows the values of $L_{1}$, $L_{3}$ and $L_{7}$ with respect to $x_{3}(k)\in[0,0.1]$ for fixed $x_{1}(k)=x_{2}(k)=1$, $x_{4}(k)=x_{5}(k)=0.05$, $x_{6}(k)=0.8$, $x_{7}(k)=0.1$ and $\bs{x}(k+1) = \bs{0}$.
They are all located in the estimated ranges that are decreasing with increasing data scale $n$.
\begin{table}[hbt!]
\begin{center}
\begin{tabular}{|c|c|c|c|c|}
\hline
{Data scale $n$} & Value of $L_1$, $L_3$ and $L_7$ & Range of $L_{1}$ & Range of $L_{3}$ & Range of $L_{7}$
\\ \hline 
{$1\times 10^{7}$} &   & $[0,1.0582]$ & $[0,1.0198]$ & $[0,1.0062]$\\
{$5\times 10^{7}$} & $0.4319$, $0.3128$, $0.3128$  &$[0,0.8672]$ & $[0,0.8215]$ &  $[0,0.8034]$\\
{$1\times 10^{8}$} &   &$[0,0.7993]$ & $[0,0.7503]$ & $[0,0.7294]$\\
\hline
\end{tabular}
\end{center}
\caption{
The estimated LC of $T_{i}(x_{i}(k+1)|\bs{x}(k))$, $i\in\{1,3,7\}$, and their intervals computed using our approach for Example~\ref{example:4} for different values of $n$.}
\label{tab2}
\end{table}
\end{exmp}

\vspace{-0.8cm}

\subsection{Introduction of 7-Dimensional Autonomous Vehicle}\label{auto_car}
For $|x_{4}(k)|<0.1$:
\begin{align*}
    &x_{i}(k+1)=x_{i}(k)+\tau a_{i}+0.5 w_{i}(k),~i\in \{1,\ldots,7 \} \setminus \{3,4\},\\
    &x_{3}(k+1)=x_{3}(k)+\tau Sat_{1}(v_{1}(k))+0.5w_{3}(k),\\
    &x_{4}(k+1)=x_{4}(k)+\tau Sat_2(v_{2}(k))+0.5w_{4}(k),
\end{align*}
and for $|x_{4}(k)|\ge0.1$:
\begin{align*}
    &x_{i}(k+1)=x_{i}(k)+\tau b_{i}+0.5 w_{i}(k),~i\in \{1,\ldots,7 \}\setminus \{3,4\},\\
    &x_{3}(k+1)=x_{3}(k)+\tau Sat_{1}(v_{1}(k))+0.5w_{3}(k),\\
    &x_{4}(k+1)=x_{4}(k)+\tau Sat_2(v_{2}(k))+0.5w_{4}(k),
\end{align*}
where,
\begin{align*}
    a_1=&x_{4}(k)\cos(x_{5}(k)),
    ~a_{2}=x_{4}(k)\sin(x_{5}(k)),
    ~a_{5}=\frac{x_{4}(k)}{l_{wb}}\tan(x_{3}(k)),\\
    a_{6}=&\frac{v_{2}(k)}{l_{wb}}\tan(x_{3}(k))+\frac{x_{4}(k)}{l_{wb}\cos^{2}(x_{3}(k))}v_{1}(k),
    ~a_7=0,\\
    b_1=&x_{4}(k)\cos( x_{5}(k)+x_{7}(k) 
    ),
    ~b_{2}=x_{4}(k)\sin(x_{5}(k)+x_{7}(k)),
    ~b_{5}=x_{6}(k),\\
    b_{6}=&\frac{\mu m}{I_{z}(l_{r}+l_{f})}( l_{f}C_{S,f}(gl_{r}-v_{2}(k)h_{cg})x_{3}(k)+(l_{r}C_{S,r}(gl_{f}+v_{2}(k)h_{cg})-l_{f}C_{S,f}(gl_{r}\\
    &-v_{2}(k)h_{cg}) )x_{7}(k)-( l^2_{f}C_{S,f}(gl_{r}-v_{2}(k)h_{cg})+l^2_{r}C_{S,r}(gl_{f}+v_{2}(k)h_{cg})  )\frac{x_{6}(k)}{x_{4}(k)} ),\\
    b_{7}=&\frac{\mu}{x_{4}(k)(l_r+l_f)}( C_{S,f}(gl_{r}-v_{2}(k)h_{cg})x_{3}(k)-( C_{S,r}(gl_{f}+v_{2}(k)h_{cg})\\
    &+C_{S,f}(gl_{r}-v_{2}(k)h_{cg}) )x_{7}(k)-( l_{f}C_{S,f}(gl_r-v_{2}(k)h_{cg})
    -l_{r}C_{S,r}(gl_{f}+v_{2}(k)h_{cg}) )\frac{x_{6}(k)}{x_{4}(k)} )-x_{6}(k).
\end{align*}
We consider the variables and parameters for a BMW 320i car, as shown in Table~\ref{tab:syspars}.
In addition, $Sat_{1}(\cdot)$ and $Sat_{2}(\cdot)$ are input saturation functions introduced by \cite{althoff2019commonroad}.
\begin{table}[t]
  \centering
  \caption{State variables and system parameters.}
  \vspace{-0pt}
  \small
  \begin{tabular}[t]{|c|c|l|}
    \hline
    \textbf{Variable}&\textbf{Value}&\textbf{Description}\\ \hline
    $x_1$, $x_2$&$\mathbb{R}$ & Position coordinates\\
    $x_3$, $x_4$&$\mathbb{R}$ & Steering angle, heading velocity\\
    $x_5$, $x_6$&$\mathbb{R}$ & Yaw angle, Yaw rate\\
    $x_7$&$\mathbb{R}$ & Slip angle\\[0.5ex]
    $v_1$, $v_{2}$&$0$, $0$ & The inputs of controlling the steering angle and heading velocity\\[0.5ex]
    $l_{wb}$&$2.5789$ & Wheelbase $[\text{kg}]$\\
    $m$&$1.093.3$& Total mass of the vehicle $[\text{kg}]$\\
    $\mu$&$1.0489$& Friction coefficient\\
    $l_f$&$1.156$&Distance from the front axle to centre of gravity (CoG) $[\text{m}]$\\
    $l_r$&$1.422$& Distance from the rear axle to CoG $[\text{m}]$\\
    $h_{cg}$&$0.574$& Hight of CoG $[\text{m}]$\\
    $I_{z}$&$1791.6$& The Moment of inertia for entire mass around $z$ axis $[\text{kg}~\text{m}^2]$ \\
    $C_{S,f}$&$20.89$& The front cornering stiffness coefficient $[1/\text{rad}]$ \\
    $C_{S,r}$&$20.89$& The rear cornering stiffness coefficient $[1/\text{rad}]$ \\
    $\tau$&$0.001$& The update period $[\text{s}]$ \\\hline
  \end{tabular}
  \label{tab:syspars}
\end{table}

\section{SUPPLEMENTARY of SECTION~\ref{datadriv_imdp}}

\subsection{Supplementary of Section~\ref{empirical_method}}\label{proof_sec41}


\subsection{Proof of Lemma~\ref{truevalue_distan}}
\begin{proof}
The proof is based on induction on $k$ and utilising the recursive relation in \eqref{upprob_path}.

\end{proof}

\subsection{Proof of Theorem~\ref{empirical_guarantee}}
\begin{proof}
Let $\hat{\Sigma}_{ss}$ be the finite abstraction of $\Sigma_{ss}$. Based on Theorem~\ref{SA13_bound} and Lemma~\ref{truevalue_distan}  we have 
\begin{align*}
    \left|P(\Sigma_{ss}\vDash \psi)-P(\Bar{\Sigma}_{ss}\vDash \psi)\right|\leq& \left| P(\Sigma_{ss}\vDash \psi)- P(\hat{\Sigma}_{ss}\vDash \psi)+ P(\hat{\Sigma}_{ss}\vDash \psi)-P(\Bar{\Sigma}_{ss}\vDash \psi) \right|\\
    \leq& \left| P(\Sigma_{ss}\vDash \psi)- P(\hat{\Sigma}_{ss}\vDash \psi)\right|+\left|P(\hat{\Sigma}_{ss}\vDash \psi)-P(\Bar{\Sigma}_{ss}\vDash \psi) \right|\\
    \leq& \epsilon+\epsilon_g, \text{with $\epsilon=k\delta B_L \mathfrak{L}$,}
\end{align*}
where $P(\Sigma_{ss}\vDash \psi)$ and $P(\hat{\Sigma}_{ss}\vDash \psi)$ are the probabilities that $\Sigma_{ss}$ and $\hat{\Sigma}_{ss}$ satisfy the specification $\psi$ under strategy $\varpi$, $P(\Bar{\Sigma}_{ss}\vDash \psi)=P^{k}_{up}(q),~ q\in Q$, $k$ is the number of steps, $n_Q$ is the number of entries of $Q$, $\delta$ is the state discretisation parameter, $B_L$ is the asymptotic upper bound of LC, and $\mathfrak{L}$ is the Lebesgue measure of the specification set.
\end{proof}

\section{SUPPLEMENTARY of SECTION~\ref{case_sec}}

\subsection{Supplementary of Example~\ref{verif_case1}}\label{supp_verif_case1}
The unknown linear stochastic system is 
\begin{align*}
    X(k+1)=AX(k)+W(k),
\end{align*}
where $
A=\begin{bmatrix}
0.4 & 0.1  \\
0 & 0.5 
\end{bmatrix}$. $W$ has Gaussian distribution with mean $\mu=\begin{bmatrix}
0  \\
0\end{bmatrix}$ and variance $\Sigma=\begin{bmatrix}
1 & 0  \\
0 & 1 
\end{bmatrix}$. The corresponding CoDF is 
\begin{equation*}
   f_{X_{k+1}|X_{k}}(X(k+1),X(k))=\frac{1}{ 2\pi\sqrt{|\Sigma|} }\exp{ ( -\frac{1}{2}( X(k+1)-AX(k) )^{T} \Sigma^{-1} ( X(k+1)-AX(k) )  ) }.
\end{equation*}

\subsection{Supplementary of Example~\ref{verif_case2}}\label{supp_verif_case2}
The unknown switched stochastic system with two actions $S_{\mathfrak a}=\{a_1,a_2 \}$ is 
\begin{align*}
    X(k+1)=f(X(k),a(k))=\begin{cases}
     A_1 X(k)+W(k)  ,~\text{if action is $a_1$},\\
    A_2 X(k)+W(k)  ,~\text{if action is $a_2$},
\end{cases}
\end{align*}
where $
A_1=\begin{bmatrix}
0.4 & 0.1  \\
0 & 0.5 
\end{bmatrix}$, $
A_2=\begin{bmatrix}
0.4 & 0.1  \\
-0.2 & 0.5 
\end{bmatrix}$. $W$ has Gaussian distribution with mean $\mu=\begin{bmatrix}
0  \\
0\end{bmatrix}$ and variance $\Sigma=\begin{bmatrix}
1 & 0  \\
0 & 1 
\end{bmatrix}$. The corresponding CoDF is 
\begin{align*}
   f_{X_{k+1}|X_{k}}(X(k+1),X(k))=\begin{cases}
        \text{if action is $a_1$},\\ 
        \frac{1}{ 2\pi\sqrt{|\Sigma|} }\exp{ ( -\frac{1}{2}( X(k+1)-A_1 X(k) )^{T} \Sigma^{-1} ( X(k+1)-A_1 X(k) )  ) }; \\
        \text{if action is $a_2$},\\
        \frac{1}{ 2\pi\sqrt{|\Sigma|} }\exp{ ( -\frac{1}{2}( X(k+1)-A_2 X(k) )^{T} \Sigma^{-1} ( X(k+1)-A_2 X(k) )  ) }.
    \end{cases}
\end{align*}

\begin{figure}[hbt!]
\centering
\includegraphics[width=0.9\textwidth]{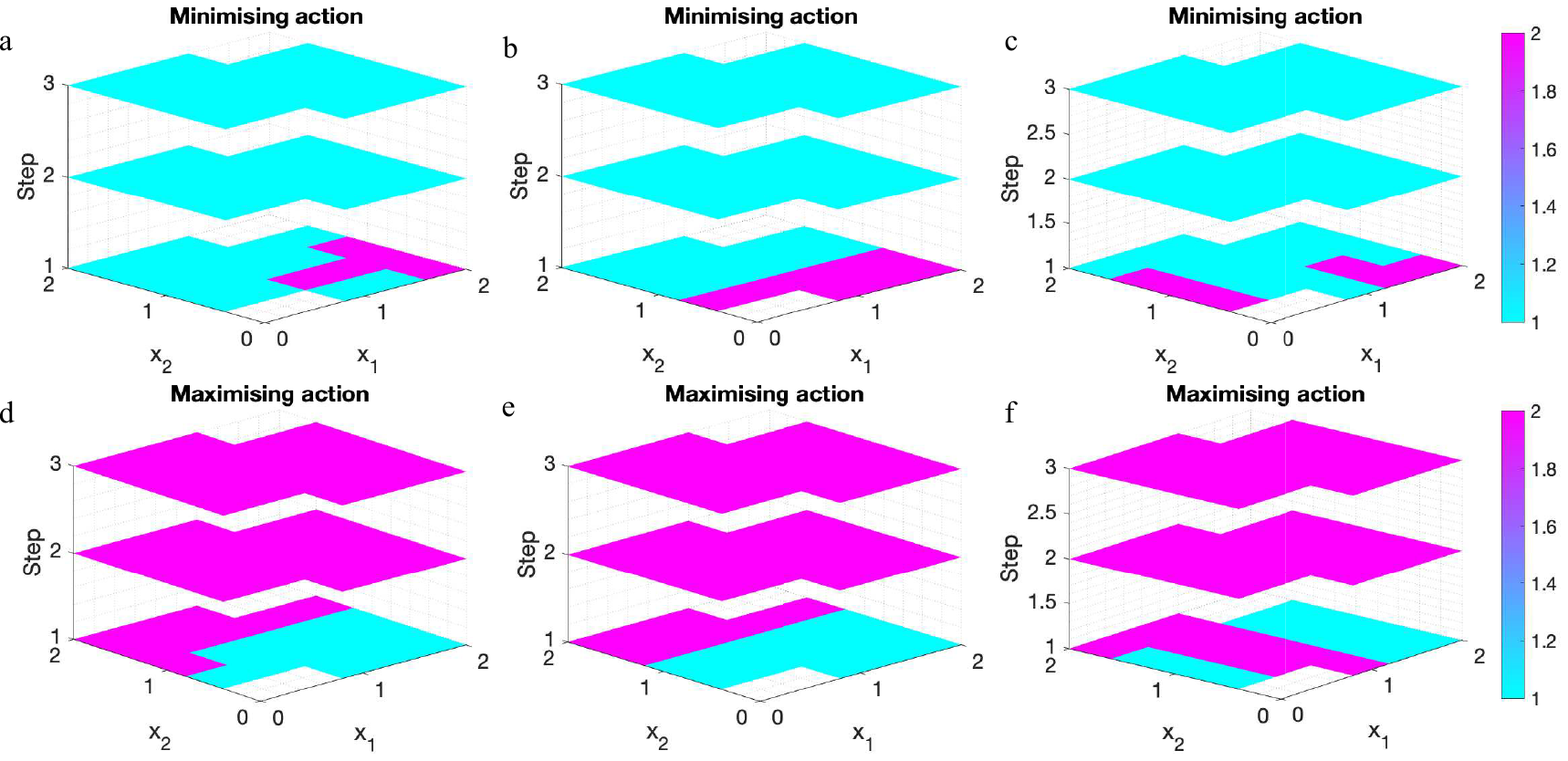}
\caption{The control policies as a function of state for minimising and maximising the satisfaction probability.
The action $a_1$ and $a_2$ are represented by cyan and magenta, respectively.
The panels (a) and (d) show the results from a model-based approach. The panels (b) and (e) show the results of the data-driven approximation using the empirical approach. The panels (c) and (f) are for the NPE under $\delta=0.4$.
}
\label{allaction04}
\end{figure}
\vspace{-0.3cm}
\vspace{-0.3cm}
\begin{figure}[hbt!]
\centering
\includegraphics[width=0.9\textwidth]{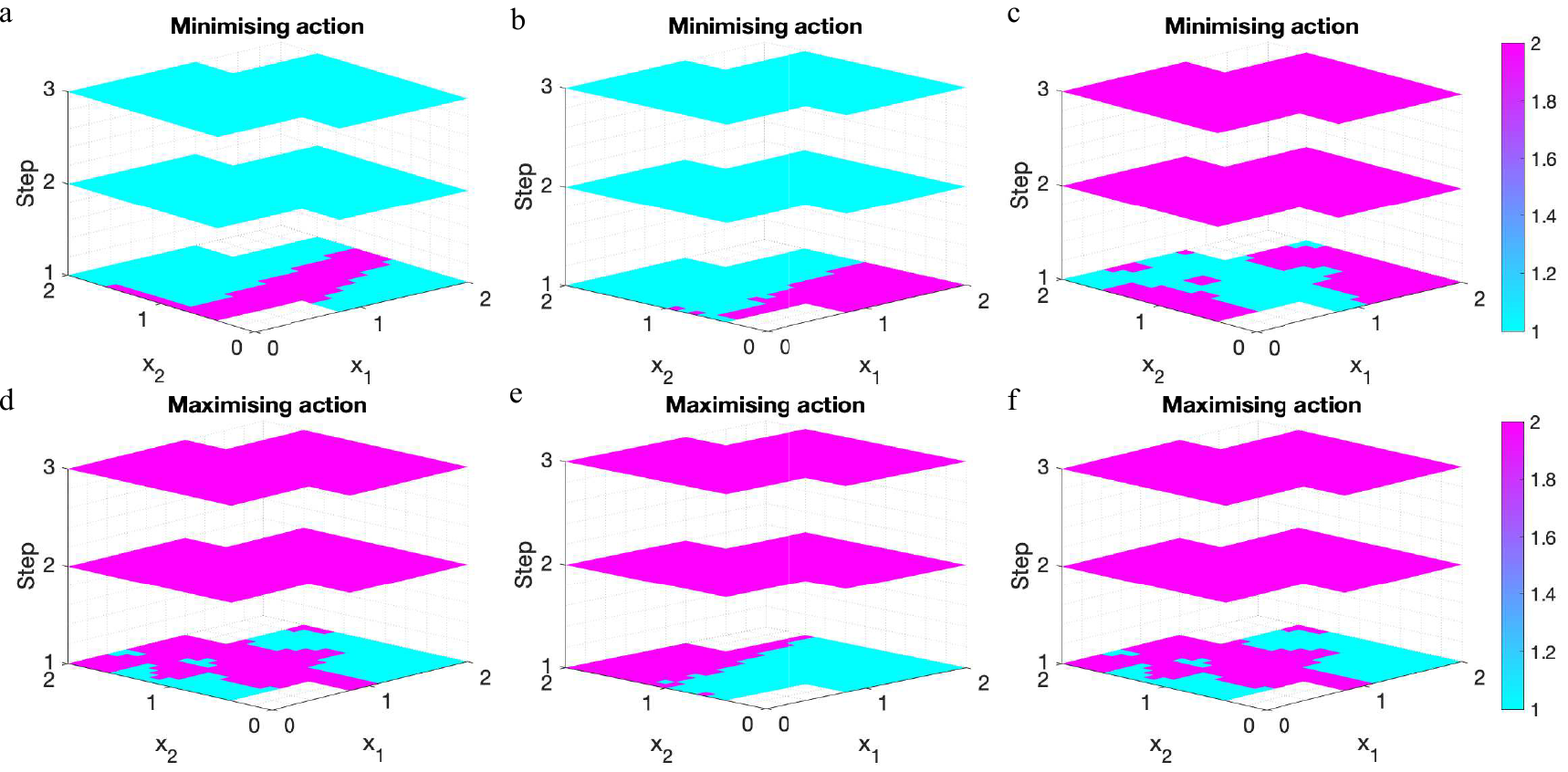}
\caption{
The control policies as a function of state for minimising and maximising the satisfaction probability.
The action $a_1$ and $a_2$ are represented by cyan and magenta, respectively.
The panels (a) and (d) show the results from a model-based approach. The panels (b) and (e) show the results of the data-driven approximation using the empirical approach. The panels (c) and (f) are for the NPE under $\delta=0.1$.
}
\label{allaction01}
\end{figure}

\end{document}